\documentclass{amsart}

\usepackage{a4wide}
\usepackage{hyperref}

\theoremstyle{plain}
\newtheorem{thm}{Theorem}
\newtheorem{prop}[thm]{Proposition}
\newtheorem{lemma}[thm]{Lemma}
\newtheorem{cor}[thm]{Corollary}
\theoremstyle{definition}
\newtheorem{definition}[thm]{Definition}
\newtheorem{remark}[thm]{Remark}
\newtheorem{example}[thm]{Example}
\newtheorem{remarks}[thm]{Remarks}

\newcommand{\ldr}[1]{\langle #1\rangle}
\newcommand{\tn}[1]{\ensuremath{\mathbb{T}^{#1}}}
\newcommand{\rn}[1]{\ensuremath{\mathbb{R}^{#1}}}
\newcommand{\cn}[1]{\ensuremath{\mathbb{C}^{#1}}}

\newcommand{\sn}[1]{\ensuremath{\mathbb{S}^{#1}}}

\newcommand{\zo}{\mathbb{Z}}
\newcommand{\ro}{\mathbb{R}}

\newcommand{\rod}{\mathrm{d}}

\newcommand{\roI}{\mathrm{I}}
\newcommand{\roII}{\mathrm{II}}

\newcommand{\roVIIz}{\mathrm{VII}_{0}}

\newcommand{\road}{\mathrm{ad}}
\newcommand{\rotr}{\mathrm{tr}}
\newcommand{\rofK}{\mathrm{fK}}
\newcommand{\rodiv}{\mathrm{div}}

\newcommand{\roRic}{\mathrm{Ric}}
\newcommand{\roRi}{\mathrm{Ri}}

\newcommand{\bd}{\bar{d}}

\newcommand{\rem}{\mathrm{rem}}

\newcommand{\grad}{\mathrm{grad}}

\newcommand{\supp}{\mathrm{supp}}

\newcommand{\rotot}{\mathrm{tot}}
\newcommand{\rocon}{\mathrm{con}}

\newcommand{\romn}{\mathrm{mn}}
\newcommand{\rosil}{\mathrm{sil}}
\newcommand{\roh}{\mathrm{h}}
\newcommand{\bk}{\bar{k}}
\newcommand{\be}{\bar{e}}
\newcommand{\bK}{\bar{K}}
\newcommand{\bR}{\bar{R}}

\newcommand{\chX}{\check{X}}

\newcommand{\bge}{\bar{g}}
\newcommand{\bS}{\bar{S}}
\newcommand{\bnabla}{\overline{\nabla}}
\newcommand{\bga}{\bar{\gamma}}

\renewcommand{\a}{\alpha}
\newcommand{\e}{\epsilon}

\newcommand{\de}{\delta}

\renewcommand{\b}{\beta}
\newcommand{\bbe}{\bar{\beta}}
\newcommand{\g}{\gamma}
\newcommand{\G}{\Gamma}
\renewcommand{\d}{\partial}
\newcommand{\me}{\mathcal{E}}
\newcommand{\mK}{\mathcal{K}}
\newcommand{\mG}{\mathcal{G}}
\newcommand{\mf}{\mathcal{F}}

\newcommand{\mP}{\mathcal{P}}
\newcommand{\mI}{\mathcal{I}}
\newcommand{\ml}{\mathcal{L}}
\newcommand{\mJ}{\mathcal{J}}

\newcommand{\ma}{\mathcal{A}}

\newcommand{\mff}{\mathfrak{f}}
\newcommand{\mfg}{\mathfrak{g}}

\newcommand{\mc}{\mathcal{C}}
\newcommand{\mb}{\mathcal{B}}
\newcommand{\mcX}{\mathcal{X}}

\newcommand{\hf}{\hat{f}}

\newcommand{\hg}{\hat{g}}
\newcommand{\chg}{\check{g}}

\newcommand{\hG}{\hat{\G}}

\newcommand{\cha}{\check{a}}

\newcommand{\hvph}{\hat{\varphi}}
\newcommand{\cvph}{\check{\varphi}}

\newcommand{\hb}{\hat{b}}
\newcommand{\ha}{\hat{a}}

\newcommand{\hna}{\hat{\nabla}}

\newcommand{\bx}{\bar{x}}

\newcommand{\bp}{\bar{p}}

\newcommand{\txi}{\tilde{\xi}}
\newcommand{\tSi}{\tilde{\Sigma}}
\newcommand{\tsi}{\tilde{\sigma}}
\newcommand{\ts}{\tilde{s}}
\newcommand{\te}{\tilde{e}}
\newcommand{\tA}{\tilde{A}}
\newcommand{\tN}{\tilde{N}}
\newcommand{\tin}{\tilde{n}}

\newcommand{\bsigma}{\bar{\sigma}}

\newcommand{\Sp}{\Sigma_{+}}
\newcommand{\Sm}{\Sigma_{-}}
\newcommand{\No}{N_{1}}
\newcommand{\Nt}{N_{2}}
\newcommand{\Nth}{N_{3}}

\begin{document}

\title{A unified approach to the Klein-Gordon equation on Bianchi backgrounds}
\author{Hans Ringstr\"{o}m}
\address{Department of Mathematics, KTH, 100 44 Stockholm, Sweden}

\begin{abstract}
In this paper, we study solutions to the Klein-Gordon equation on Bianchi backgrounds. In particular, we are interested
in the asymptotic behaviour of solutions in the direction of silent singularities. The main conclusion is that, for a given 
solution $u$ to the Klein-Gordon equation, 
there are smooth functions $u_{i}$, $i=0,1$, on the Lie group under consideration, such that $u_{\sigma}(\cdot,\sigma)-u_{1}$ 
and $u(\cdot,\sigma)-u_{1}\sigma-u_{0}$ asymptotically converge to zero in the direction of the singularity (where $\sigma$ 
is a geometrically defined time coordinate such that the singularity corresponds to $\sigma\rightarrow-\infty$). 
Here $u_{i}$, $i=0,1$, should be thought of as data on the singularity. Interestingly, it is possible to prove that the 
asymptotics are of this form for a large class of Bianchi spacetimes. Moreover, the conclusion applies for singularities
that are matter dominated; singularities that are vacuum dominated; and even when the asymptotics of the 
underlying Bianchi spacetime are oscillatory. To summarise, there seems to be a universality 
as far as the asymptotics in the direction of silent singularities are concerned. In fact, it is tempting to conjecture 
that as long as the singularity of the underlying Bianchi spacetime is silent, then the asymptotics of solutions are as 
described above. In order to contrast the above asymptotics with the non-silent setting, we, by appealing to known 
results, provide a complete asymptotic characterisation of solutions to the Klein-Gordon equation on a flat Kasner 
background. In that setting, $u_{\sigma}$ does, generically, not converge. 
\end{abstract}
\maketitle

\section{Introduction}

The subject of this paper is the asymptotic behaviour of solutions to the Klein-Gordon equation
\begin{equation}\label{eq:KG}
\Box_{g}u-m^{2}u=0
\end{equation}
on a spacetime $(M,g)$, where $m$ is a constant. In practice, we are going to consider the somewhat more general equation
\begin{equation}\label{eq:KGvarphiz}
\Box_{g}u+\varphi_{0}u=0,
\end{equation}
where $\varphi_{0}$ is a function. We impose specific restrictions on $\varphi_{0}$ as we go along. We are not going to consider 
arbitrary $(M,g)$, but rather focus on Bianchi spacetimes; i.e., Lorentz manifolds of the following type. 
\begin{definition}\label{def:Bianchispacetime}
A \textit{Bianchi spacetime} is a Lorentz manifold $(M,g)$, where $M=G\times I$; $I=(t_{-},t_{+})$ is an open interval; $G$ is a 
connected $3$-dimensional Lie group; and $g$ is of the form 
\begin{equation}\label{eq:Bianchimetricdef}
g=-dt\otimes dt+a_{ij}(t)\xi^{i}\otimes \xi^{j},
\end{equation}
where $\{\xi^{i}\}$ is the dual basis of a basis $\{e_{i}\}$ of the Lie algebra $\mfg$ and $a_{ij}\in C^{\infty}(I,\ro)$ are such that 
$a_{ij}(t)$ are the components of a positive definite matrix $a(t)$ for every $t\in I$. 
\end{definition}
We are interested in the asymptotic behaviour of solutions as $t\rightarrow t_{-}$, where we assume, at the very minimum, $t_{-}$
to represent a singularity in the following sense. 
\begin{definition}\label{def:monotonevolumesingularity}
Let $(M,g)$ be a Bianchi spacetime. Then $t_{-}$ is said to be a \textit{monotone volume singularity} if there is a $t_{0}\in I$ such that the 
mean curvature, say $\theta(t)$, of $G_{t}:=G\times\{t\}$ is strictly positive on $(t_{-},t_{0})$ and if 
\begin{equation}\label{eq:taudefinition}
\tau:=\frac{1}{3}\ln\sqrt{\det a}\rightarrow-\infty
\end{equation}
as $t\rightarrow t_{-}$, where $a$ is the matrix with components $a_{ij}$. 
\end{definition}
\begin{remark}\label{remark:tauzerogreaterthanzero}
In what follows, we always assume, without loss of generality, that a monotone volume singularity is such that $
\tau(t_{0})>0$; note that this can be ensured by rescaling the frame $\{e_{i}\}$, if necessary. 
\end{remark}
\subsection{Time coordinates} Let $(M,g)$ be a Bianchi spacetime. If $t_{-}$ is a monotone volume singularity, then $\tau$ can be used as
a time coordinate in a neighbourhood of the singularity; note that $\d_{t}\tau=\theta/3$. From now on, we refer to $\tau$ as the 
\textit{logarithmic volume density}. In some situations it is more convenient to use a time coordinate $\sigma$ satisfying
\begin{equation}\label{eq:dsigmadtdefrel}
\frac{d\sigma}{dt}=\frac{1}{3}(\det a)^{-1/2}.
\end{equation}
In what follows, we define $\sigma$ by, in addition, requiring $\sigma=0$ to correspond to $\tau=0$. For the monotone volume singularities 
of interest here, $\sigma\rightarrow-\infty$ corresponds to $\tau\rightarrow-\infty$; cf. Subsection~\ref{ssection:addchangetimecoord} below 
for a more detailed discussion.

\subsection{Silence} One aspect which is of central importance in the analysis of the asymptotics is whether observers asymptotically lose the 
ability to communicate in the direction of the singularity or not. To be more specific, let $(M,g)$ be a Bianchi spacetime and $t_{-}$ be a
monotone volume singularity. Then $t_{-}$ is said to be a \textit{silent monotone volume singularity} if 
\begin{equation}\label{eq:normaraisedtominusonehalf}
\|a^{-1/2}\|\in L^{1}(t_{-},t_{0}],
\end{equation}
where $a^{-1/2}$ denotes the inverse of the square root of the matrix $a$. In the present paper, we are mainly interested in silent monotone 
volume singularities. However, as an illustration, we also describe the asymptotics for one Bianchi spacetime with a Cauchy horizon;
cf. Section~\ref{section:nonsilentexample} below. 

\subsection{Main results, rough description}
Consider the asymptotic behaviour of solutions to the Klein-Gordon equation in the direction of a silent monotone volume singularity. 
Assume the underlying Bianchi geometry to satisfy Einstein's equations; cf. Sections~\ref{section:mainenergyestimate} 
and \ref{section:geompropBianchi} below for details. Then there are two main conclusions. First, there are general conditions on the stress
energy tensor and the geometry ensuring that $u_{\sigma}$ is bounded in the direction of the singularity, where $u$ is a solution to the 
Klein-Gordon equation; cf. Subsection~\ref{ssection:main} below. Second, there are convergence results; cf.
Sections~\ref{section:applications}--\ref{section:convresIII} below. In each of these sections, we formulate general results. 
However, we also provide examples, illustrating that the results apply in the orthogonal perfect fluid setting. The typical 
conclusion is that, for a solution $u$ to the Klein-Gordon equation on the relevant background, there are functions $u_{0}$ and $u_{1}$ such that 
\begin{equation}\label{eq:usigmauasymptoticsintro}
u_{\sigma}(\cdot,\sigma)-u_{1},\ \ \
u(\cdot,\sigma)-u_{1}\sigma-u_{0}
\end{equation}
converge to zero asymptotically; cf. Sections~\ref{section:applications}--\ref{section:convresIII} below for details. In particular, it is of 
interest to note that these asymptotics seem to be universal; i.e., independent of the background, as soon as the monotone volume singularity 
under consideration
is silent. Combining the results of this paper with unpublished results by Bernhard Brehm yields the conclusion that the asymptotics are such
even in the case of generic oscillatory vacuum Bianchi type VIII and IX solutions; cf. Examples~\ref{example:usigmaconvgeneralI}
and \ref{example:fullasymptotics} below. In order to contrast the above asymptotics with the asymptotics of solutions to the Klein-Gordon 
equation in the case of a non-silent monotone volume singularity, we consider the Klein-Gordon equation on a flat Kasner background 
in Section~\ref{section:nonsilentexample} below. 

\subsection{Outline}
The outline of the paper is as follows. In Section~\ref{section:mainenergyestimate} below, we describe the general assumptions we make concerning the 
background Bianchi spacetimes; describe the notation we use; recall the Raychaudhuri equation (the coefficients of which play a central role in the 
arguments); conformally rescale the metric; formulate the equation with respect to the conformally rescaled metric; discuss the Bianchi classification; 
and define the basic energies. We end the section by stating the main energy estimate; cf. Subsection~\ref{ssection:main}. 
This estimate implies a bound on $u_{\sigma}$. Before turning to the question of convergence, we devote Section~\ref{section:geompropBianchi}
to a discussion of the developments corresponding to Bianchi orthogonal perfect fluid initial data. In particular, we appeal to results in the 
literature in order to demonstrate that the developments have monotone volume singularities. We also discuss the silence of the monotone
volume singularities and the maximality of the developments. Given this background material, we turn to the question of convergence in 
Sections~\ref{section:applications}--\ref{section:convresIII}. In Subsection~\ref{ssection:leadingorderasymptotics}, we state a result ensuring 
convergence of $u_{\sigma}$; cf. Proposition~\ref{prop:osc} below. As an application, we demonstrate, e.g., that on generic Bianchi class A vacuum 
backgrounds, solutions to the Klein-Gordon equation are such that $u_{\sigma}$ converges; cf. Example~\ref{example:usigmaconvgeneralI}. 
In Subsection~\ref{ssection:completeasymptotics}, we provide conditions ensuring that the asymptotics are as described in 
(\ref{eq:usigmauasymptoticsintro}). We also note that on generic Bianchi class A vacuum backgrounds, the conditions are satisfied; cf. 
Example~\ref{example:fullasymptotics}. It is important to note that, due to unpublished results of Bernhard Brehm (building on the arguments of 
\cite{brehm}), the results of Section~\ref{section:applications} apply even in the generic oscillatory Bianchi type VIII and IX vacuum settings. 
In Section~\ref{section:convresII}, we turn to the convergent setting; here convergent means, in particular, that 
a rescaled version of the second fundamental form converges and that $\dot{\theta}/\theta^{2}$ converges. The main result is
Proposition~\ref{prop:asymptoticsexponconvofqtotwo}. This result is based on stronger assumptions than those of Section~\ref{section:applications}, 
but it also gives stronger conclusions. Moreover, it applies in several different settings. In Example~\ref{example:asymptoticsstifffluid}, we demonstrate
that it applies to stiff fluids for all Bianchi types except VI$_{-1/9}$ (in the case of Bianchi type VI$_{-1/9}$, we are unaware of any appropriate results
in the literature). In Examples~\ref{example:nonoscBclassAdev}, \ref{example:genericBianchiclassA} and \ref{example:nonexcBianchiclassB} we give
further applications of Proposition~\ref{prop:asymptoticsexponconvofqtotwo}. Finally, in Section~\ref{section:convresIII}, we treat an 
exceptional case, not covered by Proposition~\ref{prop:asymptoticsexponconvofqtotwo}.

The results of Sections~\ref{section:applications}--\ref{section:convresIII} concern silent monotone volume singularities. 
In order to contrast the behaviour in the silent setting with the behaviour in the 
presence of a horizon, we consider the Klein-Gordon equation on a flat Kasner background in Section~\ref{section:nonsilentexample}. In 
particular, we demonstrate that $u_{\sigma}$ does not converge in that case. It is also of interest to consider the question of blow up of the
solution. In the case of non-flat Kasner backgrounds, as well as certain isotropic Bianchi type I backgrounds, this is done in \cite{afaf}. In 
Section~\ref{section:blowupintro}, we derive conclusions concerning blow up using the results of \cite{finallinsys} (for non-flat Kasner
backgrounds). The reason for doing so is that the results of \cite{finallinsys} yield interesting information concerning the regularity of 
the function $u_{1}$ appearing in (\ref{eq:usigmauasymptoticsintro}). Moreover, the corresponding observations naturally lead to an important 
open problem. 

In Section~\ref{section:geometricbackground}, we provide the 
geometric background material on which the arguments of the paper is based. In particular, we derive the Raychaudhuri equation in the 
setting of interest here; obtain geometric inequalities on which the energy estimates are based; demonstrate a divergence type result that
we use in the proof of the energy estimates; and derive results concerning the causal structure in the silent setting. The results concerning
the causal structure are needed in order to justify that it is sufficient to consider solutions to the Klein-Gordon equation with compactly
supported initial data in the study of the singularity. In Sections~\ref{section:conformalrescaling} and \ref{section:KGonBianchi}, we 
discuss the conformal rescaling of the metric; the corresponding reformulation of the equation; and the relation between the time coordinates
$\tau$ and $\sigma$. In Sections~\ref{section:basicenergy} and \ref{section:higherorderenergies}, we then derive energy estimates. We also 
prove the propositions appearing in Sections~\ref{section:convresII} and \ref{section:convresIII}. 
In Section~\ref{section:proofsI} we write down the proofs of the propositions and examples in Section~\ref{section:applications}; 
in Section~\ref{section:proofsII} we write down the proofs of the examples in Section~\ref{section:convresII}; and
in Section~\ref{section:proofsIII} we write down the proofs of the examples in Section~\ref{section:convresIII}.
Finally, in the appendix, we provide the background material concerning the Bianchi developments needed in the examples in
Sections~\ref{section:applications}-\ref{section:convresIII}. We also prove the statements made in Section~\ref{section:blowupintro}.

\subsection{Previous results}\label{ssection:previousresults}

Clearly, there are several papers with related and partially overlapping results; cf., e.g., \cite{aren,pet,finallinsys,afaf,bac} and 
references cited therein. In \cite{aren}, Alan Rendall and Paul Allen study perturbations of FLRW cosmological models which are spatially 
flat and have $\tn{3}$ spatial topology. In particular, the authors consider a linear hyperbolic equation on isotropic Bianchi type I backgrounds. 
In \cite{pet}, Oliver Lindblad Petersen considers the linear wave equation on Kasner backgrounds, including the flat Kasner solution. In
\cite{finallinsys}, we consider systems of linear wave equations on a class of background geometries (for a detailed description of the 
requirements, cf. \cite{finallinsys}). The main results concern optimal energy estimates, but we also derive asymptotic information, in some
cases a homeomorphism between initial data and asymptotic data. In \cite{afaf} Artur Alho, Grigorios Fournodavlos and Anne Franzen consider
the wave equation on isotropic Bianchi type I backgrounds and on non-flat Kasner backgrounds. Moreover, they provide an open criterion for 
$L^{2}$-blow up of the solution. In \cite{bac}, Alain Bachelot considers Klein-Gordon type equations on FLRW backgrounds. In practice, he 
considers warped product type geometries (in this sense, the results are more general than FLRW). Moreover, he considers Big Bang, Big Crunch,
Big Rip, Big Brake and Sudden Singularities. 

\section{Main energy estimate}\label{section:mainenergyestimate}

In the previous section, we formulated the questions of interest for general Bianchi spacetimes. However, we are here, 
in practice, interested in solutions to Einstein's equations. In other words, we assume that 
\begin{equation}\label{eq:Einsteinsequations}
\roRic-\frac{1}{2}Sg+\Lambda g=T,
\end{equation}
where $\roRic$ and $S$ are the Ricci tensor and scalar curvature of $g$ respectively; $\Lambda$ is the cosmological constant; and $T$ is 
the stress energy tensor. In some of the results, we only impose general conditions on the stress energy tensor. However, we also illustrate 
that the general conditions are satisfied for large classes of Bianchi spacetimes with matter of orthogonal perfect fluid type. 

Next, we introduce the terminology in terms of which we formulate the assumptions. To begin with, we define the \textit{energy density} and 
the \textit{mean pressure} by
\begin{equation}\label{eq:rhobpdef}
\rho:=T_{00},\ \ \
\bp:=\frac{1}{3}a^{ij}T_{ij}
\end{equation}
respectively, where the components are calculated with respect to the frame $\{e_{\a}\}$; here $e_{0}:=\d_{t}$ and the $e_{i}$ are given 
in Definition~\ref{def:Bianchispacetime}. Moreover, $a^{ij}$ are the components of the matrix $a^{-1}$. From now on, sub- and superscripts 
refer to the frame $\{e_{\a}\}$, Greek indices range from $0$ to $3$ and Latin indices range from $1$ to $3$. We denote the induced metric 
and second fundamental form on $G_{t}$ by $\bge$ and $\bk$ respectively, and we think of these objects as being defined on $G$. Introduce
\begin{equation}\label{eq:thetaandsigmaijdef}
\theta:=\rotr_{\bge}\bk,\ \ \
\bsigma_{ij}:=\bk_{ij}-\frac{1}{3}\theta\bge_{ij}. 
\end{equation}
These objects are referred to as the \textit{mean curvature} and the \textit{shear} respectively. In analogy with \cite{waihsu89}, 
it is also convenient to introduce the rescaled quantities
\begin{equation}\label{eq:rescaledSigmaij}
\Sigma_{ij}:=\frac{\bsigma_{ij}}{\theta},\ \ 
\Omega_{\rho}:=\frac{3\rho}{\theta^{2}},\ \ 
\Omega_{\bp}:=\frac{3\bp}{\theta^{2}},\ \ 
\Omega_{\Lambda}:=\frac{3\Lambda}{\theta^{2}},\ \ 
\Omega_{\rotot}=\Omega_{\rho}+3\Omega_{\bp}-2\Omega_{\Lambda}.
\end{equation}
With the above notation, Raychaudhuri's equation can be written
\begin{equation}\label{eq:Raychaudhurirelpropertimeintro}
-3\frac{\dot{\theta}}{\theta^{2}}=1+q,
\end{equation}
where 
\begin{equation}\label{eq:qdefinition}
q:=3\Sigma^{ij}\Sigma_{ij}+\frac{1}{2}\Omega_{\rotot};
\end{equation}
cf. Subsection~\ref{ssection:raychaudhuri} below for a justification of this statement. We refer to $q$ as the \textit{deceleration
parameter} in analogy with the terminology introduced in \cite[p.~1414]{waihsu89}.

\subsection{Conformal rescaling} 
In practice, it is convenient to rescale the metric according to $\hg:=\theta^{2}g/9$; cf. Section~\ref{section:conformalrescaling}
below for details. In the case of a monotone volume singularity, $\hg$ can then be written 
\[
\hg=-d\tau\otimes d\tau+\ha_{ij}(\tau)\xi^{i}\otimes \xi^{i}.
\]
Here $\tau\in \mI_{0}$, where $\mI_{0}:=(-\infty,\tau_{0})$ is the image of $(t_{-},t_{0})$ under $\tau$. Moreover, 
$\ha_{ij}=\theta^{2}a_{ij}/9$. Using this notation, (\ref{eq:KGvarphiz}) can be rewritten
\begin{equation}\label{eq:KleinGordonConformallyRescaledintro}
-u_{\tau\tau}+\ha^{ij}e_{i}[e_{j}(u)]+(q-2)u_{\tau}-2X_{0}(u)+\hvph_{0}u=0;
\end{equation}
cf. Section~\ref{section:KGonBianchi} below for a justification of this statement. 
Here $q$ is defined by (\ref{eq:qdefinition}); $\hvph_{0}:=9\theta^{-2}\varphi_{0}$; and $\ha^{ij}$ are the components of 
the inverse of the matrix with components $\ha_{ij}$. In order to explain the origin of the vector field $X_{0}$, let
$\xi_{G}:\mfg\rightarrow\ro$ be the one form defined by 
\begin{equation}\label{eq:xiGdef}
\xi_{G}(X):=\frac{1}{2}\rotr\,\road_{X},
\end{equation}
where the linear map $\road_{X}:\mfg\rightarrow\mfg$ is defined by $\road_{X}Y=[X,Y]$ for all $X,Y\in\mfg$. Then $X_{0}$ is the metrically related vector 
field. In other words, if $X_{0}=X_{0}^{i}e_{i}$ and $\xi_{G,i}=\xi_{G}(e_{i})$, then $X_{0}^{i}=\ha^{ij}\xi_{G,j}$. 

\subsection{The Bianchi classification} Some of the results of the present paper are divided according to the Bianchi 
classification. Next, we therefore explain the associated terminology. Three-dimensional Lie groups can be divided according to whether 
they are unimodular or non-unimodular. For completeness, we recall one characterisation of unimodularity. 

\begin{definition}\label{definition:bianchiclasses}
A connected Lie group $G$ is said to be \textit{unimodular} if $\xi_{G}=0$, where $\xi_{G}:\mfg\rightarrow\ro$ is defined by (\ref{eq:xiGdef}).
A connected Lie group which is not unimodular is said to be \textit{non-unimodular}.
\end{definition}
This definition leads to the division into Bianchi class A and Bianchi class B. 

\begin{definition}
Let $(M,g)$ be a Bianchi spacetime. Then, if the associated Lie group is unimodular, the spacetime is said to be of \textit{Bianchi class A}, 
and if the associated Lie group is non-unimodular, the spacetime is said to be of \textit{Bianchi class B}.  
\end{definition}

In what follows, we also speak of Bianchi class A and B Lie groups. Beyond the division of Bianchi spacetimes and $3$-dimensional Lie groups 
into classes, there is, for each class, a division into types. This division can be found
in many references, but in the case of Bianchi class A, we refer the reader to \cite[Table~1, p. 409]{BianchiIXattr} (since we appeal to the
results of \cite{BianchiIXattr} quite frequently in the present paper). In the case of Bianchi class B, the division into types can be found
in, e.g., \cite[Table~1, p.~16]{RadermacherNonStiff} (again, we appeal to the results of \cite{RadermacherNonStiff} quite frequently in the present 
paper). In the case of Bianchi class B, there is one type which is considered ``exceptional'', namely Bianchi type VI${}_{-1/9}$. The reason for 
this is that if one considers (given a left invariant Riemannian initial metric) the momentum constraint as a linear equation for a left 
invariant second fundamental form, then this equation has a degeneracy which occurs exactly for Bianchi type VI${}_{-1/9}$; cf. 
\cite[Subsection~11.3, pp.~62--64]{RadermacherNonStiff}, in particular \cite[Lemma~11.13, p.~63]{RadermacherNonStiff}, for more details. This
means that for Bianchi class B, the degrees of freedom are maximised for Bianchi type VI${}_{-1/9}$. In the applications of the general results, 
we exclude the exceptional case due to the lack of references concerning the asymptotics of the corresponding Bianchi spacetimes. Finally, let us 
refer the reader interested in a historical overview of the origins of the Bianchi classification to \cite{krasetal}. 

\subsection{Orthogonal perfect fluids with a linear equation of state}\label{ssection:orthogonalperfectfluids}
A natural matter model to consider is an orthogonal perfect fluid with a linear equation of state. In practice, this means that the stress 
energy tensor is of the form 
\begin{equation}\label{eq:Torthogfluid}
T=(\rho+p)dt\otimes dt+pg,
\end{equation}
where $p=(\g-1)\rho$ and $\g$ is a constant. Here $\rho$ is the \textit{energy density}, $p$ is the \textit{pressure}, and $\rho$
and $p$ only depend on $t$. Note that $\bp$ introduced in (\ref{eq:rhobpdef}) then equals $p$. We restrict our attention to $0\leq \g\leq 2$. 
Moreover, we impose additional restrictions on $\g$ in particular cases. The special case that $\g=2$ is referred to as a \textit{stiff
fluid}. When we consider an orthogonal perfect fluid with a linear equation of state, we typically also assume that $\Lambda=0$.

\subsection{Energies}\label{ssection:energiesintro}
In order to have a flexible notion of an energy, it is convenient to introduce a function $\mff$ with the following properties: 
$\mff\in L^{1}(-\infty,0]$; $\mff>0$; $\mff\in C^{1}$;  $\mff'/\mff$ is uniformly bounded on $(-\infty,0]$; and $\mff'\geq 0$. Note 
that the requirements on $\mff$ imply that $\mff\rightarrow 0$ as $\tau\rightarrow-\infty$. Given, e.g., a solution $u$ to 
(\ref{eq:KleinGordonConformallyRescaledintro}) corresponding to initial data at $\tau=0$ that are compactly supported on $G$, define
\begin{equation}\label{eq:meepsilonzerodefinition}
\me[u]:=\frac{1}{2}\int_{G}\left[u_{\tau}^{2}+\ha^{ij}e_{i}(u)e_{j}(u)+\mff^{2}u^{2}\right]\mu_{h}.
\end{equation}
Here
\begin{equation}\label{eq:hRieMetdef}
h:=\de_{ij}\xi^{i}\otimes \xi^{j}.
\end{equation}
Note that, due to Corollary~\ref{cor:basicexistenceandcompactsupp} below, $\me[u]$ is well defined, smooth, and we are allowed to 
differentiate under the integral sign. In what follows, we tacitly consider the constituents of $\me[u]$, as well as $\me[u]$ itself,
as depending on $\tau$ as opposed to $t$. In order to define higher order energies, it is convenient to introduce the following 
terminology. 
\begin{definition}\label{definition:vectorfieldmultiindex}
A \textit{vector field multiindex} is an ordered set of pairs of integers
\begin{equation}\label{eq:vectorfieldmultiindex}
K=\{(i_{1},j_{1}),\dots,(i_{k},j_{k})\},
\end{equation}
where $i_{l}\in\{1,2,3\}$, $l=1,\dots,k$; and $0\leq j_{l}\in\zo$, $l=1,\dots,k$. Moreover, $j_{l}$ is only allowed to be zero
if $k=1$, and the vector field multiindices $\{(1,0)\}$, $\{(2,0)\}$ and $\{(3,0)\}$ should be thought of as being the same 
vector field multiindex, denoted by $0$. For the vector field multiindex $K$ given by (\ref{eq:vectorfieldmultiindex}), 
\[
e_{K}:=e_{i_{1}}^{j_{1}}\cdots e_{i_{k}}^{j_{k}},\ \ \
|K|:=j_{1}+\dots+j_{k}. 
\]
Here $|K|$ is referred to as the \textit{order} of $K$. Moreover, if $K=0$, then $e_{K}(\phi)=\phi$. 
\end{definition}
With $h$ and $\mff$ as above, define
\begin{equation}\label{eq:meepsilonldefinition}
\me_{l}[u]:=\frac{1}{2}\int_{G}\textstyle{\sum}_{|K|\leq l}\left([e_{K}(u)]_{\tau}^{2}+\ha^{ij}e_{i}[e_{K}(u)]e_{j}[e_{K}(u)]
+\mff^{2}[e_{K}(u)]^{2}\right)\mu_{h}.
\end{equation}

\subsection{The main energy estimate}\label{ssection:main}

Given the above terminology, we are in a position to formulate the main energy estimate.

\begin{thm}\label{thm:main}
Let $(M,g)$ be a Bianchi spacetime with a monotone volume singularity $t_{-}$. Assume that $g$ solves (\ref{eq:Einsteinsequations}); 
that $\rho\geq\bp$; that $\rho\geq 0$; and that $\Lambda\geq 0$. Let $\varphi_{0}\in C^{\infty}(M)$ be such that it only depends on 
$t$. Assume that $\g_{\bS}\in L^{1}(-\infty,0]$, where
\begin{equation}\label{eq:gammabSdef}
\g_{\bS}:=\frac{\bS_{+}}{\theta^{2}},
\end{equation}
$\bS_{+}:=\max\{0,\bS\}$ and $\bS(\tau)$ is the scalar curvature of the spatial hypersurface corresponding to $\tau$. Assume, moreover, 
that 
\begin{equation}\label{eq:hvphzhasharpcond}
|\hvph_{0}|+\|\ha^{-1}\|\leq \mff^{2}
\end{equation}
for all $\tau\leq 0$, where $\ha^{-1}$ is the matrix with components $\ha^{ij}$ and $\mff$ is a function with the properties 
stated at the beginning of Subsection~\ref{ssection:energiesintro}. Then there is a constant $C_{l}$ such that for every 
smooth solution $u$ to (\ref{eq:KleinGordonConformallyRescaledintro}) corresponding to compactly supported initial data,
\begin{equation}\label{eq:energyestimatehigherorderenergiesthm}
\me_{l}[u](\tau)\leq C_{l}\me_{l}[u](0)\exp\left[2\int_{\tau}^{0}[2-q(\tau')]d\tau'\right]
\end{equation}
for all $\tau\leq 0$. 
\end{thm}
\begin{remark}\label{remark:sigmainftauinfcorr}
The theorem follows from Corollary~\ref{cor:energyestimatehigherorderenergies} below. Note also that under the assumptions of 
the theorem, $\tau\rightarrow-\infty$ corresponds to $\sigma\rightarrow-\infty$; cf. Lemma~\ref{lemma:sigmataulimcorr} below. 
\end{remark}

Concerning the assumptions, let us remark the following. First, if $T$ satisfies the dominant energy condition, then $\rho\geq\bp$ 
and $\rho\geq 0$. For all the Bianchi types except IX, $\bS\leq 0$, so that $\g_{\bS}=0$; cf., e.g., \cite[Appendix~E]{stab}. The 
condition on $\g_{\bS}$ is thus only 
non-trivial in the case of Bianchi type IX. Moreover, in the case of Bianchi type IX orthogonal perfect fluids with $2/3<\g\leq 2$, 
it can be demonstrated that $\g_{\bS}$ converges to zero exponentially with respect to $\tau$-time (though we do not provide a proof of 
this statement here). In that sense, the condition on $\g_{\bS}$ is not even a restriction for a large class of Bianchi type IX solutions. 
Next, we are mainly interested in singularities 
for which $\theta\rightarrow\infty$. In that setting, $\Omega_{\Lambda}$ converges to zero; cf. (\ref{eq:rescaledSigmaij}). 
If, in addition, either the parameter $\Omega_{\bp}$ converges to zero asymptotically, or the average pressure 
$\bp$ is non-negative, then $q$ is $\geq 0$ asymptotically; cf. (\ref{eq:rescaledSigmaij}) and (\ref{eq:qdefinition}). 
Combining this observation with the fact that $\d_{\tau}\theta=-(1+q)\theta$, cf. (\ref{eq:Raychaudhurirelpropertimeintro}), 
it follows that $\theta$ tends to infinity exponentially. In particular, in the case of the Klein-Gordon equation, $\hvph_{0}$ 
converges to zero exponentially. In that sense, the condition (\ref{eq:hvphzhasharpcond}) is not very restrictive as far as 
$\hvph_{0}$ is concerned. Finally, note that if $\tau(t_{a})=0$, then 
\begin{equation}\label{eq:intainvhainv}
\int_{t_{-}}^{t_{a}}\|a^{-1/2}(t)\|dt=\int_{-\infty}^{0}\|\ha^{-1}(\tau)\|^{1/2}d\tau.
\end{equation}
In other words, the requirement that $\|\ha^{-1}\|^{1/2}$ be integrable is equivalent to the requirement of silence. To demand
that there be a function $\mff$ such that (\ref{eq:hvphzhasharpcond}) holds and such that $\mff$ is $L^{1}$ is thus natural. At the 
beginning of Subsection~\ref{ssection:energiesintro}, we require $\mff$ to satisfy a few additional technical conditions. 
However, the main requirement is integrability of $\mff$; i.e., silence. Roughly speaking, we thus expect that if the underlying 
Bianchi geometry is silent, then an estimate of the form (\ref{eq:energyestimatehigherorderenergiesthm}) should hold. 

Next, let us point out that we only demand that the initial data for $u$ have compact support in order for the energies $\me_{l}$
to be well defined. Moreover, due to the causal structure in the silent setting, analysing the asymptotics in the general case 
can be reduced to analysing the asymptotics of solutions corresponding to initial data with compact 
support. Even for non-compact Lie groups $G$ and smooth solutions $u$ to the Klein-Gordon equation (not necessarily corresponding to 
initial data with compact support), we are, in fact, under quite general circumstances able to conclude the existence of functions 
$u_{1}$ and $u_{0}$ such that the expressions appearing in (\ref{eq:usigmauasymptoticsintro}) tend to zero asymptotically; cf. 
Sections~\ref{section:applications}-\ref{section:convresIII} below.

Concerning the asymptotics, it can be verified that (\ref{eq:energyestimatehigherorderenergiesthm}) implies that 
\begin{equation}\label{eq:eKusigmaLtwoestimateintro}
\int_{G}\textstyle{\sum}_{|K|\leq l}[e_{K}(u_{\sigma})]^{2}\mu_{h}\leq C_{l}\me_{l}[u](0),
\end{equation}
cf. Subsection~\ref{ssection:addchangetimecoord},
where we use the time coordinate $\sigma$ introduced in (\ref{eq:dsigmadtdefrel}). Combining this observation with Sobolev embedding
and the character of the causal structure in the silent setting yields the conclusion that $u_{\sigma}(\cdot,\sigma)$ is asymptotically
bounded in any local (in space) $C^{k}$ norm; cf. the arguments presented in the proof of Proposition~\ref{prop:asymptoticsexponconvofqtotwo} 
below for a justification of this statement. Integrating this estimate, it is clear that $u(\cdot,\sigma)$ does not grow faster than 
linearly in any local $C^{k}$ norm.

\section{Geometric properties of Bianchi developments}\label{section:geompropBianchi}

In the applications, we wish to deduce, rather than assume, that the Bianchi spacetimes have a monotone volume singularity. In the present 
section, we therefore interrupt the presentation of energy estimates in order to consider the implications of Einstein's equations. 
We use initial data as the starting point for our discussion and focus on the orthogonal perfect fluid setting.

\begin{definition}\label{def:Bianchiid}
\textit{Bianchi orthogonal perfect fluid initial data} for Einstein's equations consist of the following: a connected $3$-dimensional Lie group $G$; 
a left invariant metric $\bge$ on $G$; a left invariant symmetric covariant $2$-tensor field $\bk$ on $G$; and a constant $\rho_{0}\geq 0$ 
satisfying
\begin{align*}
\bS-\bk^{ij}\bk_{ij}+(\rotr_{\bge}\bk)^{2} = & 2\rho_{0}\\
\bnabla_{i}\rotr_{\bge}\bk-\bnabla^{j}\bk_{ij} = & 0.
\end{align*}
\end{definition}
\begin{remark}
Here $\bS$ and $\bnabla$ denote the scalar curvature and Levi-Civita connection of $\bge$ respectively. Moreover, indices are raised
and lowered with $\bge$. 
\end{remark}
\begin{remark}
When we speak of Bianchi orthogonal perfect fluid initial data we take it for granted that $\Lambda=0$. 
\end{remark}
Next, we consider the different Bianchi classes separately. 

\subsection{Bianchi class A}\label{ssection:BianchiclassA}

Given Bianchi orthogonal perfect fluid initial data such that $G$ is a unimodular Lie group, a corresponding development is introduced in 
\cite[Definition~21.1, p. 489]{BianchiIXattr} (the value of $\g$ should here be understood from the context and the equation of state
is given by $p=(\g-1)\rho$). We here refer to it as the \textit{Bianchi class A development} of the initial data. It is a Bianchi spacetime 
in the sense of Definition~\ref{def:Bianchispacetime}. Moreover, $a(t)$ is a diagonal matrix for every $t\in I$. 

\textbf{Bianchi class A developments which are not of type IX.} In case the underlying unimodular Lie group is not of type IX, 
\cite[Lemma~21.5, p.~491]{BianchiIXattr}, \cite[Lemma~21.8, p.~493]{BianchiIXattr} and \cite[Lemma~20.6, p.~218]{minbok} imply that 
either the development is a quotient of Minkowski space (in particular, it is causally geodesically complete); or, after a 
suitable choice of time orientation, $t_{-}>-\infty$, $t_{+}=\infty$ and $\theta(t)>0$ for all $t\in I$. In what follows, we refer to the 
developments that are not quotients of Minkowski space as \textit{non-Minkowski developments}. For non-Minkowski developments, 
the interval $I$ in $t$-time corresponds to $(-\infty,\infty)$ in $\tau$-time; $\theta(t)\rightarrow\infty$ as $t\rightarrow t_{-}$; 
and $\theta(t)\rightarrow 0$ as $t\rightarrow t_{+}$. This follows from \cite[Lemma~22.4, p.~497]{BianchiIXattr} and its proof
(note that the $\tau$ appearing in \cite{BianchiIXattr} can be assumed to coincide with the $\tau$ of the present paper; cf. the
beginning of Subsection~\ref{ssection:genobBianchiA}). In 
particular, for non-Minkowski developments, $t_{-}$ is a monotone volume singularity. Finally, the non-Minkowski developments are past 
causally geodesically incomplete and future causally geodesically complete; cf. \cite[Lemma~21.8, p.~493]{BianchiIXattr}. 

\textbf{Bianchi type IX developments.} Bianchi type IX developments typically recollapse in the sense that they are past and future causally 
geodesically incomplete. For diagonal Bianchi type IX solutions with a vanishing cosmological constant; non-negative mean pressure; and matter satisfying 
the dominant energy condition, this was demonstrated by Xue-feng Lin and Robert Wald in \cite{law}. Here we, via \cite{BianchiIXattr,minbok},
appeal to this result. From now on, we therefore take for granted that $1\leq \g\leq 2$ (in order to ensure that the mean pressure is 
non-negative). Due to \cite[Lemma~21.6, p.~492]{BianchiIXattr}, it then follows that there is a $t_{0}\in I$ such that $\theta(t)>0$ in 
$(t_{-},t_{0})$ and $\theta(t)<0$ in $(t_{0},t_{+})$. In other words, the spacetime expands; reaches a moment of maximal expansion (as measured
by the volume); and then starts to contract. Moreover, due to \cite[Lemma~21.8, p.~493]{BianchiIXattr}, $t_{+}<\infty$, $t_{-}>-\infty$, and the 
spacetime is future and past causally geodesically incomplete. By an argument which is essentially identical to the proof of 
\cite[Lemma~20.8, p.~219]{minbok} (this lemma covers the vacuum case), it can also be demonstrated that $\theta(t)\rightarrow\mp\infty$
as $t\rightarrow t_{\pm}\mp$. Finally, due to \cite[Lemma~22.5, p.~498]{BianchiIXattr}, the intervals $(t_{-},t_{0})$ and $(t_{0},t_{+})$ in 
$t$-time correspond to $(-\infty,\tau_{a})$ in $\tau$-time for some $\tau_{a}\in\ro$
corresponding to the maximal volume of the spatial hypersurfaces of homogeneity. In particular, if $1\leq\g\leq 2$, then Bianchi type IX 
developments have monotone volume singularities both to the future and to the past.

\subsection{Non-exceptional Bianchi class B developments}\label{ssection:BianchiclassB}
Assuming $G$ to be a non-exceptional Bianchi class B Lie group, the notion of initial data introduced in Definition~\ref{def:Bianchiid} coincides
with that introduced in \cite[Definition~1.5, p.~4]{RadermacherNonStiff}. Moreover, there is a notion of \textit{Bianchi class B development} of the 
data, introduced in \cite[Definition~11.15, p.~71]{RadermacherNonStiff}. A Bianchi class B development is a Bianchi spacetime in the sense of 
Definition~\ref{def:Bianchispacetime}; cf. Subsection~\ref{ssection:genobBianchiB} below. Moreover, due to \cite[Lemma~11.16, p.~71]{RadermacherNonStiff}, 
it either has $\theta(t)>0$ for all $t\in I$, or it arises from 
initial data whose universal covering space is initial data for Minkowski space; Bianchi type I is included in the framework of 
\cite{RadermacherNonStiff}. Restricting our attention to Bianchi class B, it is thus clear that $\theta(t)>0$ for all $t\in I$. Due to the 
construction of the Bianchi class B development, $(t_{-},t_{+})$ corresponds to $(-\infty,\infty)$ in $\tau$-time; $t_{-}>-\infty$; 
$t_{+}=\infty$; $\theta(t)\rightarrow\infty$ as $t\rightarrow t_{-}$; and $\theta(t)\rightarrow \theta_{\infty}$ as $t\rightarrow t_{+}$
for some constant $\theta_{\infty}\geq 0$. The reader interested in a justification of these statements is referred to  
\cite[Subsection~11.5, pp.~66-67]{RadermacherNonStiff} (the $\tau$ appearing in \cite{RadermacherNonStiff} can be assumed to coincide with the 
$\tau$ of the present paper; cf. the beginning of Subsection~\ref{ssection:genobBianchiB}). Note, in particular, that $t_{-}$ is a monotone 
volume singularity. Finally, note that a Bianchi class B development with $\rho=0$ or with $\rho>0$ and $0<\g\leq 2$ is past causally geodesically 
incomplete and future causally geodesically incomplete; cf. \cite[Lemma~11.20, p.~74]{RadermacherNonStiff}. 

\subsection{Silence} 
To determine whether the singularity is silent or not is, in general, quite complicated. For example, Charles Misner 
suggested in 1969 that Bianchi type IX singularities might not be silent; cf. \cite{misner}. He also suggested that the non-silent nature 
might have important implications in physics. However, it was only 47 years later that relevant results were obtained. In fact, in \cite{brehm}, 
Bernhard Brehm demonstrated that generic Bianchi type VIII and IX vacuum spacetimes have silent monotone volume singularities; here generic 
means that the relevant subsets of initial data have full measure. For this reason, silence is part of the assumptions in the general 
results that we formulate. However, we also demonstrate that the assumptions are satisfied for large classes of Bianchi developments. 

\subsection{Maximality} 
Next, we clarify in which sense the Bianchi class A and B developments described above are maximal. In the non-vacuum setting, i.e. if 
$\rho>0$, the spacetime Ricci curvature contracted with itself tends to infinity in the incomplete directions of causal geodesics, 
assuming $0<\g\leq 2$ in the case of Bianchi class B and $1\leq\g\leq 2$ in the case of Bianchi class A. In the case of Bianchi class A,
this statement follows from \cite[Lemma~22.3, p.~497]{BianchiIXattr} and the above observations concerning causal geodesic incompleteness. 
In the case of Bianchi class B, it follows from \cite[Lemma~12.1, p.~81]{RadermacherNonStiff} and the above observations concerning 
causal geodesic incompleteness. In these settings, the Bianchi developments are thus inextendible as $C^{2}$-manifolds, 
and in this sense maximal; cf., e.g., the proof of \cite[Lemma~18.18, pp.~204--205]{minbok}. In the Bianchi class A vacuum developments, 
the Kretschmann scalar, defined to be the contraction of the spacetime Riemann tensor with itself, blows up in the incomplete directions
of causal geodesics, except if the development is non-generic in the following sense: it is a quotient of Minkowski space; it is a flat
Kasner solution; or it is a locally rotationally symmetric Bianchi type II, VIII or IX solution. This follows 
from \cite[Theorem~24.12, p.~258]{minbok}. Generic Bianchi class A vacuum developments are thus maximal in the sense that they are 
$C^{2}$-inextendible; cf. \cite[Lemma~18.18, p.~204]{minbok}. In the case of Bianchi class B vacuum developments, the same inextendibility holds, 
except if the development belongs to one of the following non-generic types: it is a plane wave solution or it is a 
locally rotationally symmetric Bianchi type VI${}_{-1}$ vacuum development. This follows from \cite[Theorem~1.11, p.~5]{RadermacherNonStiff}. 
To conclude: all the developments of interest are such that they have monotone volume singularities and are $C^{2}$-inextendible. Since
all the developments are Bianchi spacetimes in the sense of Definition~\ref{def:Bianchispacetime}, and since Bianchi spacetimes are globally
hyperbolic, cf.  Lemma~\ref{lemma:Bianchispacetimegloballyhyperbolic} below, it is tempting to say that the relevant Bianchi developments
are the maximal globally hyperbolic developments (or maximal Cauchy developments) of the initial data. In the vacuum setting, this is certainly
true. However, it is important to note that in the case that $\g<1$, it is not clear that the underlying initial value problem is well posed; 
cf., e.g., \cite[p.~211]{far}. In that case, it is therefore not clear that it is possible to argue as in, e.g., \cite{cag} in order to 
demonstrate the existence of a maximal globally hyperbolic development.

\section{Convergence results I}\label{section:applications}

Theorem~\ref{thm:main} ensures that $u_{\sigma}$ is bounded. Next, we provide conditions ensuring convergence. 

\subsection{Leading order asymptotics}\label{ssection:leadingorderasymptotics}

In order to deduce that $u_{\sigma}$ converges, it is sufficient to make slightly stronger assumptions than those appearing in 
Theorem~\ref{thm:main}. 

\begin{prop}\label{prop:osc}
Given that the conditions of Theorem~\ref{thm:main} are satisfied, assume, in addition, that 
\begin{equation}\label{eq:mffsqmomoneintegrable}
\int_{-\infty}^{0}\ldr{\tau}\mff^{2}(\tau)d\tau<\infty.
\end{equation}
Then, if $u$ is a smooth solution to (\ref{eq:KGvarphiz}), there is a $u_{1}\in C^{\infty}(G)$ with the property that for every compact 
subset $K\subseteq G$ and $0\leq l\in\zo$, 
\[
\lim_{\sigma\rightarrow-\infty}\|u_{\sigma}(\cdot,\sigma)-u_{1}\|_{C^{l}(K)}=0.
\]
\end{prop}
\begin{remark}
For $\xi\in\cn{d}$, the notation $\ldr{\xi}$ is defined by 
\begin{equation}\label{eq:ldrxidef}
\ldr{\xi}:=(1+|\xi|^{2})^{1/2}. 
\end{equation}
\end{remark}
\begin{remark}
The proof of this proposition is to be found in Subsection~\ref{ssection:condresoscillatory} below. Note also that under the assumptions
of the proposition, $\tau\rightarrow-\infty$ corresponds to $\sigma\rightarrow-\infty$; cf. Remark~\ref{remark:sigmainftauinfcorr}. 
\end{remark}

Next, we give an example of the implications of this result. 

\begin{example}[Leading order asymptotics]\label{example:usigmaconvgeneralI}
Let $(G,\bge,\bk,\rho_{0})$ be Bianchi orthogonal perfect fluid initial data in the sense of Definition~\ref{def:Bianchiid}. Assume, moreover, 
that $G$ is a unimodular Lie group and that $1\leq \g\leq 2$. Let $(M,g)$ be the corresponding Bianchi class A development; cf. 
Subsection~\ref{ssection:BianchiclassA}. Focusing on the non-Minkowski developments, we can assume the time orientation to be such that 
$t_{-}$ is a monotone volume singularity. Assume now that $\varphi_{0}$ is a bounded function of $t$ only. Assume, finally, that there are 
$C_{0}$ and $0<\lambda_{0},\a_{0}\leq 1$ such that 
\begin{equation}\label{eq:hasharpsubexpbd}
\|\ha^{-1}(\tau)\|\leq C_{0}\exp(-2\lambda_{0}\ldr{\tau}^{\a_{0}})
\end{equation}
for all $\tau\leq 0$. Note that this assumption implies that the monotone volume singularity is silent; cf. (\ref{eq:intainvhainv}). 
Then, if $u$ is a smooth solution to (\ref{eq:KGvarphiz}), there is a $u_{1}\in C^{\infty}(G)$ with the property that for every compact 
subset $K\subseteq G$ and $0\leq l\in\zo$, there is a constant $C_{K,l}$ such that 
\begin{equation}\label{eq:usigmauoasymptoticswithrate}
\|u_{\sigma}(\cdot,\sigma)-u_{1}\|_{C^{l}(K)}\leq C_{K,l}\ldr{\tau}^{2-\a_{0}}e^{-2\lambda_{0}\ldr{\tau}^{\a_{0}}}
\end{equation}
for all $\tau\leq 0$. In particular, the left hand side of this inequality converges to zero as $\sigma\rightarrow-\infty$. 

Turning to the assumptions, the conditions on $\varphi_{0}$ can be relaxed; the interested reader is referred to the proof. However, the 
main assumption here is that (\ref{eq:hasharpsubexpbd}) hold. Consider non-Minkowski Bianchi type I, II, VI${}_{0}$ and VII${}_{0}$ vacuum 
developments. Assume the corresponding monotone volume singularity to be silent (this only excludes a subset of the locally rotationally
symmetric developments, in case they exist for the given Bianchi type). Then (\ref{eq:hasharpsubexpbd}) holds with $\a_{0}=1$. Turning 
to Bianchi type VIII and IX vacuum developments, there are unpublished results by Bernhard 
Brehm demonstrating that (\ref{eq:hasharpsubexpbd}) holds generically. In fact, there is a subset $\mG$ of the set of Bianchi type IX 
vacuum initial data such that the following holds: $\mG$ has full measure, and for a solution corresponding to initial data in $\mG$, there 
is a constant $C_{0}>0$ with the property that (\ref{eq:hasharpsubexpbd}) holds for all $\tau\leq 0$, $\lambda_{0}=1$ and $\a_{0}=1/6$ (in 
fact, for any $\a_{0}<1/5$ such a statement holds). The statement concerning Bianchi type VIII is similar. 
\end{example}
\begin{remark}
We justify the statements made in the example in Subsection~\ref{ssection:condresoscillatory} below.
\end{remark}

\subsection{Complete asymptotics}\label{ssection:completeasymptotics}

Strengthening the assumptions of Proposition~\ref{prop:osc} slightly yields complete asymptotics. 

\begin{prop}\label{prop:fullasymptoticsBVIIIandIXgen}
Given that the conditions of Theorem~\ref{thm:main} are satisfied, assume, in addition, that (\ref{eq:mffsqmomoneintegrable}) holds, 
and that 
\[
\int_{-\infty}^{0}\int_{-\infty}^{\tau}\ldr{\tau'}[|\hvph_{0}(\tau')|+\|\ha^{-1}(\tau')\|]
\exp\left(-\int_{\tau'}^{0}[q(\tau'')-2]d\tau''\right)d\tau' d\tau<\infty.
\]
Then, if $u$ is a smooth solution to (\ref{eq:KGvarphiz}), there are $u_{0},u_{1}\in C^{\infty}(G)$ such that for every compact subset 
$K\subseteq G$ and $0\leq l\in\zo$, 
\[
\lim_{\sigma\rightarrow-\infty}\|u_{\sigma}(\cdot,\sigma)-u_{1}\|_{C^{l}(K)}=0,\ \ \
\lim_{\sigma\rightarrow-\infty}\|u(\cdot,\sigma)-u_{1}\cdot\sigma-u_{0}\|_{C^{l}(K)}=0.
\]
\end{prop}
\begin{remark}
The proof of this statement is to be found in Subsection~\ref{ssection:condresoscillatoryII} below.
\end{remark}

Next, we give an example illustrating that the conditions are satisfied for generic Bianchi class A vacuum solutions. 

\begin{example}[Complete asymptotics]\label{example:fullasymptotics}
Given that the assumptions stated at the beginning of Example~\ref{example:usigmaconvgeneralI} are satisfied, assume, in addition, 
that there are $C_{0}$ and $0<\lambda_{0},\a_{0}\leq 1$ such that 
\begin{equation}\label{eq:hasharpsubexpbdint}
\|\ha^{-1}(\tau)\|\exp\left(-\int_{\tau}^{0}(q-2)d\tau'\right)\leq C_{1}\exp(-2\lambda_{0}\ldr{\tau}^{\a_{0}})
\end{equation}
for all $\tau\leq 0$. Then the conclusions of Proposition~\ref{prop:fullasymptoticsBVIIIandIXgen} hold. In fact, if $u$ is a smooth 
solution to (\ref{eq:KGvarphiz}), there are $u_{0},u_{1}\in C^{\infty}(G)$
with the property that for each compact set $K$ and $0\leq l\in\zo$, there is a constant $C_{K,l}$ such that
(\ref{eq:usigmauoasymptoticswithrate}) and
\begin{equation}\label{eq:asymptoticsforucompas}
\|u(\cdot,\sigma)-u_{1}\cdot\sigma-u_{0}\|_{C^{l}(K)}\leq C_{K,l}\ldr{\tau}^{3-2\a_{0}}e^{-2\lambda_{0}\ldr{\tau}^{\a_{0}}}
\end{equation}
hold for all $\tau\leq 0$. 

As in Example~\ref{example:usigmaconvgeneralI}, the assumptions concerning $\varphi_{0}$ can be relaxed. Consider non-Minkowski Bianchi 
type I, II, VI${}_{0}$ and VII${}_{0}$ vacuum developments. Assume the corresponding monotone volume singularity to be silent (this only 
excludes a subset of the locally rotationally symmetric developments, in case they exist for the given Bianchi type). Then the assumptions 
are satisfied with $\a_{0}=1$. Turning to Bianchi type VIII and IX vacuum developments, there are unpublished results by Bernhard Brehm 
demonstrating that (\ref{eq:hasharpsubexpbdint}) holds generically. In fact, there is a subset $\mG$ of the set of Bianchi type IX vacuum 
initial data such that the following holds: $\mG$ has full measure, and for a solution corresponding to initial data in $\mG$, there is a 
constant $C_{0}>0$ with the property that (\ref{eq:hasharpsubexpbdint}) holds for all $\tau\leq 0$, $\lambda_{0}=1$ and $\a_{0}=1/6$ (in fact, 
for any $\a_{0}<1/5$ such a statement holds). The statement concerning Bianchi type VIII is similar. 
\end{example}
\begin{remark}
We justify the statements made in the example in Subsection~\ref{ssection:condresoscillatoryII} below.
\end{remark}

\section{Convergence results II}\label{section:convresII}

Recall the deceleration parameter $q$ introduced in (\ref{eq:qdefinition}). Considering Theorem~\ref{thm:main} and 
Proposition~\ref{prop:fullasymptoticsBVIIIandIXgen}, it is clear that the difference $q-2$ plays an important role in 
the analysis. It is also of interest to note that for large classes of solutions, $q-2$ converges to zero exponentially. 
The following result is therefore of interest. 

\begin{prop}\label{prop:asymptoticsexponconvofqtotwo}
Let $(M,g)$ be a Bianchi spacetime with a monotone volume singularity $t_{-}$. Assume that $g$ solves (\ref{eq:Einsteinsequations}); 
that $\rho\geq\bp$; that $\rho\geq 0$; and that $\Lambda\geq 0$. Let $\varphi_{0}\in C^{\infty}(M)$ be such that it only depends on 
$t$. Assume that there are constants $C_{0}$ and $\eta_{0}>0$ such that 
\begin{equation}\label{eq:qminustwoetcdecexp}
|q(\tau)-2|+\g_{\bS}(\tau)+|\hvph_{0}(\tau)|+\|\ha^{-1}(\tau)\|\leq C_{0}e^{\eta_{0}\tau}
\end{equation}
for all $\tau\leq 0$. Then there is a constant $\eta_{1}>0$ such that the following holds. Given a smooth solution $u$ to 
(\ref{eq:KGvarphiz}), there are $u_{a},u_{b},u_{i}\in C^{\infty}(G)$, $i=0,1$, such that for every compact subset 
$K\subseteq G$, there are constants $C_{K,l}$, depending on $0\leq l\in\zo$, $K$ and $u$, 
such that 
\begin{equation}\label{eq:utauqconvtotwo}
\|u_{\tau}(\cdot,\tau)-u_{1}\|_{C^{l}(K)}+\|u(\cdot,\tau)-u_{1}\tau-u_{0}\|_{C^{l}(K)}\leq C_{K,l}\ldr{\tau}e^{\eta_{0}\tau}
\end{equation}
for all $\tau\leq 0$ and such that 
\begin{equation}\label{eq:usigmaqconvtotwo}
\|u_{\sigma}(\cdot,\sigma)-u_{b}\|_{C^{l}(K)}+\|u(\cdot,\sigma)-u_{b}\sigma-u_{a}\|_{C^{l}(K)}\leq C_{K,l}\ldr{\sigma}e^{\eta_{1}\sigma}
\end{equation}
for all $\sigma\leq 0$.
\end{prop}
\begin{remarks}
The proof is to be found in Subsection~\ref{ssection:limitsinamodelcase} below. In (\ref{eq:utauqconvtotwo}), it is assumed that $u$ is considered 
to be a function of the time coordinate $\tau$ introduced in Definition~\ref{def:monotonevolumesingularity}. In (\ref{eq:usigmaqconvtotwo}), it is 
assumed that $u$ is considered to be a function of the time coordinate $\sigma$ introduced in connection with (\ref{eq:dsigmadtdefrel}). 
Note also that, under the assumptions of the proposition, $\sigma$ and $\tau$ are effectively related by an affine transformation; cf. 
Subsection~\ref{ssection:limitsinamodelcase} below, in particular (\ref{eq:sigmataurelconvtotwo}). Therefore, the asymptotics with respect to 
the $\tau$-time and the $\sigma$-time are effectively the same. 
\end{remarks}
\begin{remark}
For generic Bianchi type IX orthogonal perfect fluids with $1\leq\g<2$, $q$ does not converge; cf., e.g., \cite{BianchiIXattr}. In particular, the 
proposition does not apply in that case. However, as illustrated by Proposition~\ref{prop:fullasymptoticsBVIIIandIXgen} and 
Example~\ref{example:fullasymptotics}, the asymptotics can nevetheless be the same with respect to the $\sigma$-time, even though the error 
estimate might be worse. In Section~\ref{section:convresIII} below, we also consider the case that $q$ converges to a limit different 
from $2$. 
\end{remark}

\subsection{Applications to stiff fluids}

Proposition~\ref{prop:asymptoticsexponconvofqtotwo} applies in several different situations. We begin by considering stiff fluids. 

\begin{example}[Stiff fluids]\label{example:asymptoticsstifffluid}
Let $(G,\bge,\bk,\rho_{0})$ be Bianchi orthogonal perfect fluid initial data in the sense of Definition~\ref{def:Bianchiid}. Assume, moreover, 
that $G$ is not of type VI${}_{-1/9}$; that $\rho_{0}>0$; and that $\g=2$ (i.e., that the fluid is stiff). The reason we exclude Bianchi type VI${}_{-1/9}$ 
is that we are unaware of any results concerning the asymptotics of Bianchi type VI${}_{-1/9}$ spacetimes. Let $(M,g)$ be the Bianchi class A/Bianchi 
class B development corresponding to the given initial data; cf. Subsections~\ref{ssection:BianchiclassA} and \ref{ssection:BianchiclassB}. If $G$ is 
not of Bianchi type IX, the development can be time oriented so that $t_{-}$ is a monotone volume singularity and if $G$ is of Bianchi type IX, $t_{-}$ 
and $t_{+}$ are both monotone volume singularities; cf. Subsections~\ref{ssection:BianchiclassA} and \ref{ssection:BianchiclassB}. Fix a 
$\varphi_{0}$ depending only on $t$. Fix, moreover, a monotone volume singularity, say $t_{-}$, and assume that there is a $t_{0}\in I$ and constants 
$C_{\varphi}$ and $\eta_{\varphi}>0$ such that 
\begin{equation}\label{eq:varphizeroestimate}
|\varphi_{0}(t)|\leq C_{\varphi}[\det a(t)]^{-1+\eta_{\varphi}}
\end{equation}
for all $t\in (t_{-},t_{0})$. Then the conditions of Proposition~\ref{prop:asymptoticsexponconvofqtotwo} are fulfilled, so that the conclusions apply
to solutions to (\ref{eq:KGvarphiz}) in the direction of the monotone volume singularity $t_{-}$. We justify this statement 
in Subsection~\ref{ssection:proofsstifffluidcase} below. 
\end{example}

Even though it is of interest to study the Klein-Gordon equation on stiff fluid backgrounds, it is worth pointing out that there are several 
full non-linear results concerning Einstein's equations in the stiff fluid setting. 

\begin{remark}[Full non-linear results]
There is a construction of solutions to the Einstein-stiff fluid equations with prescribed asymptotics by Lars Andersson and Alan 
Rendall; cf. \cite{aar}. Note, however, that this construction is restricted to the case that $\bK^{i}_{\phantom{i}j}:=\bk^{i}_{\phantom{i}j}/\theta$ converges 
to a positive definite matrix, a restriction we do not impose in Example~\ref{example:asymptoticsstifffluid}. There is also a proof of full non-linear 
stability of Big Bang formation in the Einstein-stiff fluid setting by Igor Rodnianski and Jared 
Speck; cf. \cite{rasql,rasq}. This result concerns perturbations of isotropic Bianchi type I solutions. In particular, the behaviour of solutions is such that 
$\bK$ is asymptotically close to the identity multiplied by $1/3$. In \cite{specks3}, Jared Speck also demonstrates past and future global 
non-linear stability of Big Bang/Big Crunch formation in the case of isotropic Bianchi type IX solutions. In particular, $\bK$ is again asymptotically close 
to a multiple of the identity. 
\end{remark}

\subsection{Bianchi class A developments, non-stiff fluids}\label{ssection:BianchiAgeneric}
Next, we apply Proposition~\ref{prop:asymptoticsexponconvofqtotwo} to Bianchi class A non-stiff fluids. In the oscillatory setting, $q$ is not 
expected to converge. For that reason, we exclude Bianchi types VIII and IX from the current discussion. This leaves Bianchi types I, II, VI${}_{0}$ 
and VII${}_{0}$. To begin with, we consider the vacuum setting. 

\begin{example}[Non-oscillatory Bianchi class A vacuum developments]\label{example:nonoscBclassAdev}
Let $(G,\bge,\bk,\rho_{0})$ be Bianchi orthogonal perfect fluid initial data in the sense of Definition~\ref{def:Bianchiid}. Assume that $G$ is a 
unimodular Lie group and that $\rho_{0}=0$; i.e., the data are Bianchi class A vacuum initial data. Assume, moreover, that $G$ is neither of type VIII or IX. 
Let $(M,g)$ be the corresponding Bianchi class A development, cf. Subsection~\ref{ssection:BianchiclassA}, and assume that it is not a quotient of 
Minkowski space. Then, by an appropriate choice of time orientation, the development is past causally geodesically incomplete and future causally 
geodesically complete. Moreover, the incomplete direction corresponds to a monotone volume singularity, say, $t_{-}$. Assume this monotone volume singularity
to be silent. Finally, fix a $\varphi_{0}$ depending only on $t$, and assume that there is a $t_{0}\in I$ and constants $C_{\varphi}$ and $\eta_{\varphi}>0$ such 
that (\ref{eq:varphizeroestimate}) holds for all $t\in (t_{-},t_{0})$. Then the conditions of Proposition~\ref{prop:asymptoticsexponconvofqtotwo} are 
fulfilled, so that the conclusions apply to solutions to (\ref{eq:KGvarphiz}) in the direction of the monotone volume singularity 
$t_{-}$. We justify this statement in Subsection~\ref{ssection:asymptoticsinthenonstifffluidcase} below. 
\end{example}
\begin{remark}
There are some Bianchi class A vacuum developments that can be extended through Cauchy horizons; cf., e.g., \cite[Theorem~24.12, p.~258]{minbok}. 
One example is the flat Kasner solution, which is discussed in greater detail in Section~\ref{section:nonsilentexample} below. 
In Example~\ref{example:nonoscBclassAdev}, these developments are excluded by the requirement that the monotone volume singularity be silent; cf. 
Subsection~\ref{ssection:genobBianchiA} below for details. 
\end{remark}

Next, we consider Bianchi class A orthogonal perfect fluids with $2/3<\g<2$. Again, we exclude Bianchi types VIII and IX for reasons mentioned above. 
However, we also exclude Bianchi type VI${}_{0}$, since we are unaware of any appropriate results in the literature. On the other hand, we expect the 
relevant conclusions concerning solutions to the Klein-Gordon equation to hold in the Bianchi type VI${}_{0}$ setting as well. 

\textit{Generic Bianchi class A solutions.}
Unfortunately, the above exclusions are not enough. In order to explain why, it is of interest to note that in the Bianchi class A orthogonal perfect fluid
setting, there is a formulation of the equations due to Wainwright and Hsu, cf. \cite{waihsu89}. The formulation involves scale free variables $\Sigma_{\pm}$, 
$N_{i}$, $i=1,2,3$, and $\Omega$; we describe some of their properties in Subsection~\ref{ssection:genobBianchiA} below. Note, however, that the variables 
are also introduced in \cite[p.~487]{BianchiIXattr}. In the Wainwright-Hsu formulation of the equations, there are fixed points 
$F$ and $P_{i}^{+}(II)$, $i=1,2,3$, defined as follows (cf. \cite[Definition~4.1, p.~417]{BianchiIXattr}).

\begin{definition}\label{def:fixedpoints}
The critical point $F$ is defined by the condition that $\Omega=1$; $\Sigma_{+}=\Sigma_{-}=0$; and $N_{i}=0$, $i=1,2,3$. In case $2/3<\g<2$, the critical 
point $P_{1}^{+}(II)$ is defined to be the type II point with $\Sigma_{-}=0$; $N_{1}>0$; $\Sigma_{+}=(3\g-2)/8$; and $\Omega=1-(3\g-2)/16$. 
\end{definition}
\begin{remark}
Combining the requirement that $P_{1}^{+}(II)$ be a type II point with the requirement that $N_{1}>0$ implies that $N_{2}=N_{3}=0$. Moreover, the Hamiltonian
constraint, cf., e.g., \cite[(11), p.~415]{BianchiIXattr}, together with the requirements stated in the definition determines the value of $N_{1}$ uniquely. 
\end{remark}
\begin{remark}
The fixed points $P_{i}^{+}(II)$, $i=2,3$, are defined analogously; cf. \cite[Definition~4.1, p.~417]{BianchiIXattr}.
\end{remark}
Note that the fixed point $F$ corresponds to the isotropic Bianchi type I solutions. Moreover, in the case of both $F$ and $P_{i}^{+}(II)$, the deceleration 
parameter is constant. However, the relevant constant is different from $2$, assuming $\g<2$. In particular, it is thus clear that 
Proposition~\ref{prop:asymptoticsexponconvofqtotwo} does not apply. In addition to the fixed points introduced in Definition~\ref{def:fixedpoints}, there 
are solutions converging to them. Let $\mP_{\roII}$ and $\mP_{\roVIIz}$ be the subsets of Bianchi type II and VII${}_{0}$ initial data such that the 
corresponding solutions converge to one of the $P_{i}^{+}(II)$, $i=1,2,3$. Then $\mP_{\roII}$ consists of points (in a state space of dimension $3$); 
and $\mP_{\roVIIz}$ is a $C^{1}$-submanifold of dimension $1$ (in a state space of dimension $4$); cf. \cite[p.~417]{BianchiIXattr} for a justification. 
In particular, the set of solutions converging to one of the $P_{i}^{+}(II)$, $i=1,2,3$, is clearly non-generic. 

Next, let $\mf_{\roI}$, $\mf_{\roII}$ and $\mf_{\roVIIz}$ be the subsets of Bianchi type I, II and VII${}_{0}$ initial data respectively such that the 
corresponding solutions converge to $F$. Then $\mf_{\roI}$ consists of the point $F$ (in a state space of dimension $2$); $\mf_{\roII}$  is a 
$C^{1}$-submanifold of dimension $1$ (in a state space of dimension $3$); and $\mf_{\roVIIz}$ is a $C^{1}$-submanifold of dimension $2$ (in a state space of 
dimension $4$); cf. \cite[p.~417]{BianchiIXattr} for a justification. Again, it is clear that the set of solutions converging to $F$ is non-generic. 

Due to the above observations, it is convenient to refer to Bianchi type I, II and VII${}_{0}$ initial data that do not belong to any of 
$\mf_{\roI}$, $\mf_{\roII}$, $\mf_{\roVIIz}$, $\mP_{\roII}$ or $\mP_{\roVIIz}$ as \textit{generic}. We return to a discussion of 
non-generic initial data in Section~\ref{section:convresIII} below. However, we here focus on the generic case. 

\begin{example}[Generic Bianchi type I, II and VII${}_{0}$ developments]\label{example:genericBianchiclassA}
Let $(G,\bge,\bk,\rho_{0})$ be Bianchi orthogonal perfect fluid initial data in the sense of Definition~\ref{def:Bianchiid}. Assume, moreover, that $G$ is a 
unimodular Lie group of type I, II or VII${}_{0}$; that $\rho_{0}>0$; and that $2/3<\g<2$. Assume, finally, that the initial data are generic in the sense 
described immediately prior to the statement of the present example. Let $(M,g)$ be the corresponding Bianchi class A development; cf. 
Subsection~\ref{ssection:BianchiclassA}. Then, by an appropriate choice of time orientation, the development is past causally 
geodesically incomplete and future causally geodesically complete. Moreover, the incomplete direction corresponds to a monotone volume singularity, say, 
$t_{-}$. Assume the monotone volume singularity to be silent (this only excludes a subset of the locally rotationally
symmetric developments). 
Finally, fix a $\varphi_{0}$ depending only on $t$, and assume that there is a $t_{0}\in I$ and constants $C_{\varphi}$ and $\eta_{\varphi}>0$ such 
that (\ref{eq:varphizeroestimate}) holds for all $t\in (t_{-},t_{0})$. Then the conditions of Proposition~\ref{prop:asymptoticsexponconvofqtotwo} are 
fulfilled, so that the conclusions apply to solutions to (\ref{eq:KGvarphiz}) in the direction of the monotone volume singularity 
$t_{-}$. We justify this statement in Subsection~\ref{ssection:asymptoticsinthenonstifffluidcase} below. 
\end{example}

\subsection{Bianchi class B developments, non-stiff fluids}

As a final application of Proposition~\ref{prop:asymptoticsexponconvofqtotwo}, we consider non-exceptional Bianchi class B developments in the 
non-stiff fluid setting. We restrict our attention to the case that $0\leq\g<2/3$. The reason for this is that there are results in this setting 
due to \cite{RadermacherNonStiff}. It would be of interest to extend the results of \cite{RadermacherNonStiff} to $2/3\leq\g<2$. However, we do not
attempt to do so here. 

\begin{example}[Non-exceptional Bianchi class B developments]\label{example:nonexcBianchiclassB}
Let $(G,\bge,\bk,\rho_{0})$ be Bianchi orthogonal perfect fluid initial data in the sense of Definition~\ref{def:Bianchiid}. Assume, moreover, 
that $G$ is a non-exceptional Bianchi class B Lie group and that $0\leq\g<2/3$. Let $(M,g)$ be the corresponding Bianchi class B development; cf. 
Subsection~\ref{ssection:BianchiclassB}. Then, by an appropriate choice of time orientation, the development is past causally 
geodesically incomplete and future causally geodesically complete. Moreover, the incomplete direction corresponds to a monotone volume singularity, say, 
$t_{-}$. Assume this volume singularity to be silent. Assume, moreover, that the development is not a locally rotationally symmetric Bianchi type
VI$_{-1}$ development. Finally, fix a $\varphi_{0}$ depending only on $t$, and assume that there is a $t_{0}\in I$ and 
constants $C_{\varphi}$ and $\eta_{\varphi}>0$ such that (\ref{eq:varphizeroestimate}) holds for all $t\in (t_{-},t_{0})$. Then the conditions of 
Proposition~\ref{prop:asymptoticsexponconvofqtotwo} are fulfilled, so that the conclusions apply to solutions to 
(\ref{eq:KGvarphiz}) in the direction of the monotone volume singularity $t_{-}$. We justify this statement in 
Subsection~\ref{ssection:asymptoticsinthenonstifffluidcase} below. 
\end{example}

\section{Convergence results III}\label{section:convresIII}

Due to the results of Section~\ref{section:convresII}, it is clear that solutions to the Klein-Gordon equation on non-oscillatory Bianchi orthogonal 
perfect fluid backgrounds quite generally have the asymptotics described in Proposition~\ref{prop:asymptoticsexponconvofqtotwo}. However, for 
Proposition~\ref{prop:asymptoticsexponconvofqtotwo} to be applicable, the deceleration parameter $q$ introduced in (\ref{eq:qdefinition}) has to 
converge to $2$. On the other hand, even in the generic oscillatory Bianchi type VIII and IX vacuum settings, the asymptotics with respect to 
$\sigma$-time are as described in Proposition~\ref{prop:asymptoticsexponconvofqtotwo}, though the estimate concerning the error term is worse; cf. 
Example~\ref{example:fullasymptotics}. This is true in spite of the fact that $q$ does not converge in the oscillatory setting. Finally, in the 
analysis in Subection~\ref{ssection:BianchiAgeneric}, we exclude a non-generic subset. In the corresponding Bianchi developments $q$ converges, but 
to a $q_{\infty}\neq 2$. In the present section, 
we therefore consider the case that $q$ converges to a number different from $2$. We also apply the corresponding result to the Klein-Gordon 
equation on backgrounds corresponding to the non-generic subset excluded in Subsection~\ref{ssection:BianchiAgeneric}. 

\begin{prop}\label{prop:asymptoticsexponconvofqtoqinfdifffromtwo}
Let $(M,g)$ be a Bianchi spacetime with a monotone volume singularity $t_{-}$. Assume that $g$ solves (\ref{eq:Einsteinsequations}); 
that $\rho\geq\bp$; that $\rho\geq 0$; and that $\Lambda\geq 0$. Let $\varphi_{0}\in C^{\infty}(M)$ be such that it only depends on 
$t$. Assume that there are constants $C_{0}$ and $\eta_{0}>0$ such that 
\begin{equation}\label{eq:qminusqinfetcdecexp}
\g_{\bS}(\tau)+|\hvph_{0}(\tau)|+\|\ha^{-1}(\tau)\|\leq C_{0}e^{\eta_{0}\tau}
\end{equation}
for all $\tau\leq 0$. Assume, moreover, that there is a $q_{\infty}<2$ such that $q(\tau)\rightarrow q_{\infty}$ as 
$\tau\rightarrow-\infty$ and such that 
\[
\int_{\tau}^{0}[q(\tau')-q_{\infty}]d\tau'
\]
converges as $\tau\rightarrow-\infty$. Finally, let $\eta_{c}$ be defined by 
\begin{equation}\label{eq:etacdef}
\eta_{c}:=\frac{\eta_{0}}{2-q_{\infty}}.
\end{equation}
If $\eta_{c}\leq 1$, then, given a smooth solution $u$ to (\ref{eq:KGvarphiz}), there is a function 
$u_{1}\in C^{\infty}(G)$ such that for every compact subset $K\subseteq G$, there are constants $C_{K,l}$, depending on 
$0\leq l\in\zo$, $K$ and $u$, such that 
\begin{equation}\label{eq:usigmaasetacleqone}
\|u_{\sigma}(\cdot,\sigma)-u_{1}\|_{C^{l}(K)}\leq C_{K,l}\ldr{\sigma}^{-\eta_{c}}
\end{equation}
for all $\sigma\leq 0$. If $\eta_{c}>1$, then, given a smooth solution $u$ to (\ref{eq:KGvarphiz}), there 
are $u_{i}\in C^{\infty}(G)$, $i=0,1$, such that for every compact subset $K\subseteq G$, there are constants $C_{K,l}$, depending on 
$0\leq l\in\zo$, $K$ and $u$, such that 
\[
\ldr{\sigma}\|u_{\sigma}(\cdot,\sigma)-u_{1}\|_{C^{l}(K)}+\|u(\cdot,\sigma)-u_{1}\sigma-u_{0}\|_{C^{l}(K)}\leq C_{K,l}\ldr{\sigma}^{1-\eta_{c}}
\]
for all $\sigma\leq 0$. 
\end{prop}
\begin{remark}
The proof is to be found in Subsection~\ref{ssection:qconvtoqinfdifftwo} below. 
\end{remark}
\begin{remark}
In Proposition~\ref{prop:asymptoticsexponconvofqtotwo}, we consider the case that $q$ converges to $2$. In 
Proposition~\ref{prop:asymptoticsexponconvofqtoqinfdifffromtwo}, we consider the case that $q$ converges to $q_{\infty}<2$. Naively, it is then 
of interest to ask what happens if $q$ converges to $q_{\infty}>2$. Note, however, that if $\Lambda\geq 0$ and $\rho\geq\bp$ (the latter inequality
follows from the dominant energy condition), then Remark~\ref{remark:boundonqminustwo} below implies that $q-2\leq 3\g_{\bS}$. In other words, 
$q\leq 2$ for all Bianchi types except IX. Moreover, as was noted in the comments following the statement of Theorem~\ref{thm:main}, 
$\g_{\bS}$ converges to zero even in the case of Bianchi type IX orthogonal perfect fluid developments. If $q$ converges, the limit should 
thus belong to the disc of radius $2$. In Proposition~\ref{prop:asymptoticsexponconvofqtotwo} we consider convergence to the boundary, and in 
Proposition~\ref{prop:asymptoticsexponconvofqtoqinfdifffromtwo} we consider convergence to the interior. 
\end{remark}
\begin{remark}
In case $\eta_{c}\leq 1$, the asymptotics are incomplete. On the other hand, inserting the information that (\ref{eq:usigmaasetacleqone}) holds
into the equation (strictly speaking, the estimate (\ref{eq:usigmaasetacleqone}), as well as an integrated version of it, should be inserted into 
(\ref{eq:KleinGordonwrtsigmatime}) below) yields improved asymptotics; cf. the proof in Subsection~\ref{ssection:qconvtoqinfdifftwo} below. 
\end{remark}

Next, we consider the non-generic subset excluded in Subsection~\ref{ssection:BianchiAgeneric}. 

\begin{example}[Non-generic Bianchi class A developments]\label{example:nongenBclassAdev}
Let $(G,\bge,\bk,\rho_{0})$ be Bianchi orthogonal perfect fluid initial data in the sense of Definition~\ref{def:Bianchiid}. Assume that $G$ is 
unimodular and that $\rho_{0}>0$.
If $G$ is not of type IX, assume that $2/3<\g<2$. If $G$ is of type IX, assume that $1\leq\g<2$. Let $(M,g)$ be the corresponding Bianchi class 
A development; cf. Subsection~\ref{ssection:BianchiclassA}. If $G$ is of type IX, the corresponding development has a monotone volume singularity
both to the future and to the past. Fix, in that case, one monotone volume singularity, say $t_{-}$. If $G$ is not of type IX, then the development
can be assumed to be past causally geodesically incomplete and future causally geodesically complete. Moreover, $t_{-}$ is a monotone volume singularity. 
Fix a bounded $\varphi_{0}$ depending only on $t$. Next, we consider two cases:
\begin{itemize}
\item Assume the asymptotics in the direction of the monotone volume singularity $t_{-}$ to be such that the Wainwright-Hsu variables converge to one 
of the $P_{i}(II)$, $i=1,2,3$. Then the conditions of Proposition~\ref{prop:asymptoticsexponconvofqtoqinfdifffromtwo} are satisfied with 
$q_{\infty}=(3\g-2)/2$ and $\eta_{0}=3q_{\infty}/2$. In particular, if $\g\leq 6/5$, then $\eta_{c}\leq 1$ and the first conclusion of the proposition holds. 
If $6/5<\g<2$, then the second conclusion holds. 
\item Assume the asymptotics in the direction of the monotone volume singularity $t_{-}$ to be such that the Wainwright-Hsu variables converge to 
the fixed point $F$. Then the conditions of Proposition~\ref{prop:asymptoticsexponconvofqtoqinfdifffromtwo} are satisfied with 
$q_{\infty}=(3\g-2)/2$ and $\eta_{0}=2q_{\infty}$. In particular, if $\g\leq 10/9$, then $\eta_{c}\leq 1$ and the first conclusion of the proposition holds. 
If $10/9<\g<2$, then the second conclusion holds. 
\end{itemize}
We justify the above statements in Section~\ref{section:proofsIII} below. 
\end{example}
\begin{remark}\label{remark:afaf}
Isotropic Bianchi type I backgrounds are also discussed in \cite{afaf}. Note that these backgrounds correspond to the fixed point $F$ itself. 
In that case, the metric coefficients are explicit functions of proper time. In \cite{afaf}, the question of blowup of solutions is also addressed. 
We return to this topic in Section~\ref{section:blowupintro} below. 
\end{remark}

\section{A non-silent example}\label{section:nonsilentexample}
In order to contrast the asymptotics derived in the present paper with the non-silent setting, it is of interest to describe the 
asymptotics of solutions to the Klein-Gordon equation in the case of a flat Kasner background. Here, we therefore consider
\begin{equation}\label{eq:gfKdef}
g_{\rofK}:=-dt\otimes dt+t^{2}dx\otimes dx+dy\otimes dy+dz\otimes dz
\end{equation}
on $M_{\rofK}:=\tn{3}\times (0,\infty)$. In this case the Klein-Gordon equation $\Box_{g}u-m^{2}u=0$ (where $m$ is a constant)
can be written 
\[
t\d_{t}(tu_{t})-u_{xx}-t^{2}u_{yy}-t^{2}u_{zz}+m^{2}t^{2}u=0.
\]
On the other hand, the logarithmic volume density corresponding to (\ref{eq:gfKdef}) is $\tau=\ln t/3$. With respect to this 
time coordinate, the Klein-Gordon equation takes the form 
\[
u_{\tau\tau}-9u_{xx}-9e^{6\tau}u_{yy}-9e^{6\tau}u_{zz}+9m^{2}e^{6\tau}u=0. 
\]
When analysing the asymptotics in the direction $\tau\rightarrow-\infty$, it is convenient to divide a given solution, say $u$, 
into two parts: $u=u_{a}+u_{b}$, where 
\[
u_{a}(x,y,z,\tau):=\frac{1}{2\pi}\int_{0}^{2\pi}u_{a}(x',y,z,\tau)dx';
\]
here we think of $\tn{3}$ as being $[0,2\pi]^{3}$ with the ends identified. Due to arguments similar to those presented in 
\cite[Example~4.19, pp.~41--42]{finallinsys} and \cite[Example~4.20, pp.~42--43]{finallinsys}, there are functions 
$u_{a,i}\in C^{\infty}(\tn{2},\ro)$, $i=0,1$, such that 
\[
\d_{\tau}u_{a}-u_{a,1},\ \ \
u_{a}(\cdot,\tau)-u_{a,1}\tau-u_{a,0}
\]
converge to zero exponentially in any $C^{k}$ norm. In fact, there is a homeomorphism between the initial data to $x$-independent
solutions to the Klein-Gordon equation and the functions $u_{a,i}$, $i=0,1$. 

Turning to $u_{b}$, an argument similar to that presented in \cite[Section~5.5, pp.~56--57]{finallinsys} ensures that there is a 
unique corresponding $U_{b}\in C^{\infty}(\tn{3}\times\ro,\ro)$ with the following properties: $U_{b}$ solves the equation 
$u_{\tau\tau}-9u_{xx}=0$; 
\[
\frac{1}{2\pi}\int_{0}^{2\pi}U_{b}(x',y,z,\tau)dx'=0
\]
for all $\tau$; and 
\[
\d_{\tau}u_{b}-\d_{\tau}U_{b},\ \ \
u_{b}-U_{b}
\]
converge to zero exponentially in any $C^{k}$ norm as $\tau\rightarrow-\infty$. In fact, there is a homeomorphism in the 
$C^{\infty}$-topology corresponding to the map from the initial data for $u_{b}$ to the initial data for $U_{b}$. 

Returning to the solution $u$ to the Klein-Gordon equation, it follows that there are functions 
$u_{a,i}\in C^{\infty}(\tn{2},\ro)$, $i=0,1$, and a $U_{b}\in C^{\infty}(\tn{3}\times\ro,\ro)$, solving the equation 
$U_{\tau\tau}-9U_{xx}=0$ and with zero mean value in the $x$-direction, such that 
\[
\d_{\tau}u-\d_{\tau}U_{b}-u_{a,1},\ \ \
u-U_{b}-u_{a,1}\tau-u_{a,0}
\]
converge to zero exponentially with respect to any $C^{k}$-norm. Moreover, there is a homeomorphism (in the $C^{\infty}$-topology)
taking the initial data of $u$ to $u_{a,0}$, $u_{a,1}$ and the initial data for $U_{b}$. One particular consequence of the above 
observations is that $u_{\tau}$ does, in general, not converge. Moreover, $\sigma$ and $\tau$ are, in the present setting, related
by an affine transformation, so that $u_{\sigma}$ also does not converge in the present setting. 

\section{Blow up of solutions}\label{section:blowupintro}

In \cite{afaf}, the authors discuss the question of blow up of solutions to the wave equation on certain Bianchi type I backgrounds; cf.
Remark~\ref{remark:afaf}. It is conceivable that the arguments presented in the present paper could also be used to derive blow up criteria. 
However, we here focus on conclusions that follow from \cite{finallinsys}. The reason for this is that applying the methods of \cite{finallinsys} 
give an indication of what it would be desirable to prove more generally. Here we restrict ourselves to the special case of the Klein-Gordon 
equation (\ref{eq:KG}) on non-flat Kasner backgrounds (where we allow $m=0$ as well as purely imaginary $m$). In particular, we assume the 
background metric to be given by 
\[
g=-dt\otimes dt+\textstyle{\sum}_{i=1}^{n}t^{2p_{i}}dx^{i}\otimes dx^{i}
\]
on $\tn{n}\times (0,\infty)$, where $2\leq n\in\zo$. Here the $p_{i}$ are constants satisfying 
\[
\textstyle{\sum}_{i=1}^{n}p_{i}=\textstyle{\sum}_{i=1}^{n}p_{i}^{2}=1
\]
and the non-flatness condition corresponds to the requirement that $p_{i}<1$ for all $i$. Letting $\tau:=-\ln t$, the Klein-Gordon equation 
on these backgrounds can be written
\begin{equation}\label{eq:nonflatKasnerKleinGordon}
u_{\tau\tau}-\textstyle{\sum}_{i=1}^{d}e^{2\b_{i}\tau}u_{ii}+m^{2}e^{-2\tau}u=0
\end{equation}
on $M_{\rocon}:=\tn{n}\times \ro$, where $m$ is a constant and $\b_{i}=p_{i}-1<0$. The reader interested in a justification of this statement is 
referred to \cite[Example~4.20, pp.~42--43]{finallinsys}. Note that the choice of time coordinate is such that $\tau\rightarrow\infty$ 
corresponds to the singularity. Let $\mu$ be the smallest of the $-\b_{i}$. Then, given $s$ and a smooth solution $u$ to 
(\ref{eq:nonflatKasnerKleinGordon}), there is a
constant $C_{s}$ (depending only on $s$, the coefficients of the equation and the solution); a constant $N$ (depending only on the 
coefficients of the equation); and functions $v_{\infty},u_{\infty}\in C^{\infty}(\tn{n})$ such that 
\begin{equation}\label{eq:uudotasnonflatKasner}
\left\|\left(\begin{array}{c} u(\cdot,\tau) \\ u_{\tau}(\cdot,\tau)\end{array}\right)
-\left(\begin{array}{c} v_{\infty}\tau+u_{\infty} \\ v_{\infty}\end{array}\right)\right\|_{(s)} \leq C_{s}\ldr{\tau}^{N}e^{-2\mu \tau}
\end{equation} 
for all $\tau\geq 0$. Note that the $v_{\infty}$ appearing in (\ref{eq:uudotasnonflatKasner}) corresponds to the $A_{\mathrm{Kasner}}$ appearing in 
\cite[Theorem~1.1]{afaf}. Due to \cite[Example~4.20, pp.~42--43]{finallinsys}, there is a homeomorphism (with respect to the $C^{\infty}$-topology) taking 
$[u(\cdot,0),u_{\tau}(\cdot,0)]$ to $(v_{\infty},u_{\infty})$, let us denote it by $\Psi_{\infty}$. However, this statement can be improved substantially. 
In fact, for every $s\in\ro$ and $\e>0$, $\Psi_{\infty}$ extends to a bounded linear map
\begin{equation}\label{eq:Psiinfext}
\Psi_{\infty,s,\e}:H^{(s+1/2+\e)}(\tn{n})\times H^{(s-1/2+\e)}(\tn{n})\rightarrow H^{(s)}(\tn{n})\times H^{(s)}(\tn{n});
\end{equation}
cf. Subsection~\ref{ssection:blowupcriteria} below for a justification of this statement.
Note, in particular, that even if $u_{\tau}(\cdot,0)$ only belongs to $H^{(s-1/2+\e)}(\tn{n})$, it follows from the continuity of 
$\Psi_{\infty,s,\e}$ that 
\[
v_{\infty}=\lim_{\tau\rightarrow\infty}u_{\tau}(\cdot,\tau)\in H^{(s)}(\tn{n}). 
\]
In other words, the limit of $u_{\tau}$ is in general more regular than the initial datum for $u_{\tau}$, and the gain is, essentially, half a 
derivative. Let $\Phi_{\infty}:=\Psi_{\infty}^{-1}$. Then, for every $s\in\ro$ and $\e>0$, $\Phi_{\infty}$ extends to a bounded linear map
\begin{equation}\label{eq:Phiinfext}
\Phi_{\infty,s,\e}:H^{(s+1/2+\e)}(\tn{n})\times H^{(s+1/2+\e)}(\tn{n})\rightarrow H^{(s+1)}(\tn{n})\times H^{(s)}(\tn{n});
\end{equation}
cf. Subsection~\ref{ssection:blowupcriteria} below for a justification of this statement.
Combining (\ref{eq:Psiinfext}) and (\ref{eq:Phiinfext}) yields the conclusion that the regularity properties of both $\Psi_{\infty,s,\e}$
and $\Phi_{\infty,s,\e}$ are essentially optimal (up to the loss of an $\e$). 

\textbf{$L^{2}$-blow up.} Turning to blow up criteria, there is, for every $\e>0$, a subset, say $\ma_{\e}$, of the set of smooth initial data.
This set is open with respect to the $H^{(1/2+\e)}(\tn{n})\times H^{(-1/2+\e)}(\tn{n})$-topology and dense with respect to the $C^{\infty}$-topology. 
Moreover, if $u$ is a solution arising from initial data in $\ma_{\e}$, then the corresponding $v_{\infty}$ is such that $\|v_{\infty}\|_{L^{2}}>0$, and 
$\|u(\cdot,\tau)\|_{L^{2}}\rightarrow\infty$ as $\tau\rightarrow\infty$. The reader interested in a justification of these statements is referred to 
Subsection~\ref{ssection:blowupcriteria} below. For comparison, note that \cite[Theorem~1.2]{afaf} yields an open criterion ensuring that 
$\|v_{\infty}\|_{L^{2}}>0$. However, the openness is with respect to the $H^{2}(\tn{n})\times H^{1}(\tn{n})$-topology on initial data. It is of course not
necessary to restrict one's attention to smooth initial data; we encourage the interested reader to derive the relevant statements in lower regularity. 

\textbf{$C^{1}$-blow up.} Next, let us formulate a criterion ensuring $C^{1}$-blow up of the solution. First, let $\mc$ be the subset of $C^{\infty}(\tn{n})$ 
such that $\varphi\in\mc$ if and only if $d\varphi(\bx)\neq 0$ for every $\bx$ such that $\varphi(\bx)=0$. Then there is, for every $\e>0$, a subset, 
say $\mb_{\e}$, of the set of smooth initial data. This set is open with respect to the $H^{((n+3)/2+\e)}(\tn{n})\times H^{((n+1)/2+\e)}(\tn{n})$-topology and 
dense with respect to the $C^{\infty}$-topology. Moreover, if $u$ is a solution arising from initial data in $\mb_{\e}$, then the corresponding 
$v_{\infty}$ belongs to $\mc$. In particular, 
$\Sigma_{\infty}:=v_{\infty}^{-1}(0)$ is a smooth hypersurface in $\tn{n}$. Moreover, if $\bx\notin\Sigma_{\infty}$, then $|u(\bx,\tau)|\rightarrow\infty$
as $\tau\rightarrow\infty$. If, on the other hand, $\bx\in\Sigma_{\infty}$, then $|\bd u(\bx,\tau)|_{\bge_{0}}\rightarrow\infty$ as $\tau\rightarrow\infty$. 
Here $\bd$ is the spatial $d$-operator; in other
words, we consider $u$ to be a function of the spatial variable only and calculate the differential on $\tn{n}$. Moreover, $\bge_{0}$ is the 
standard Riemannian metric on $\tn{n}$. The conclusion is thus that at a given $\bx\in\tn{n}$, either $u$ blows up, or the spatial 
derivative of $u$ blows up. Again, it is of course not necessary to restrict one's attention to smooth initial data. 

\textbf{Comparison.} Comparing (in the special case of (\ref{eq:nonflatKasnerKleinGordon})) the methods and results of \cite{afaf} and the present paper 
with those of \cite{finallinsys}, it is clear that the methods used in the first two papers have an advantage in that one would expect 
them to be more robust. On the other hand, it is unclear how to use energy-type methods to prove continuity results such as (\ref{eq:Psiinfext}). In 
order to emphasize the importance of this continuity, note the following: 
\begin{itemize}
\item The limit of $u_{\tau}$ determines the blow up of solutions.
\item Using the results of \cite{finallinsys}, the limit of $u_{\tau}$ is roughly half a derivative more regular than the 
initial data for the first derivatives. 
\item Due to (\ref{eq:Psiinfext}) and (\ref{eq:Phiinfext}), the gain of half a derivative can be expected to be, essentially, optimal. 
\end{itemize}
It would of course be desirable to obtain continuity as in (\ref{eq:Psiinfext}) using energy-type methods. We leave this as an open problem. 
The interested reader is also encouraged to apply the methods of \cite[Chapter~8]{finallinsys} to the Klein--Gordon equation on other Bianchi type I 
backgrounds.

\section{Geometric background material}\label{section:geometricbackground}

In this section, we collect background information concerning the geometry that will be of importance in what follows. 

\subsection{The Raychaudhuri equation}\label{ssection:raychaudhuri}

In the analysis of the asymptotic behaviour in the direction of the singularity, $\Sigma^{i}_{\phantom{i}j}:=\bsigma^{i}_{\phantom{i}j}/\theta$ 
and $\dot{\theta}/\theta^{2}$ play a central role; here we use the notation introduced in (\ref{eq:thetaandsigmaijdef}) and Latin indices 
are raised and lowered with $\bge$. In the present subsection, we focus on the second of these objects. 

\begin{lemma}
Let $(M,g)$ be a Bianchi spacetime satisfying the Einstein equations (\ref{eq:Einsteinsequations}). Then
\begin{equation}\label{eq:Raychaudhurirelpropertime}
-3\frac{\dot{\theta}}{\theta^{2}}=1+q,
\end{equation}
where we use the notation introduced in (\ref{eq:rescaledSigmaij}) and (\ref{eq:qdefinition}).
\end{lemma}
\begin{remark}
With respect to the time coordinate $\tau$, (\ref{eq:Raychaudhurirelpropertime}) can be written
\begin{equation}\label{eq:thetaprimeitoq}
\theta'=-(1+q)\theta,
\end{equation}
where $\theta':=\d_{\tau}\theta$. 
\end{remark}
\begin{proof}
Note, first of all, that 
\begin{equation}\label{eq:bkijdefandformula}
\bk_{ij}:=\bk(e_{i},e_{j})=\ldr{\nabla_{e_{i}}\d_{t},e_{j}}=\frac{1}{2}\dot{a}_{ij},
\end{equation}
where $\ldr{\cdot,\cdot}:=g$. Raising indices with $a^{ij}$, this relation implies that $\dot{a}_{mj}=2a_{mi}\bk^{i}_{\phantom{i}j}$. It can also be 
calculated that 
\[
\d_{t}\bk_{lm}=2\bk^{i}_{\phantom{i}l}\bk_{mi}-(\rotr_{\bge}\bk)\bk_{lm}+R_{lm}-\bR_{lm},
\]
where $R_{lm}$ and $\bR_{lm}$ are the $lm$'th components of the Ricci curvatures of the spacetime and hypersurface respectively;
cf., e.g., \cite[(25.12), p.~438]{stab}. On the other hand, (\ref{eq:bkijdefandformula}) yields
$\dot{a}^{jl}=-2\bk^{jl}$. Combining these two observations yields
\[
\d_{t}\bk^{j}_{\phantom{j}m}=-(\rotr_{\bge}\bk)\bk^{j}_{\phantom{j}m}+a^{jl}R_{lm}-a^{jl}\bR_{lm}.
\]
In particular, 
\[
\dot{\theta}=-\theta^{2}+R_{00}+S-\bS,
\]
where $R_{00}$ is the $00$ component of the spacetime Ricci curvature and $S$ and $\bS$ are the scalar curvatures of the spacetime 
and hypersurface respectively. Taking the trace of (\ref{eq:Einsteinsequations}) yields
\begin{equation}\label{eq:Sexpressionfirstversion}
S=4\Lambda-\rotr T=4\Lambda+T_{00}-a^{ij}T_{ij},
\end{equation}
where $T_{\a\b}:=T(e_{\a},e_{\b})$. Combining this observation with (\ref{eq:Einsteinsequations}) yields
\[
R_{00}+S=T_{00}+\Lambda+\frac{1}{2}S=\frac{3}{2}(\rho-\bp+2\Lambda),
\]
where we used the terminology introduced in (\ref{eq:rhobpdef}). Summarising, 
\begin{equation}\label{eq:dotthetafirstversion}
\dot{\theta}=-\theta^{2}-\bS+3\Lambda+\frac{3}{2}(\rho-\bp).
\end{equation}
On the other hand, the Hamiltonian constraint reads
\[
\frac{1}{2}[\bS-\bk^{ij}\bk_{ij}+(\rotr_{\bge}\bk)^{2}]-\Lambda=\rho.
\]
This equality can be rewritten
\begin{equation}\label{eq:bSexpressionusingHamcon}
\bS=\bsigma^{ij}\bsigma_{ij}-\frac{2}{3}\theta^{2}+2\Lambda+2\rho.
\end{equation}
Combining this observation with (\ref{eq:dotthetafirstversion}) yields the conclusion of the lemma. 
\end{proof}

\subsection{Geometric estimates}

Before proceeding, it is of interest to develop some intuition concerning the matrix with components $2\de^{i}_{m}-3\Sigma^{i}_{\phantom{i}m}$.
The reason for this is that this object appears in the basic energy estimate, cf. (\ref{eq:dtaumebasicestimate}) below. 
\begin{lemma}\label{lemma:Sigmaupperandlowerbound}
Let $(M,g)$ be a Bianchi spacetime. Assume, moreover, that $g$ solves (\ref{eq:Einsteinsequations}), that $\rho\geq 0$ and that 
$\Lambda\geq 0$. Then for all $v\in\rn{3}$ and $t\in I$ such that $\theta(t)\neq 0$, 
\begin{equation}\label{eq:twodeijminusthreeSigmaij}
(2\de^{i}_{m}-3\Sigma^{i}_{\phantom{i}m})a^{mj}v_{i}v_{j}\geq -\frac{3}{2}\frac{\bS_{+}}{\theta^{2}}a^{ij}v_{i}v_{j},
\end{equation}
where $\bS_{+}$ is the positive part of the spatial scalar curvature $\bS$; i.e., $\bS_{+}:=\max\{\bS,0\}$. 
\end{lemma}
\begin{remark}
Only Bianchi type IX is consistent with $\bS>0$. In other words, only if the universal covering group of the Lie group $G$ appearing in 
Definition~\ref{def:Bianchispacetime} is isomorphic to $\mathrm{SU}(2)$ can $\bS$ be strictly positive. For a proof of this statement, see, 
e.g., \cite[Appendix~E]{stab}.
\end{remark}
\begin{proof}
Note, to begin with, that the Hamiltonian constraint can be written
\begin{equation}\label{eq:rescaledHamcon}
1=\frac{3}{2}\Sigma^{ij}\Sigma_{ij}-\frac{3}{2}\frac{\bS}{\theta^{2}}+\Omega_{\Lambda}+\Omega_{\rho};
\end{equation}
cf. (\ref{eq:bSexpressionusingHamcon}). On the other hand, since $\rho,\Lambda\geq 0$, we know that 
$\Omega_{\Lambda}$, $\Omega_{\rho}\geq 0$. As a consequence, $\Sigma$ satisfies the estimate
\begin{equation}\label{eq:Sigmasquaredbound}
\Sigma^{i}_{\phantom{i}j}\Sigma^{j}_{\phantom{j}i}=\Sigma^{ij}\Sigma_{ij}\leq \frac{2}{3}+\frac{\bS_{+}}{\theta^{2}}=\frac{2}{3}+\g_{\bS}
\end{equation}
where $\g_{\bS}$ is defined by (\ref{eq:gammabSdef}) and we used the fact that $\Sigma$ is symmetric. Note that $\Sigma$
can be considered to be an element of $\mathrm{Hom}(TG)$. Moreover, this element is symmetric with respect to the inner product defined
by $\bge$. In particular, there are thus real eigenvalues $\lambda_{l}$, $l=1,2,3$, and corresponding eigenvectors $v_{l}$ such that 
$\Sigma^{i}_{\phantom{i}k}v^{k}_{l}=\lambda_{l}v^{i}_{l}$ (no summation on $l$). Moreover, if $\lambda_{1}\leq\lambda_{2}\leq\lambda_{3}$, then 
\[
\lambda_{1}a_{ij}v^{i}v^{j}\leq a_{ij}\Sigma^{i}_{\phantom{i}k}v^{k}v^{j}\leq \lambda_{3}a_{ij}v^{i}v^{j}.
\]
Note also that the matrix with components $\Sigma^{k}_{\phantom{k}j}\Sigma^{j}_{\phantom{j}i}$ has eigenvectors $v_{l}$, but eigenvalues 
$\lambda_{l}^{2}$. Thus (\ref{eq:Sigmasquaredbound}) implies that 
\begin{equation}\label{eq:sumlambdalsquaredbd}
\lambda_{1}^{2}+\lambda_{2}^{2}+\lambda_{3}^{2}\leq \frac{2}{3}+\g_{\bS}.
\end{equation}
The sum of the $\lambda_{l}$'s vanishes, since $\Sigma$ is trace free. We therefore focus on
\[
\lambda_{+}:=\frac{3}{2}(\lambda_{2}+\lambda_{3}),\ \ \
\lambda_{-}:=\frac{\sqrt{3}}{2}(\lambda_{2}-\lambda_{3}). 
\]
Using this terminology, the estimate (\ref{eq:sumlambdalsquaredbd}) can be written 
\[
\lambda_{+}^{2}+\lambda_{-}^{2}\leq 1+\frac{3}{2}\g_{\bS}.
\]
In particular, $9\lambda_{1}^{2}=4\lambda_{+}^{2}\leq 4+6\g_{\bS}$. Similarly, 
\[
9\lambda_{2}^{2}=(\lambda_{+}+\sqrt{3}\lambda_{-})^{2}=\lambda_{+}^{2}+2\sqrt{3}\lambda_{+}\lambda_{-}+3\lambda_{-}^{2}
\leq 4(\lambda_{+}^{2}+\lambda_{-}^{2})\leq 4+6\g_{\bS}.
\]
The estimate for $\lambda_{3}$ is similar, so that 
\begin{equation*}
\begin{split}
3\Sigma^{i}_{\phantom{i}m}a^{mj}v_{i}v_{j} = & 3a_{il}\Sigma^{i}_{\phantom{i}m}(a^{mj}v_{j})\cdot (a^{lk}v_{k})
 \leq 3\lambda_{3}a_{ij}(a^{ik}v_{k})(a^{jl}v_{l})=3\lambda_{3}a^{ij}v_{i}v_{j}\\
 \leq & (2+3\g_{\bS}/2)a^{ij}v_{i}v_{j}. 
\end{split}
\end{equation*}
In particular, 
\[
(2\de^{i}_{m}-3\Sigma^{i}_{\phantom{i}m})a^{mj}v_{i}v_{j}\geq -3\g_{\bS}a^{ij}v_{i}v_{j}/2. 
\]
The lemma follows.
\end{proof}

There is a related upper bound on $q$ which we simply state as a remark. 

\begin{remark}\label{remark:boundonqminustwo}
Due to the definition of $q$ and (\ref{eq:rescaledHamcon}), it follows that 
\[
q=2+3\frac{\bS}{\theta^{2}}+\frac{3}{2}(\Omega_{\bp}-\Omega_{\rho}-2\Omega_{\Lambda}).
\]
Here, we are mainly interested in the case that $\Lambda\geq 0$ and $\rho\geq\bp$ (note that the latter inequality follows 
from the dominant energy condition). In that case, 
\[
q-2\leq 3\g_{\bS}.
\]
\end{remark}

\subsection{Partial integration}

The analysis of the asymptotic behaviour of solutions is based on a study of energies. In order to evaluate how the energies
evolve, we need to integrate by parts. In that context, the following simple observation is useful. 

\begin{lemma}\label{lemma:divergencecalculations}
Let $G$ be a connected Lie group, $X\in\mfg$ and $h$ be a left invariant Riemannian metric on $G$. Then there is a constant $c_{X}$ such that 
for every $f\in C^{\infty}_{0}(G)$, 
\begin{equation}\label{eq:intGXfmuhfinal}
\int_{G}X(f)\mu_{h}=2c_{X}\int_{G}f\mu_{h}.
\end{equation}
Moreover, 
\[
c_{X}=\xi_{G}(X),
\]
where $\xi_{G}:\mfg\rightarrow\ro$ is defined by (\ref{eq:xiGdef}). In particular, if $G$ is a unimodular Lie group, then $c_{X}=0$. 
\end{lemma}
\begin{remark}
The notation $f\in C_{0}^{\infty}(G)$ signifies that $f$ is a smooth function from $G$ to $\ro$ with compact support. 
\end{remark}
\begin{proof} 
Note, to begin with, that 
\begin{equation}\label{eq:intGXfmuh}
\begin{split}
 & \int_{G}X(f)\mu_{h}=\int_{G}\ml_{X}(f\mu_{h})-\int_{G}f\ml_{X}\mu_{h}\\
 = & \int_{G}d[i_{X}(f\mu_{h})]-\int_{G}f(\rodiv_{h}X)\mu_{h}=-\int_{G}f(\rodiv_{h}X)\mu_{h},
\end{split}
\end{equation}
where we used Cartan's magic formula and Stokes's theorem. Next, let $\{\be_{i}\}$ be an orthonormal frame of the Lie algebra 
(relative to $h$). Then 
\[
\rodiv_{h}X=\textstyle{\sum}_{i}h(\nabla^{h}_{\be_{i}}X,\be_{i})=\textstyle{\sum}_{i}h(\be_{i},[\be_{i},X])
=-\textstyle{\sum}_{i}h(\be_{i},\road_{X}\be_{i})=-2\xi_{G}(X),
\]
where $\nabla^{h}$ is the Levi-Civita connection associated with $h$. In particular, $c_{X}:=\xi_{G}(X)$ is a constant, and the 
statements of the lemma follow.  
\end{proof}

\subsection{Causal structure}\label{subsection:causalstructure}
As a first step in the study of the Klein-Gordon equation, we wish to know that there is a unique solution corresponding to initial data 
specified on a hypersurface $G_{t}$. The following results ensure that this is the case. 

\begin{lemma}\label{lemma:Bianchispacetimegloballyhyperbolic}
Let $(M,g)$ be a Bianchi spacetime. Then $(M,g)$ is globally hyperbolic, and each $G_{t}$, $t\in I$, is a Cauchy hypersurface.
\end{lemma}
\begin{proof}
The statement follows by an argument which is essentially identical to that presented at the end of the proof of 
\cite[Proposition~20.3, p.~215]{minbok}; cf. \cite[p.~217]{minbok}.
\end{proof}

Due to this observation, we can solve the Klein-Gordon equation on a Bianchi spacetime. In fact, the following result holds. 

\begin{cor}\label{cor:basicexistenceandcompactsupp}
Let $(M,g)$ be a Bianchi spacetime and $\varphi_{i}\in C^{\infty}(M)$, $i=0,1$. Then, given $u_{i}\in C^{\infty}(G)$, $i=0,1$, 
and $t_{0}\in I$, there is a unique smooth solution to
\begin{align}
\Box_{g}u+\varphi_{0}u = & \varphi_{1},\label{eq:ivpeq}\\
u(\cdot,t_{0}) = & u_{0},\label{eq:ivpidz}\\
u_{t}(\cdot,t_{0}) = & u_{1}.\label{eq:ivpido}
\end{align}
Assume, moreover, that the $u_{i}$ have compact support and that for every compact interval $J\subset I$, the support of 
$\varphi_{1}|_{G\times J}$ is compact. Then the solution to (\ref{eq:ivpeq})--(\ref{eq:ivpido}) is such that for every compact 
interval $J\subset I$, the support of $u|_{G\times J}$ is compact. 
\end{cor}
\begin{proof}
That there is a unique smooth solution to (\ref{eq:ivpeq})--(\ref{eq:ivpido}) follows, e.g., from \cite[Theorem~12.19, p.~144]{minbok}. 
Next, let $J\subset I$ be a compact interval and assume, without loss of generality, that $t_{0}\in J$. Let
\[
K:=\supp\varphi_{1}|_{G\times J}\cup (\supp u_{0}\times \{t_{0}\})\cup (\supp u_{1}\times \{t_{0}\}).
\] 
Then, due to the global hyperbolicity of $(M,g)$, 
\[
K_{1}:=J^{+}(K)\cap G_{t_{0}}
\]
is compact; cf., e.g., \cite[Lemma~21.6, p.~361]{stab}. Let $U_{1}\subseteq G$ be the open subset such that $K_{1}=(G-U_{1})\times \{t_{0}\}$.
Defining $t_{a}$ and $t_{b}$ by $J=[t_{a},t_{b}]$, let, moreover, 
\[
K_{2}:=J^{-}(K_{1})\cap J^{+}(G_{t_{a}}).
\] 
Note that $K_{2}$ is a compact subset of $G\times [t_{a},t_{0}]$; cf., e.g., \cite[Lemma~21.6, p.~361]{stab}. 

Assume that $p\in G\times [t_{a},t_{0}]$ does not belong to $K_{2}$. 
We then wish to prove that $p\in D^{-}(U_{1}\times \{t_{0}\})$. Assume, to this end, that $\g$ is a future inextendible causal curve
through $p$ that does not meet $U_{1}\times \{t_{0}\}$. Since $G_{t_{0}}$ is a Cauchy surface, there is an $s_{0}$ in the domain of 
definition of $\g$ such that $\g(s_{0})\in G_{t_{0}}$. Since $\g$ does not meet $U_{1}\times \{t_{0}\}$, we conclude that 
$\g(s_{0})\in K_{1}$. In particular, $p\in J^{-}(K_{1})$, so that $p\in K_{2}$. This is a contradiction. Thus 
$p\in D^{-}(U_{1}\times \{t_{0}\})$.

Next, we wish to prove that $\varphi_{1}$ vanishes in $D^{-}(U_{1}\times \{t_{0}\})\cap J^{+}(G_{t_{a}})$. Assume, to this end, that 
$p$ belongs to this set and that $\varphi_{1}(p)\neq 0$. Then $p\in K$, and every future inextendible curve through $p$
intersects $K_{1}$. In particular, there is a future inextendible causal curve which does not intersect $U_{1}\times \{t_{0}\}$. 
Thus $p\notin D^{-}(U_{1}\times \{t_{0}\})$, a contradiction.

Combining the above observations with \cite[Corollary~12.14, p.~141]{minbok} and \cite[Remark~12.15, p.~141]{minbok} yields the 
conclusion that the support of $u|_{G\times [t_{a},t_{0}]}$ is contained in $K_{2}$. The argument concerning $u|_{G\times [t_{0},t_{b}]}$
is similar, and the corollary follows. 
\end{proof}

Consider a silent geometry; i.e., assume that (\ref{eq:normaraisedtominusonehalf}) is satisfied. Then the following holds. 

\begin{lemma}\label{lemma:localisationuntilthesingularity}
Let $(M,g)$ be a Bianchi spacetime. Assume $t_{-}$ to be a silent monotone volume singularity. Then for every compact $K_{0}\subseteq G$,
there is a compact subset $K_{1}$ of $G$ such that 
\begin{equation}\label{eq:JminusKzeroincylinder}
J^{-}(K_{0}\times \{t_{0}\})\subseteq K_{1}\times (t_{-},t_{0}]. 
\end{equation}
Similarly, for every compact $K_{0}\subseteq G$, there is a compact subset $K_{1}$ of $G$ such that 
\begin{equation}\label{eq:JplusKzeroincylinder}
J^{+}(K_{0}\times (t_{-},t_{0}])\cap J^{-}(G_{t_{0}})\subseteq K_{1}\times (t_{-},t_{0}]. 
\end{equation}
In particular, 
\begin{equation}\label{eq:cylindercontinDminusKonCa}
K_{0}\times (t_{-},t_{0}]\subseteq D^{-}(K_{1}\times \{t_{0}\}).
\end{equation}
\end{lemma}
\begin{proof}
If $\lambda_{\min}(t)>0$ is the smallest eigenvalue of $a(t)$, then $\|a^{-1/2}(t)\|=[\lambda_{\min}(t)]^{-1/2}$; note that 
$a$ is a symmetric matrix. In particular, if $\g$ is a causal curve such that $\g(t)=[\bga(t),t]$ and the $\dot{\bga}^{i}$ 
are defined by $\dot{\g}=\d_{t}+\dot{\bga}^{i}e_{i}$, then $0\geq -1+a_{ij}\dot{\bga}^{i}\dot{\bga}^{j}$. As a consequence, 
\[
|\dot{\bga}(t)|_{h}\leq \|a^{-1/2}(t)\|.
\]
In particular, the length of $\bga|_{(t_{-},t_{0}]}$ with 
respect to $h$ (defined by (\ref{eq:hRieMetdef})) is bounded. Since $h$ is a complete Riemannian metric, it follows that if 
$K_{0}\subset G$ is a compact set, then there is a compact subset $K_{1}$ of $G$ such that (\ref{eq:JminusKzeroincylinder}) holds.
The proof of the second statement is similar. In order to prove (\ref{eq:cylindercontinDminusKonCa}), let $p\in K_{0}\times (t_{-},t_{0}]$ 
and let $\g$ be a future inextendible causal curve through $p$. Since $G_{t_{0}}$ is a Cauchy hypersurface, $\g$ intersects $G_{t_{0}}$. 
The intersection point belongs to the left hand side of (\ref{eq:JplusKzeroincylinder}). Thus, due to (\ref{eq:JplusKzeroincylinder}),
the intersection point belongs to $K_{1}\times \{t_{0}\}$. To conclude: every future inextendible causal curve through $p$ intersects
$K_{1}\times \{t_{0}\}$. Thus $p\in D^{-}(K_{1}\times \{t_{0}\})$ and (\ref{eq:cylindercontinDminusKonCa}) follows. 
\end{proof}

\section{Conformal rescaling}\label{section:conformalrescaling}

The analysis is simplified by writing down the Klein-Gordon equation with respect to a conformally rescaled metric. In fact, in analogy 
with the arguments presented in \cite[Section~1.3, pp.~8--13]{finallinsys}, it is convenient to multiply the metric with the mean curvature 
squared. It is also convenient to change the time coordinate. In what follows, we use the following rescaling. 

\begin{lemma}\label{lemma:hgintro}
Let $(M,g)$ be a Bianchi spacetime with a monotone volume singularity $t_{-}$. Then $\tau:I_{0}\rightarrow \mI_{0}$ is a diffeomorphism, 
where $I_{0}:=(t_{-},t_{0})$, $\mI_{0}:=(-\infty,\tau_{0})$ and $\tau_{0}:=\tau(t_{0})$. Moreover, if $\hg:=\theta^{2}g/9$,
then $\hg$ is a Lorentz metric on $M_{0}:=G\times I_{0}$ and 
\begin{equation}\label{eq:hgitohaij}
\hg=-d\tau\otimes d\tau+\ha_{ij}(\tau)\xi^{i}\otimes \xi^{i},
\end{equation}
where $\ha_{ij}=\theta^{2}a_{ij}/9$. 
\end{lemma}
\begin{remark}
In the statement of the lemma, we use the terminology introduced in Definitions~\ref{def:Bianchispacetime} and
\ref{def:monotonevolumesingularity}.
\end{remark}
\begin{proof}
Note that $d\tau=(\theta/3)dt$ and that $\theta>0$ on $I_{0}$. The lemma follows.
\end{proof}

\section{The Klein-Gordon equation on a Bianchi spacetime}\label{section:KGonBianchi}

Next, let us formulate the Klein-Gordon equation with respect to the conformally rescaled metric introduced in 
Lemma~\ref{lemma:hgintro}.

\begin{lemma}\label{lemma:conformalrescalingoftheequation}
Let $(M,g)$ be a Bianchi spacetime with a monotone volume singularity $t_{-}$. Then, on $M_{0}$ introduced in Lemma~\ref{lemma:hgintro}, 
the equation
\begin{equation}\label{eq:KleinGordonvarphivarphi}
\Box_{g}u+\varphi_{0}u=\varphi_{1}
\end{equation}
can be written  
\begin{equation}\label{eq:KleinGordonConformallyRescaled}
-u_{\tau\tau}+\ha^{ij}e_{i}[e_{j}(u)]+(q-2)u_{\tau}-2X_{0}(u)+\hvph_{0}u=\hvph_{1},
\end{equation}
where $q$ is defined by (\ref{eq:qdefinition}), 
\begin{equation}\label{eq:Xzerodefinition}
X_{0}:=\xi_{G}^{\sharp}=\frac{1}{2}\ha^{il}\g^{j}_{lj}e_{i},
\end{equation}
$\xi_{G}$ is given by (\ref{eq:xiGdef}), the constants $\g^{i}_{jk}$ are defined by $[e_{j},e_{k}]=\g^{i}_{jk}e_{i}$ and 
$\hvph_{i}:=9\theta^{-2}\varphi_{i}$ for $i=0,1$. Here the operator $\sharp$ is calculated with respect to $\chg$, where 
$\chg$ is the metric induced on $G_{t}$ by $\hg$; i.e., 
for all $Y\in\mfg$, 
\[
\chg(\xi_{G}^{\sharp},Y)=\xi_{G}(Y).
\]
\end{lemma}
\begin{remark}
In the statement, we use the terminology introduced in Lemma~\ref{lemma:hgintro} as well as in Definitions~\ref{def:Bianchispacetime} 
and \ref{def:monotonevolumesingularity}. 
\end{remark}
\begin{remark}
The unimodular Lie groups are characterised by the property that $X_{0}=0$. In particular, for Bianchi class A, $X_{0}=0$ and for 
Bianchi class B, $X_{0}\neq 0$. 
\end{remark}
\begin{proof}
Let $\Omega=\theta/3$. It can then be calculated that 
\[
\Omega^{-2}\Box_{g}u=\Box_{\hg}u-2\hg(\grad_{\hg}\ln\Omega,\grad_{\hg}u).
\]
Here, 
\[
\hg(\grad_{\hg}\ln\Omega,\grad_{\hg}u)=\hg^{\a\b}\d_{\a}\ln\Omega\d_{\b}u=-\d_{\tau}\ln(\theta/3)u_{\tau}=(1+q)u_{\tau},
\]
where we appealed to (\ref{eq:thetaprimeitoq}) in the last step. Next, if $\{e_{\a}\}$ is the frame consisting of 
$e_{0}:=\d_{\tau}$ combined with the $\{e_{i}\}$, then
\[
\Box_{\hg}u=\hg^{\a\b}e_{\a}[e_{\b}(u)]-\hG^{\a}e_{\a}(u),
\]
where $\hG^{\a}$ is defined by the relations
\[
\hna_{e_{\a}}e_{\b}=\hG^{\g}_{\a\b}e_{\g},\ \ \
\hG^{\g}:=\hg^{\a\b}\hG^{\g}_{\a\b}
\]
and $\hna$ is the Levi-Civita connection associated with $\hg$. Note that 
\[
\hG^{0}_{\a\b}=-\ldr{\hna_{e_{\a}}e_{\b},e_{0}},
\]
where $\ldr{\cdot,\cdot}=\hg$. In particular, $\hG^{0}_{00}=0$ and 
\[
\hG^{0}_{ij}=\frac{1}{2}\d_{\tau}\ldr{e_{i},e_{j}}=\frac{1}{2}\d_{\tau}\ha_{ij}.
\]
On the other hand, (\ref{eq:bkijdefandformula}) can be written $\dot{a}_{mj}=2a_{mi}\bk^{i}_{\phantom{i}j}$, so that
\[
\d_{\tau}a_{mj}=6a_{mi}\left(\Sigma^{i}_{\phantom{i}j}+\frac{1}{3}\de^{i}_{j}\right).
\]
Combining this observation with (\ref{eq:thetaprimeitoq}) yields
\begin{equation}\label{eq:dtauhamjfinal}
\d_{\tau}\ha_{mj}=6\ha_{mi}\left(\Sigma^{i}_{\phantom{i}j}-\frac{1}{3}q\de^{i}_{j}\right).
\end{equation}
Thus
\[
\hG^{0}=-\hG^{0}_{00}+\ha^{ij}\hG^{0}_{ij}=\ha^{ij}\frac{1}{2}\cdot 6 \ha_{il}\left(\Sigma^{l}_{\phantom{l}j}-\frac{1}{3}q\de^{l}_{j}\right)=-3q,
\]
where we used the fact that $\Sigma^{j}_{\phantom{j}j}=0$. Next, note that 
\[
\hG^{i}_{\a\b}=\ha^{ij}\ldr{\hna_{e_{\a}}e_{\b},e_{j}}.
\]
Thus
\[
\hG^{i}_{00}=\ha^{ij}\ldr{\hna_{e_{0}}e_{0},e_{j}}=-\ha^{ij}\ldr{e_{0},\hna_{e_{0}}e_{j}}
=-\ha^{ij}\ldr{e_{0},\hna_{e_{j}}e_{0}}=0.
\]
Moreover, the Koszul formula yields
\[
\hG^{i}_{jk}=\ha^{il}\ldr{\hna_{e_{j}}e_{k},e_{l}}=\frac{1}{2}\ha^{il}(-\g^{m}_{kl}\ha_{jm}+\g^{m}_{lj}\ha_{km}+\g^{m}_{jk}\ha_{lm}),
\]
where the constants $\g^{i}_{jk}$ are defined by $[e_{j},e_{k}]=\g^{i}_{jk}e_{i}$. In particular, it can thus be computed that 
\[
\hG^{i}=\ha^{jk}\hG^{i}_{jk}=\ha^{il}\g^{j}_{lj}.
\]
In order to calculate which vector field the $\hG^{i}$ correspond to, let $Y\in\mfg$, $Y=Y^{l}e_{l}$ and compute
\[
\ldr{\hG^{i}e_{i},Y}=\ha^{il}\g^{j}_{lj}\ha_{im}Y^{m}=\g^{j}_{lj}Y^{l}.
\]
On the other hand, if we define $\xi_{G,i}:=\xi_{G}(e_{i})$ (so that $\xi_{G}=\xi_{G,i}\xi^{i}$), then 
\[
\xi_{G,i}=\frac{1}{2}\rotr\, \road_{e_{i}}=\frac{1}{2}\g^{j}_{ij}. 
\]
Thus
\[
\xi_{G}(Y)=\xi_{G,i}\xi^{i}(Y)=\frac{1}{2}\g^{j}_{ij}Y^{i}=\frac{1}{2}\ldr{\hG^{i}e_{i},Y}. 
\]
Thus, defining $X_{0}$ by the first relation in (\ref{eq:Xzerodefinition}), the second relation in (\ref{eq:Xzerodefinition}) 
follows. Moreover, $\hG^{i}e_{i}=2X_{0}$. To conclude, 
\[
-\hG^{\a}e_{\a}(u)=3q u_{\tau}-2X_{0}(u).
\]
Combining the above observations yields the conclusion of the lemma. 
\end{proof}

\subsection{Additional change of time coordinate}\label{ssection:addchangetimecoord}

The form of the equation (\ref{eq:KleinGordonConformallyRescaled}) indicates that it would be of interest to change time 
coordinate according to 
\begin{equation}\label{eq:dsigmadtau}
\frac{d\sigma}{d\tau}=C_{\sigma}\exp\left(-\int_{\tau}^{0}(q-2)d\tau'\right)
\end{equation}
for some $C_{\sigma}>0$. In fact, with respect to such a time coordinate, (\ref{eq:KleinGordonConformallyRescaled}) takes the form
\begin{equation}\label{eq:KleinGordonwrtsigmatime}
-u_{\sigma\sigma}+\cha^{ij}e_{i}[e_{j}(u)]-2\chX_{0}(u)+\cvph_{0}u=\cvph_{1},
\end{equation}
where
\begin{align*}
\cha^{ij} := & C_{\sigma}^{-2}\exp\left(2\int_{\tau}^{0}(q-2)d\tau'\right)\ha^{ij},\ \ \ 
\cvph_{0} := C_{\sigma}^{-2}\exp\left(2\int_{\tau}^{0}(q-2)d\tau'\right)\hvph_{0},\\
\chX_{0} := & C_{\sigma}^{-2}\exp\left(2\int_{\tau}^{0}(q-2)d\tau'\right)X_{0},\ \ \ 
\cvph_{1} := C_{\sigma}^{-2}\exp\left(2\int_{\tau}^{0}(q-2)d\tau'\right)\hvph_{1}.
\end{align*}
On the other hand, if $\sigma$ satisfies (\ref{eq:dsigmadtdefrel}), then 
\begin{equation}\label{eq:dsigmadtauprimaryrelation}
\frac{d\sigma}{d\tau}=\frac{d\sigma}{dt}\frac{dt}{d\tau}=\frac{1}{3}(\det a)^{-1/2}\frac{3}{\theta}
=[\theta(0)]^{-1}\exp\left(-\int_{\tau}^{0}(q-2)d\tau'\right),
\end{equation}
where we appealed to (\ref{eq:taudefinition}), (\ref{eq:dsigmadtdefrel}) and (\ref{eq:thetaprimeitoq}). In particular, 
a time coordinate $\sigma$ satisfying (\ref{eq:dsigmadtdefrel}) is of the desired type. In the applications, the following
observation is of interest. 

\begin{lemma}\label{lemma:sigmataulimcorr}
Let $(M,g)$ be a Bianchi spacetime with a monotone volume singularity $t_{-}$. Assume that $g$ solves (\ref{eq:Einsteinsequations}); 
that $\rho\geq\bp$; that $\rho\geq 0$; and that $\Lambda\geq 0$. Assume that $\g_{\bS}\in L^{1}(-\infty,0]$, where $\g_{\bS}$ is given 
by (\ref{eq:gammabSdef}). Then $\sigma\rightarrow-\infty$ corresponds to $\tau\rightarrow-\infty$. 
\end{lemma}
\begin{remark}
The assumption $\g_{\bS}\in L^{1}(-\infty,0]$ is automatically fulfilled for all Bianchi types except IX. 
\end{remark}
\begin{proof}
Due to Remark~\ref{remark:boundonqminustwo}, we know that $q-2\leq 3\g_{\bS}$. Combining this estimate with the assumptions, it
is clear that $(q-2)_{+}\in L^{1}(-\infty,0]$. In particular, (\ref{eq:dsigmadtau}) implies that $d\sigma/d\tau\geq c_{0}$ for all
$\tau\leq 0$ and some $c_{0}>0$. Thus, using the fact that $\sigma(0)=0$ (cf. the requirement following (\ref{eq:dsigmadtdefrel})),
\[
-\sigma(\tau)\geq -c_{0}\tau
\]
for all $\tau\leq 0$, so that $\sigma(\tau)\rightarrow-\infty$ as $\tau\rightarrow-\infty$. On the other hand, due to 
(\ref{eq:dsigmadtau}), the only way for $\sigma$ to tend to $-\infty$ is if $\tau\rightarrow-\infty$. The lemma follows. 
\end{proof}

\section{The basic energy}\label{section:basicenergy}

Next, we analyse how the energy (\ref{eq:meepsilonzerodefinition}) evolves over time. Assume, to this end, $\varphi_{1}|_{G\times J}$ to 
have compact support for every compact interval $J\subset I$. Next, fix an $\mff$ with the properties stated in connection with 
(\ref{eq:meepsilonzerodefinition}) and let $h$ be given by (\ref{eq:hRieMetdef}). Given a solution $u$ to (\ref{eq:KleinGordonConformallyRescaled}) 
corresponding to initial data at $\tau=0$ that are compactly supported on $G$, define $\me[u]$ by (\ref{eq:meepsilonzerodefinition});
recall that we, without loss of generality, can assume $\tau_{0}>0$, cf. Remark~\ref{remark:tauzerogreaterthanzero}. 
Note that, due to Corollary~\ref{cor:basicexistenceandcompactsupp}, $\me[u]$ is well defined, smooth, and we are allowed to 
differentiate under the integral sign. In what follows, we tacitly consider the constituents of $\me[u]$, as well as $\me[u]$ itself,
as depending on $\tau$ as opposed to $t$. 
\begin{lemma}\label{lemma:dtaumebasicestimate}
Let $(M,g)$ be a Bianchi spacetime with a monotone volume singularity $t_{-}$. Let $\varphi_{i}\in C^{\infty}(M)$, $i=0,1$, be given and 
assume $\varphi_{1}|_{G\times J}$ to have compact support for every compact interval $J\subset I$. Let $h$ be given by (\ref{eq:hRieMetdef})
and fix an $\mff$ with the properties stated in connection with (\ref{eq:meepsilonzerodefinition}). Let $u$ be a solution to 
(\ref{eq:KleinGordonvarphivarphi}) corresponding to initial data that are compactly supported on $G$. Finally, define $\me[u]$ by 
(\ref{eq:meepsilonzerodefinition}), where all the constituents are considered to be functions of $\tau$. Then
\begin{equation}\label{eq:dtaumebasicestimate}
\begin{split}
\d_{\tau}\me[u] = & \int_{G}\left[(q-2)[u_{\tau}^{2}+\ha^{ij}e_{i}(u)e_{j}(u)]
+(2\de^{i}_{m}-3\Sigma^{i}_{\phantom{i}m})\ha^{mj}e_{i}(u)e_{j}(u)\right]\mu_{h}\\
 & +\int_{G}\left[\hvph_{0}uu_{\tau}-\hvph_{1}u_{\tau}+\mff'\mff u^{2}+\mff^{2}uu_{\tau}\right]\mu_{h}.
\end{split}
\end{equation}
\end{lemma}
\begin{remark}
In the statement, we use the terminology introduced in Lemma~\ref{lemma:hgintro} and Definitions~\ref{def:Bianchispacetime} and 
\ref{def:monotonevolumesingularity}.
\end{remark}
\begin{remark}
Considering (\ref{eq:dtaumebasicestimate}), it is clear that the estimate (\ref{eq:twodeijminusthreeSigmaij}) is of interest. 
Assume, therefore, that $g$ solves (\ref{eq:Einsteinsequations}), that $\rho\geq 0$ and that $\Lambda\geq 0$. Then the second term 
in the integrand (on the first line of the right hand side of (\ref{eq:dtaumebasicestimate})) is non-negative for all the Bianchi 
types except IX (since $\bS\leq 0$ for all the Bianchi types except IX); cf. (\ref{eq:twodeijminusthreeSigmaij}). In the Bianchi 
type IX cases of interest here, the positive part of the scalar curvature, say $\bS_{+}$, is such that $\bS_{+}/\theta^{2}$ decays 
exponentially as $\tau\rightarrow-\infty$. Returning to (\ref{eq:dtaumebasicestimate}), it is then clear that the second term in the 
integrand (on the first line of the right hand side) does not contribute to any significant growth of the energy as $\tau\rightarrow-\infty$. 
\end{remark}
\begin{proof}
Time differentiating $\me[u]$ yields 
\[
\d_{\tau}\me[u]=\int_{G}\left[u_{\tau}u_{\tau\tau}+\frac{1}{2}(\d_{\tau}\ha^{ij})e_{i}(u)e_{j}(u)
+\ha^{ij}e_{i}(u)e_{j}(u_{\tau})+\mff'\mff u^{2}+\mff^{2}uu_{\tau}\right]\mu_{h}.
\]
However, 
\begin{equation*}
\begin{split}
\int_{G}\ha^{ij}e_{i}(u)e_{j}(u_{\tau})\mu_{h}
 = & \int_{G}e_{j}\left[\ha^{ij}e_{i}(u)u_{\tau}\right]\mu_{h}
-\int_{G}\ha^{ij}e_{j}[e_{i}(u)]u_{\tau}\mu_{h}\\
 = & \int_{G}2c_{j}\ha^{ij}e_{i}(u)u_{\tau}\mu_{h}
-\int_{G}\ha^{ij}e_{j}[e_{i}(u)]u_{\tau}\mu_{h}\\
 = & \int_{G}2X_{0}(u)u_{\tau}\mu_{h}
-\int_{G}\ha^{ij}e_{j}[e_{i}(u)]u_{\tau}\mu_{h},
\end{split}
\end{equation*}
where $c_{i}:=\xi_{G}(e_{i})$ and we appealed to (\ref{eq:intGXfmuhfinal}); cf. Lemma~\ref{lemma:divergencecalculations} and the proof of 
Lemma~\ref{lemma:conformalrescalingoftheequation}. Combining this observation with (\ref{eq:KleinGordonConformallyRescaled}) yields
\begin{equation*}
\begin{split}
\d_{\tau}\me[u] = & \int_{G}\left[u_{\tau}\left(u_{\tau\tau}-\ha^{ij}e_{j}[e_{i}(u)]+2X_{0}(u)\right)
+\frac{1}{2}(\d_{\tau}\ha^{ij})e_{i}(u)e_{j}(u)\right]\mu_{h}\\
 & +\int_{G}\left[\mff'\mff u^{2}+\mff^{2}uu_{\tau}\right]\mu_{h}\\
 = & \int_{G}\left[(q-2)u_{\tau}^{2}-(3\Sigma^{i}_{\phantom{i}m}-q\de^{i}_{m})\ha^{mj}e_{i}(u)e_{j}(u)
+\hvph_{0}uu_{\tau}-\hvph_{1}u_{\tau}\right]\mu_{h}\\
 & +\int_{G}\left[\mff'\mff u^{2}+\mff^{2}uu_{\tau}\right]\mu_{h},
\end{split}
\end{equation*}
where we appealed to (\ref{eq:dtauhamjfinal}). The lemma follows.
\end{proof}

Next, we give an example of an estimate of $\me$ that follows from the above calculations. 
\begin{cor}\label{cor:energyestimatefirstexample}
Let $(M,g)$ be a Bianchi spacetime with a monotone volume singularity $t_{-}$. Assume that $g$ solves (\ref{eq:Einsteinsequations}); 
that $\rho\geq\bp$; that $\rho\geq 0$; and that $\Lambda\geq 0$. Let $\varphi_{0}\in C^{\infty}(M)$ be given and
let $\varphi_{1}:=0$. Let $h$ be given by (\ref{eq:hRieMetdef}) and fix an $\mff$ with the properties stated in connection with 
(\ref{eq:meepsilonzerodefinition}). Assume, moreover, that $\g_{\bS}$, given by (\ref{eq:gammabSdef}),
satisfies $\g_{\bS}\in L^{1}((-\infty,0])$ and that there is a function $\mff_{\hvph}\in L^{1}(-\infty,0]$ such that 
$|\hvph_{0}|\leq \mff_{\hvph}\mff$ for all $\tau\leq 0$. Then there is a constant $C_{0}$ such that for every smooth solution $u$ to 
(\ref{eq:KleinGordonConformallyRescaled}) corresponding to compactly supported initial data,
\[
\me[u](\tau)\leq C_{0}\me[u](0)\exp\left[2\int_{\tau}^{0}[2-q(\tau')]d\tau'\right]
\]
for all $\tau\leq 0$. 
\end{cor}
\begin{remark}
In the statement, we use the terminology introduced in Lemma~\ref{lemma:hgintro} and Definitions~\ref{def:Bianchispacetime} and 
\ref{def:monotonevolumesingularity}.
\end{remark}
\begin{remark}
The constant $C_{0}$ only depends on $\|\mff\|_{L^{1}(-\infty,0]}$, $\|\mff_{\hvph}\|_{L^{1}(-\infty,0]}$ and 
$\|\g_{\bS}\|_{L^{1}(-\infty,0]}$. 
\end{remark}
\begin{proof}
Combining Lemmas~\ref{lemma:Sigmaupperandlowerbound}, \ref{lemma:dtaumebasicestimate}, Remark~\ref{remark:boundonqminustwo} and the 
assumptions, it follows that 
\[
\d_{\tau}\me\geq 2(q-2)\me-\g_{\rotot}\me
\]
for all $\tau\leq 0$, where $\g_{\rotot}\in L^{1}((-\infty,0])$. Moreover, $\|\g_{\rotot}\|_{L^{1}(-\infty,0]}$ only depends on 
$\|\mff\|_{L^{1}(-\infty,0]}$, $\|\mff_{\hvph}\|_{L^{1}(-\infty,0]}$ and $\|\g_{\bS}\|_{L^{1}(-\infty,0]}$. The lemma follows. 
\end{proof}

\section{Higher order energies}\label{section:higherorderenergies}

In order to estimate higher order energies, we apply $e_{K}$ to (\ref{eq:KleinGordonConformallyRescaled}), where $K$ is a 
vector field multiindex; cf. Definition~\ref{definition:vectorfieldmultiindex}. We then appeal to Lemma~\ref{lemma:dtaumebasicestimate}. 
To be allowed to do so, we need to calculate the commutator of 
$e_{K}$ with the operator defined by the left hand side of (\ref{eq:KleinGordonConformallyRescaled}).

\subsection{Commutators}
To begin with, we compute the commutator of $e_{k}$ and $\ha^{ij}e_{i}e_{j}$:
\begin{equation}\label{eq:ekeisquaredreformulation}
\begin{split}
\ha^{ij}e_{k}e_{i}e_{j} = & \ha^{ij}[e_{k},e_{i}]e_{j}+\ha^{ij}e_{i}e_{k}e_{j}=\ha^{ij}\g_{ki}^{l}e_{l}e_{j}
+\ha^{ij}e_{i}[e_{k},e_{j}]+\ha^{ij}e_{i}e_{j}e_{k}\\
 = & \ha^{ij}\g_{ki}^{l}\g_{lj}^{m}e_{m}+\ha^{ij}\g_{ki}^{l}e_{j}e_{l}+\ha^{ij}\g_{kj}^{l}e_{i}e_{l}+\ha^{ij}e_{i}e_{j}e_{k}\\
 = & \ha^{ij}\g_{ki}^{l}\g_{lj}^{m}e_{m}+2\ha^{ij}\g_{ki}^{l}e_{j}e_{l}+\ha^{ij}e_{i}e_{j}e_{k}.
\end{split}
\end{equation}
Next, we calculate the commutator $[e_{K},\ha^{ij}e_{i}e_{j}]$. 

\begin{lemma}\label{lemma:higherordercommutators}
Let $G$ be a $3$-dimensional Lie group, $\{e_{i}\}$ be a basis of the Lie algebra and $\ha^{ij}$, $i,j=1,2,3$, be the components of 
a symmetric matrix. Given a vector field multiindex $K$ with $|K|\geq 1$, there are constants $c_{K,ij}^{K_{1}}$ and $c_{K,j}^{K_{2}}$ 
(for all vector field multiindices $K_{1}$ and $K_{2}$ with the property that $|K_{1}|,|K_{2}|=|K|$) such that 
\begin{equation}\label{eq:haijEKeiejcomm}
\ha^{ij}e_{K}e_{i}e_{j}=\textstyle{\sum}_{|K_{1}|=|K|}\ha^{ij}c_{K,ij}^{K_{1}}e_{K_{1}}
+\textstyle{\sum}_{|K_{2}|=|K|}\ha^{ij}c_{K,j}^{K_{2}}e_{i}e_{K_{2}}+\ha^{ij}e_{i}e_{j}e_{K}.
\end{equation}
\end{lemma}
\begin{remark}
The constants $c_{K,ij}^{K_{1}}$ and $c_{K,j}^{K_{2}}$ only depend on $i$, $j$, $K_{1}$, $K_{2}$, $K$ and the structure constants 
$\g^{k}_{ij}$. 
\end{remark}
\begin{proof}
Note that, due to (\ref{eq:ekeisquaredreformulation}), (\ref{eq:haijEKeiejcomm}) holds for $|K|=1$. Assume, inductively, 
that there is an $l\geq 1$ such that (\ref{eq:haijEKeiejcomm}) holds for $|K|\leq l$. Given a vector field multiindex $K$ 
such that $|K|=l+1$, there are $K_{1}$ and $i_{1}$ such that $e_{K}=e_{i_{1}}e_{K_{1}}$. Appealing to the inductive assumption, 
\begin{equation}\label{eq:termstobehandledininductiveargument}
\ha^{ij}e_{K}e_{i}e_{j}=\textstyle{\sum}_{|K_{a}|=|K_{1}|}\ha^{ij}c_{K_{1},ij}^{K_{a}}e_{i_{1}}e_{K_{a}}
+\textstyle{\sum}_{|K_{b}|=|K_{1}|}\ha^{ij}c_{K_{1},j}^{K_{b}}e_{i_{1}}e_{i}e_{K_{b}}+\ha^{ij}e_{i_{1}}e_{i}e_{j}e_{K_{1}}.
\end{equation}
The first term on the right hand side is already of a form consistent with the inductive hypothesis. Since 
\[
\ha^{ij}c_{K_{1},j}^{K_{b}}e_{i_{1}}e_{i}e_{K_{b}}=\ha^{ij}c_{K_{1},j}^{K_{b}}\g_{i_{1}i}^{l}e_{l}e_{K_{b}}+\ha^{ij}c_{K_{1},j}^{K_{b}}e_{i}e_{i_{1}}e_{K_{b}}
\]
and both of the terms appearing on the right hand side of this equality are of a form consistent with the inductive hypothesis, the 
second term appearing in (\ref{eq:termstobehandledininductiveargument}) can be handled. In order to demonstrate that the last
term on the right hand side of (\ref{eq:termstobehandledininductiveargument}) can be rewritten in the desired form, it is sufficient to 
appeal to (\ref{eq:ekeisquaredreformulation}). To conclude, the inductive hypothesis holds for all $l\geq 1$. 
\end{proof}

\textbf{Computing $[e_{K},X_{0}]$.} Finally, before applying $e_{K}$ to the equation, we need to compute the commutator of $e_{K}$ 
and $X_{0}$. Note, to this end, that 
\begin{equation}\label{eq:ekXzerocomm}
e_{k}X_{0}=\frac{1}{2}\ha^{il}\g^{j}_{lj}\g^{m}_{ki}e_{m}+X_{0}e_{k},
\end{equation}
where we appealed to (\ref{eq:Xzerodefinition}). We wish to prove that, given a vector field multiindex $K$, there are constants 
$d_{K,ij}^{K_{1}}$ (for all vector field multiindices $K_{1}$ with $|K_{1}|=|K|$ and all $i,j=1,2,3$) such that 
\begin{equation}\label{eq:eKXzerocomm}
e_{K}X_{0}=\textstyle{\sum}_{|K_{1}|=|K|}\ha^{ij}d_{K,ij}^{K_{1}}e_{K_{1}}+X_{0}e_{K}.
\end{equation}
Due to (\ref{eq:ekXzerocomm}), we know that (\ref{eq:eKXzerocomm}) holds for $|K|=1$. Assuming (\ref{eq:eKXzerocomm}) to hold for
$|K|=l\geq 1$, it can be demonstrated that it holds for $|K|=l+1$; the argument is similar to, but simpler than, the proof of 
Lemma~\ref{lemma:higherordercommutators}. Thus (\ref{eq:eKXzerocomm}) holds for all $|K|\geq 1$. 

\subsection{Higher order energy estimates.} Fixing a vector field multiindex $K$ and applying $e_{K}$ to 
(\ref{eq:KleinGordonConformallyRescaled}) yields (assuming $\varphi_{0}$ to only depend on $t$)
\begin{equation}\label{eq:eKappliedtoequation}
\begin{split}
-(e_{K}u)_{\tau\tau}+\ha^{ij}e_{i}e_{j}(e_{K}u)+(q-2)(e_{K}u)_{\tau}-2X_{0}(e_{K}u)+\hvph_{0}e_{K}u=\hvph_{1,K},
\end{split}
\end{equation}
where 
\begin{equation}\label{eq:hvphKonedef}
\hvph_{1,K}:=\left[\ha^{ij}e_{i}e_{j},e_{K}\right]u+2[e_{K},X_{0}]u+e_{K}(\hvph_{1}).
\end{equation}

\begin{lemma}\label{lemma:dtaumehigherorderestimate}
Let $(M,g)$ be a Bianchi spacetime with a monotone volume singularity $t_{-}$. Let $\varphi_{i}\in C^{\infty}(M)$, $i=0,1$, be given and 
assume $\varphi_{1}|_{G\times J}$ to have compact support for every compact interval $J\subset I$. Assume, moreover, $\hvph_{0}$ to only 
depend on $\tau$. Let $h$ be given by (\ref{eq:hRieMetdef}) and fix an $\mff$ with the properties stated in connection with 
(\ref{eq:meepsilonzerodefinition}). Given a smooth solution $u$ to (\ref{eq:KleinGordonConformallyRescaled}) corresponding to initial data 
that are compactly supported on $G$, define $\me_{l}[u]$ by (\ref{eq:meepsilonldefinition}). Then
\begin{equation}\label{eq:dtaumehigherorderestimate}
\begin{split}
\d_{\tau}\me_{l}[u] = & \int_{G}\textstyle{\sum}_{|K|\leq l}(q-2)\left[(e_{K}u)_{\tau}^{2}+\ha^{ij}e_{i}e_{K}(u)e_{j}e_{K}(u)\right]\mu_{h}\\
 & +\int_{G}\textstyle{\sum}_{|K|\leq l}\left[ (2\de^{i}_{m}-3\Sigma^{i}_{\phantom{i}m})\ha^{mj}e_{i}e_{K}(u)e_{j}e_{K}(u)
+\hvph_{0}e_{K}(u)(e_{K}u)_{\tau}\right]\mu_{h}\\
 & +\int_{G}\textstyle{\sum}_{|K|\leq l}\textstyle{\sum}_{|K_{1}|=|K|}\ha^{ij}c_{K,ij}^{K_{1}}e_{K_{1}}(u)(e_{K}u)_{\tau}\mu_{h}\\
 & +\int_{G}\textstyle{\sum}_{|K|\leq l}\textstyle{\sum}_{|K_{2}|=|K|}\ha^{ij}c_{K,j}^{K_{2}}e_{i}e_{K_{2}}(u)(e_{K}u)_{\tau}\mu_{h}\\
 & +\int_{G}\textstyle{\sum}_{|K|\leq l}\left[-e_{K}(\hvph_{1})(e_{K}u)_{\tau}
+\mff'\mff[e_{K}(u)]^{2}+\mff^{2}e_{K}(u)(e_{K}u)_{\tau}\right]\mu_{h},
\end{split}
\end{equation}
where $c_{K,ij}^{K_{1}}$ and $c_{K,j}^{K_{2}}$ are constants only depending on $i$, $j$, $K_{1}$, $K_{2}$, $K$ and the structure constants 
$\g^{k}_{ij}$. 
\end{lemma}
\begin{remark}
In the statement, we use the terminology introduced in Lemma~\ref{lemma:hgintro} and Definitions~\ref{def:Bianchispacetime} and 
\ref{def:monotonevolumesingularity}.
\end{remark}
\begin{proof}
We wish to apply Lemma~\ref{lemma:dtaumebasicestimate} to (\ref{eq:eKappliedtoequation}). In order to be allowed to do so, we need to verify that 
the function corresponding to $\varphi_{1}$ in Lemma~\ref{lemma:dtaumebasicestimate} satisfies the stated requirements. However, this follows from 
the assumptions of the present lemma and Corollary~\ref{cor:basicexistenceandcompactsupp}. Appealing to Lemma~\ref{lemma:dtaumebasicestimate}, it 
can be calculated that 
\begin{equation*}
\begin{split}
\d_{\tau}\me_{l}[u] = & \int_{G}\textstyle{\sum}_{|K|\leq l}(q-2)\left[(e_{K}u)_{\tau}^{2}+\ha^{ij}e_{i}e_{K}(u)e_{j}e_{K}(u)\right]\mu_{h}\\
 & +\int_{G}\textstyle{\sum}_{|K|\leq l}\left[ (2\de^{i}_{m}-3\Sigma^{i}_{\phantom{i}m})\ha^{mj}e_{i}e_{K}(u)e_{j}e_{K}(u)
+\hvph_{0}e_{K}(u)(e_{K}u)_{\tau}\right]\mu_{h}\\
 & +\int_{G}\textstyle{\sum}_{|K|\leq l}\left[-\hvph_{1,K}(e_{K}u)_{\tau}
+\mff'\mff[e_{K}(u)]^{2}+\mff^{2}e_{K}(u)(e_{K}u)_{\tau}\right]\mu_{h}.
\end{split}
\end{equation*}
We need to consider the term $-\hvph_{1,K}(e_{K}u)_{\tau}$ appearing in the integrand more carefully. Due to (\ref{eq:haijEKeiejcomm}),
(\ref{eq:eKXzerocomm}) and (\ref{eq:hvphKonedef}),
\begin{equation*}
\begin{split}
-\hvph_{1,K}(e_{K}u)_{\tau} = & \textstyle{\sum}_{|K_{1}|=|K|}\ha^{ij}c_{K,ij}^{K_{1}}e_{K_{1}}(u)(e_{K}u)_{\tau}\\
 & +\textstyle{\sum}_{|K_{2}|=|K|}\ha^{ij}c_{K,j}^{K_{2}}e_{i}e_{K_{2}}(u)(e_{K}u)_{\tau}-e_{K}(\hvph_{1})(e_{K}u)_{\tau}.
\end{split}
\end{equation*}
Combining the last two observations yields the desired result. 
\end{proof}

Due to (\ref{eq:dtaumehigherorderestimate}), we obtain a result analogous to Corollary~\ref{cor:energyestimatefirstexample}.

\begin{cor}\label{cor:energyestimatehigherorderenergies}
Let $(M,g)$ be a Bianchi spacetime with a monotone volume singularity $t_{-}$. Assume that $g$ solves (\ref{eq:Einsteinsequations}); 
that $\rho\geq\bp$; that $\rho\geq 0$; and that $\Lambda\geq 0$. Let $\varphi_{0}\in C^{\infty}(M)$ be such that it only depends on 
$t$ and let $\varphi_{1}:=0$. Let $h$ be given by (\ref{eq:hRieMetdef}) and fix an $\mff$ with the properties stated in connection with 
(\ref{eq:meepsilonzerodefinition}). Assume that $\g_{\bS}$, given by (\ref{eq:gammabSdef}), satisfies 
$\g_{\bS}\in L^{1}(-\infty,0]$; and that there is a function $\mff_{\hvph}\in L^{1}(-\infty,0]$ such that 
$|\hvph_{0}|\leq \mff_{\hvph}\mff$ for all $\tau\leq 0$. Finally, assume that there is a function 
$\mff_{a}\in L^{1}(-\infty,0]$ such that $\|\ha^{-1}\|\leq \mff_{a}\mff$, where $\ha^{-1}$ is the matrix with components $\ha^{ij}$. 
Then there is a constant $C_{l}$ such that for every smooth solution $u$ to (\ref{eq:KleinGordonConformallyRescaled}) corresponding to 
compactly supported initial data,
\begin{equation}\label{eq:energyestimatehigherorderenergies}
\me_{l}[u](\tau)\leq C_{l}\me_{l}[u](0)\exp\left[2\int_{\tau}^{0}[2-q(\tau')]d\tau'\right]
\end{equation}
for all $\tau\leq 0$. 
\end{cor}
\begin{remark}
The constant $C_{l}$ only depends on  $l$, the structure constants associated with $\{e_{i}\}$, $\|\mff\|_{L^{1}(-\infty,0]}$, 
$\|\mff_{\hvph}\|_{L^{1}(-\infty,0]}$, $\|\mff_{a}\|_{L^{1}(-\infty,0]}$ and $\|\g_{\bS}\|_{L^{1}(-\infty,0]}$. 
\end{remark}
\begin{remark}
Under the assumptions of the lemma, 
\begin{equation}\label{eq:eKusigmaLtwoestimate}
\int_{G}\textstyle{\sum}_{|K|\leq l}(e_{K}u)_{\sigma}^{2}\mu_{h}\leq 2C_{\sigma}^{-2}C_{l}\me_{l}[u](0)
\end{equation}
for all $\sigma\leq 0$; cf. (\ref{eq:dsigmadtau}). 
\end{remark}
\begin{proof}
Consider (\ref{eq:dtaumehigherorderestimate}). Under the conditions of the present corollary, we wish to demonstrate that the right hand
side can be estimated from below by 
\begin{equation}\label{eq:lowerboundToBeDemonstrated}
2(q-2)\me_{l}-\bga_{\rotot}\me_{l},
\end{equation}
where $\bga_{\rotot}$ is an element of $L^{1}(-\infty,0]$. All the terms on the right hand side of (\ref{eq:dtaumehigherorderestimate}), 
except the third and fourth lines, can be estimated as in the proof of Corollary~\ref{cor:energyestimatefirstexample}. Consider the third 
line on the right hand side of (\ref{eq:dtaumehigherorderestimate}). Since the $c^{K_{1}}_{K,ij}$ are constants, we need to estimate
\[
|\ha^{ij}e_{K_{1}}(u)(e_{K}u)_{\tau}|\leq \|\ha^{-1}\|\cdot |e_{K_{1}}(u)(e_{K}u)_{\tau}|
\leq \mff_{a}\frac{1}{2}[\mff^{2}|e_{K_{1}}(u)|^{2}+|(e_{K}u)_{\tau}|^{2}].
\]
This is an estimate of the desired type. Next, consider the fourth line on the right hand side of (\ref{eq:dtaumehigherorderestimate}).
We need to estimate
\[
\ha^{ij}e_{i}e_{K_{2}}(u)(e_{K}u)_{\tau}=\textstyle\sum_{l}\hb^{il}\hb^{lj}e_{i}e_{K_{2}}(u)(e_{K}u)_{\tau},
\]
where $\hb^{ij}$ are the components of the square root of $\ha^{-1}$. Note that 
\[
|\hb^{lj}|\leq \|\ha^{-1}\|^{1/2}\leq \mff_{a}^{1/2}\mff^{1/2}\leq\frac{1}{2}(\mff+\mff_{a})
\]
and that 
\[
|\hb^{il}e_{i}e_{K_{2}}(u)|\leq \left[\textstyle{\sum}_{l}\hb^{il}e_{i}e_{K_{2}}(u)\hb^{jl}e_{j}e_{K_{2}}(u)\right]^{1/2}
\leq \left[\ha^{ij}e_{i}e_{K_{2}}(u)e_{j}e_{K_{2}}(u)\right]^{1/2}.
\]
Thus the fourth line on the right hand side of (\ref{eq:dtaumehigherorderestimate}) can be estimated as desired.
\end{proof}

\subsection{Limits in a model case}\label{ssection:limitsinamodelcase}

Finally, we are in a position to prove Proposition~\ref{prop:asymptoticsexponconvofqtotwo}. 

\begin{proof}[Proposition~\ref{prop:asymptoticsexponconvofqtotwo}]
The idea of the proof is to appeal to Corollary~\ref{cor:energyestimatehigherorderenergies}, Sobolev embedding and the equation
(\ref{eq:KleinGordonConformallyRescaled}). Note, however, that Corollary~\ref{cor:energyestimatehigherorderenergies} only applies 
to solutions corresponding to initial data with compact support, a restriction we do not impose here. The first step of the argument
is therefore to demonstrate that we can ``localise'' the solution $u$. 

\textbf{Localising the solution $u$.}
Let $U$ be an open subset of $G$ with compact closure $K$. Assume, moreover, $U$ to be diffeomorphic to the ball of radius $1$ and center
$0$ in $\rn{3}$ and denote the diffeomorphism by $\psi$. Let $U_{0}=\psi^{-1}[B_{1/2}(0)]$. Then $U_{0}$ has compact closure contained in $U$. 
Due to the fact that $\|\ha^{-1}\|$ decays exponentially, it is clear that 
\[
\int_{t_{-}}^{t_{a}}\|a^{-1/2}\|dt=\int_{-\infty}^{0}\frac{3}{\theta}\|a^{-1/2}\|d\tau=\int_{-\infty}^{0}\|\ha^{-1}(\tau)\|^{1/2}d\tau<\infty,
\]
where $\tau(t_{a})=0$. We are thus allowed to appeal to Lemma~\ref{lemma:localisationuntilthesingularity}. In particular, there is a compact 
subset $K_{1}$ such that 
\begin{equation}\label{eq:KcylinderDminusKone}
K\times (t_{-},t_{0}]\subseteq D^{-}(K_{1}\times\{t_{0}\});
\end{equation}
cf. (\ref{eq:cylindercontinDminusKonCa}). Let $\chi\in C^{\infty}_{0}(G)$ be such that $\chi(x)=1$ for all $x\in K_{1}$. Let $u_{a}$ be the 
solution to (\ref{eq:KleinGordonvarphivarphi}) corresponding to the initial data given by $\chi u(\cdot,t_{0})$ and $\chi u_{t}(\cdot,t_{0})$. 
Then, due to (\ref{eq:KcylinderDminusKone}); \cite[Corollary~12.14, p.~141]{minbok}; and \cite[Remark~12.15, p.~141]{minbok}; the functions $u$ 
and $u_{a}$ coincide in $K\times (t_{-},t_{0}]$. If we want to analyse the asymptotic behaviour
of $u$ in $K$ as $t\rightarrow t_{-}$, we might thus as well consider $u_{a}$. 

\textbf{Appealing to the energy estimates.} In order to justify that we are allowed to appeal to 
Corollary~\ref{cor:energyestimatehigherorderenergies}, note that if we define $\mff(\tau):=e^{\eta_{0}\tau/2}$, then $\mff$ has the properties 
stated in connection with (\ref{eq:meepsilonzerodefinition}). Choosing $\mff_{\hvph}(\tau)=\mff_{a}(\tau)=C_{0}e^{\eta_{0}\tau/2}$, we are then allowed
to apply Corollary~\ref{cor:energyestimatehigherorderenergies} to $u_{a}$. Since $q-2$ is integrable, there are thus constants $C_{l}$ such 
that 
\[
\me_{l}[u_{a}](\tau)\leq C_{l}\me_{l}[u_{a}](0)
\]
for all $\tau\leq 0$. Introducing
\[
\mf_{l}[u_{a}]:=\frac{1}{2}\int_{G}\textstyle{\sum}_{|K|\leq l}[e_{K}(u_{a})]^{2}\mu_{h}, 
\]
it can be calculated that 
\begin{equation*}
\begin{split}
\d_{\tau}\mf_{l} = & \int_{G}\textstyle{\sum}_{|K|\leq l}e_{K}(u_{a})e_{K}(\d_{\tau}u_{a})]\mu_{h}\geq -2\mf_{l}^{1/2}[u_{a}]\me_{l}^{1/2}[u_{a}]\\
 \geq & -2C_{l}^{1/2}\me_{l}^{1/2}[u_{a}](0)\mf_{l}^{1/2}[u_{a}].
\end{split}
\end{equation*}
Using this estimate, it can be demonstrated that 
\[
\mf_{l}^{1/2}[u_{a}](\tau)\leq \mf_{l}^{1/2}[u_{a}](0)+C_{l}^{1/2}\me_{l}^{1/2}[u_{a}](0)|\tau|
\]
for all $\tau\leq 0$. In particular, there are constants $D_{a,l}$ such that 
\[
\mf_{l}^{1/2}[u_{a}](\tau)\leq D_{l,a}\ldr{\tau}
\]
for all $\tau\leq 0$. 

\textbf{Appealing to the equation.}
With the above information at hand, it is of interest to return to (\ref{eq:eKappliedtoequation}). In particular, 
\begin{equation*}
\begin{split}
\|\ha^{ij}e_{i}e_{j}(e_{K}u_{a})\|_{L^{2}(G)}+\|X_{0}(e_{K}u_{a})\|_{L^{2}(G)} & \\
+\|[\ha^{ij}e_{i}e_{j},e_{K}]u_{a}\|_{L^{2}(G)}
+\|[e_{K},X_{0}]u_{a}\|_{L^{2}(G)}  &\leq E_{a,l}\ldr{\tau}e^{\eta_{0}\tau}
\end{split}
\end{equation*}
for all $\tau\leq 0$ and $|K|\leq l$. Here
\[
\|\phi\|_{L^{2}(G)}:=\left(\int_{G}|\phi|^{2}\mu_{h}\right)^{1/2}
\]
for every $\phi\in C_{0}(G)$. Moreover, due to the boundedness of $\me_{l}[u_{a}]$ for $\tau\leq 0$ and (\ref{eq:qminustwoetcdecexp}),
\[
\|(q-2)(e_{K}u_{a})_{\tau}\|_{L^{2}(G)}\leq E_{a,l}e^{\eta_{0}\tau}
\]
for all $\tau\leq 0$ and $|K|\leq l$. Combining these observations with (\ref{eq:qminustwoetcdecexp}) and (\ref{eq:eKappliedtoequation}) yields 
the conclusion that 
\[
\|(e_{K}u_{a})_{\tau\tau}\|_{L^{2}(G)}\leq E_{a,l}\ldr{\tau}e^{\eta_{0}\tau}
\]
for all $\tau\leq 0$ and $|K|\leq l$. Combining this observation with Sobolev embedding yields the conclusion that 
\[
\|\d_{\tau}^{2}u_{a}(\cdot,\tau)\|_{C^{l}(U_{0})}\leq E_{a,l}\ldr{\tau}e^{\eta_{0}\tau}
\]
for all $\tau\leq 0$. Thus (\ref{eq:utauqconvtotwo}) holds. 

In order to prove (\ref{eq:usigmaqconvtotwo}), note, to begin with, that 
\[
\exp\left(-\int_{\tau}^{0}(q-2)d\tau'\right)=C_{\sigma}^{-1}C_{a}[1+O(e^{\eta_{0}\tau})]
\]
for some $C_{a}>0$, where $C_{\sigma}$ is the constant appearing in (\ref{eq:dsigmadtau}); cf. (\ref{eq:dsigmadtauprimaryrelation}). Thus 
there is a constant $C_{b}$ such that 
\[
-\sigma(\tau)=-C_{a}\tau-C_{b}+O(e^{\eta_{0}\tau}),
\]
where we used the fact that $\sigma(0)=0$; cf. the requirement following (\ref{eq:dsigmadtdefrel}). In particular, 
\begin{equation}\label{eq:sigmataurelconvtotwo}
\sigma(\tau)=C_{a}\tau+C_{b}+O(e^{\eta_{0}\tau}).
\end{equation}
Combining these observations with (\ref{eq:dsigmadtau}) and (\ref{eq:utauqconvtotwo}) yields (\ref{eq:usigmaqconvtotwo}). The lemma follows. 
\end{proof}

\subsection{Limits when $q$ does not converge to $2$}\label{ssection:qconvtoqinfdifftwo}
Next, we prove Proposition~\ref{prop:asymptoticsexponconvofqtoqinfdifffromtwo}. 

\begin{proof}[Proposition~\ref{prop:asymptoticsexponconvofqtoqinfdifffromtwo}]
We begin by changing to the time coordinate $\sigma$ introduced in (\ref{eq:dsigmadtdefrel}). Due to the assumptions,
there is a constant $c_{q}$ such that 
\[
\int_{\tau}^{0}[q(\tau')-2]d\tau'=\int_{\tau}^{0}(q_{\infty}-2)d\tau'+\int_{\tau}^{0}[q(\tau')-q_{\infty}]d\tau'=(2-q_{\infty})\tau-c_{q}+o(1).
\]
Inserting this information into (\ref{eq:dsigmadtau}) yields
\[
\frac{d\sigma}{d\tau}=C_{\sigma}\exp[-(2-q_{\infty})\tau+c_{q}][1+o(1)];
\]
note that (\ref{eq:dsigmadtauprimaryrelation}) is satisfied. Integrating this relation yields
\[
\sigma(\tau)=-\frac{C_{\sigma}}{2-q_{\infty}}e^{-(2-q_{\infty})\tau+c_{q}}[1+o(1)]
\]
for all $\tau\leq 0$. For future reference, it is of interest to note that the above computations yield the existence of a
constant $c_{a}>0$ such that 
\[
\ldr{\sigma(\tau)}=c_{a}e^{-(2-q_{\infty})\tau}[1+o(1)]
\]
for all $\tau\leq 0$. Due to this relation, it can be deduced that there is a constant $c_{b}$ such that 
\[
\tau=-\frac{1}{2-q_{\infty}}\ln\ldr{\sigma}+c_{b}+o(1).
\]
In particular, 
\[
\exp\left(2\int_{\tau}^{0}(q-2)d\tau'\right)=O(\ldr{\sigma}^{-2}). 
\]
Returning to (\ref{eq:KleinGordonwrtsigmatime}) and appealing to (\ref{eq:qminusqinfetcdecexp}), it is clear that 
\[
\|\cha^{-1}\|+|\cvph_{0}|\leq C\ldr{\sigma}^{-2-\eta_{c}}
\]
for all $\sigma\leq 0$ and some constant $C>0$; here $\eta_{c}$ is given by (\ref{eq:etacdef}). The same holds for the coefficients of $\chX$. 

Due to (\ref{eq:qminusqinfetcdecexp}), it is possible to choose exponentially decaying functions $\mff$, $\mff_{\hvph}$ 
and $\mff_{a}$ in such a way that the conditions of Corollary~\ref{cor:energyestimatehigherorderenergies} are satisfied. 
Thus (\ref{eq:eKusigmaLtwoestimate}) holds (for a suitably localised solution; cf. the proof of 
Proposition~\ref{prop:asymptoticsexponconvofqtotwo}). Combining this observation with (\ref{eq:KleinGordonwrtsigmatime}); the estimate 
for the coefficients of this equation (described above); and arguments similar to those presented in the proof of 
Proposition~\ref{prop:asymptoticsexponconvofqtotwo} yields the conclusion that a suitably localised solution (such as 
$u_{a}$ in Proposition~\ref{prop:asymptoticsexponconvofqtotwo}) satisfies
\[
\|e_{L}\d_{\sigma}^{2}u_{a}\|_{L^{2}(G)}\leq C_{K,l}\ldr{\sigma}^{-1-\eta_{c}}
\]
for all $\sigma\leq 0$ and all vector field multiindices $L$ satisfying $|L|\leq l$. Appealing to Sobolev embedding (for, potentially, infinitely
many different localisations $u_{a}$) yields the conclusion that there is a $u_{1}\in C^{\infty}(G)$ such that for every compact set $K\subset G$ and 
every $0\leq l\in\zo$, there is a constant $C_{K,l}$ such that 
\[
\|u_{\sigma}(\cdot,\sigma)-u_{1}\|_{}\leq C_{K,l}\ldr{\sigma}^{-\eta_{c}}
\]
for all $\sigma\leq 0$. Thus the first conclusion of the proposition holds. If $\eta_{c}>1$, this estimate can be integrated in order to yield
the second conclusion. 
\end{proof}

\section{Proofs I}\label{section:proofsI}

The purpose of the present section is to prove the statements made in Section~\ref{section:applications}. We begin with Proposition~\ref{prop:osc} 
and the statements made in Example~\ref{example:usigmaconvgeneralI}.

\subsection{Conditional results yielding convergence of the $\sigma$-derivative}\label{ssection:condresoscillatory}
In the present subsection, we prove Proposition~\ref{prop:osc} and the statements made in Example~\ref{example:usigmaconvgeneralI}.

\begin{proof}[Proposition~\ref{prop:osc}]
By assumption, the conditions of Theorem~\ref{thm:main} are satisfied. In particular, we are thus allowed to use the 
conclusions of that theorem, such as (\ref{eq:eKusigmaLtwoestimateintro}). Combining (\ref{eq:eKusigmaLtwoestimateintro}); localisations
as in the proof of Proposition~\ref{prop:asymptoticsexponconvofqtotwo}; and Sobolev embedding, it is clear that, given a solution $u$, a compact 
set $K\subseteq G$ and an $0\leq l\in\zo$, there is a constant $C_{K,l}$ such that 
\[
\|u(\cdot,\sigma)\|_{C^{l}(K)}\leq C_{K,l}\ldr{\sigma}
\]
for all $\sigma\leq 0$. Inserting this information into (\ref{eq:KleinGordonwrtsigmatime}), it is clear that 
\begin{equation}\label{eq:usigmasigmaCkestimate}
\|u_{\sigma\sigma}[\cdot,\sigma(\tau)]\|_{C^{l}(K)}\leq C_{K,l}\exp\left(2\int_{\tau}^{0}(q-2)d\tau'\right)\mff^{2}(\tau)\ldr{\sigma(\tau)},
\end{equation}
where we appealed to (\ref{eq:hvphzhasharpcond}) and the fact that if $X_{0}=X_{0}^{i}e_{i}$, then there is a constant
$C$ such that $|X^{i}_{0}|\leq C\|\ha^{-1}\|$, cf. (\ref{eq:Xzerodefinition}). Integrating 
(\ref{eq:usigmasigmaCkestimate}) from $\sigma_{0}$ to $\sigma_{1}\leq 0$ yields
\begin{equation}\label{eq:usigmadiffCkestimate}
\|u_{\sigma}(\cdot,\sigma_{1})-u_{\sigma}(\cdot,\sigma_{0})\|_{C^{l}(K)}
\leq C_{K,l}\int_{\tau_{0}}^{\tau_{1}}\exp\left(\int_{\tau}^{0}(q-2)d\tau'\right)\mff^{2}(\tau)\ldr{\sigma(\tau)}d\tau,
\end{equation}
where we used (\ref{eq:dsigmadtau}) and $\tau_{0}$, $\tau_{1}$ correspond to $\sigma_{0}$, $\sigma_{1}$ respectively. In order to 
proceed, note that (\ref{eq:dsigmadtau}) yields
\begin{equation}\label{eq:sigmatauform}
-\sigma(\tau)=\int_{\tau}^{0}C_{\sigma}\exp\left(-\int_{\tau'}^{0}(q-2)d\tau''\right)d\tau';
\end{equation}
recall that $\sigma(0)=0$ due to the requirements made in connection with (\ref{eq:dsigmadtdefrel}). Inserting this information into 
(\ref{eq:usigmadiffCkestimate}), it is clear that we need to estimate two integrals. First, we need to verify that 
\[
\int_{\tau_{0}}^{\tau_{1}}\exp\left(\int_{\tau}^{0}(q-2)d\tau'\right)\mff^{2}(\tau)d\tau
\]
is bounded as $\tau_{0}\rightarrow-\infty$. However, this is obvious due to the fact that $(q-2)_{+}\in L^{1}(-\infty,0]$ and the fact that 
(\ref{eq:mffsqmomoneintegrable}) holds; recall that under the present circumstances, Remark~\ref{remark:boundonqminustwo} applies. 
Next, we need to estimate
\begin{equation}\label{eq:secondintegraltobeestimated}
\int_{\tau_{0}}^{\tau_{1}}\exp\left(\int_{\tau}^{0}(q-2)d\tau'\right)\mff^{2}(\tau)\int_{\tau}^{0}\exp\left(-\int_{\tau'}^{0}(q-2)d\tau''\right)d\tau'
d\tau. 
\end{equation}
Note, to this end, that 
\begin{equation*}
\begin{split}
 & \exp\left(\int_{\tau}^{0}(q-2)d\tau'\right)\int_{\tau}^{0}\exp\left(-\int_{\tau'}^{0}(q-2)d\tau''\right)d\tau'\\
 = & \int_{\tau}^{0}\exp\left(\int_{\tau}^{\tau'}(q-2)d\tau''\right)d\tau'\leq C|\tau|
\end{split}
\end{equation*}
for all $\tau\leq 0$, where we used the fact that $(q-2)_{+}\in L^{1}(-\infty,0]$. Due to this observation, the expression
(\ref{eq:secondintegraltobeestimated}) can be estimated by 
\[
\int_{\tau_{0}}^{\tau_{1}}C\ldr{\tau}\mff^{2}(\tau)d\tau 
\]
which is finite as $\tau_{0}\rightarrow-\infty$; cf. (\ref{eq:mffsqmomoneintegrable}). To conclude, the expression on the right 
hand side of (\ref{eq:usigmadiffCkestimate}) converges to a finite number as $\tau_{0}\rightarrow-\infty$. Combining this 
observation with (\ref{eq:usigmadiffCkestimate}) yields the conclusion of the proposition. 
\end{proof}

Next, we justify the statements made in Example~\ref{example:usigmaconvgeneralI}. Before turning to the details, note the following. If 
$\lambda_{0},\kappa_{0}>0$ and $0<\kappa_{1}\leq 1$, then there is a constant $C_{\kappa}$ (depending 
only on $\lambda_{0}$, $\kappa_{0}$ and $\kappa_{1}$) such that 
\begin{equation}\label{eq:explintest}
\int_{\tau_{a}}^{\tau_{b}}\ldr{\tau}^{\kappa_{0}}e^{-\lambda_{0}\ldr{\tau}^{\kappa_{1}}}d\tau
\leq C_{\kappa}\ldr{\tau_{b}}^{\kappa_{0}+1-\kappa_{1}}e^{-\lambda_{0}\ldr{\tau_{b}}^{\kappa_{1}}}
\end{equation}
for all $\tau_{a}\leq\tau_{b}\leq -1$. We leave the verification of this statement to the reader. In what follows, we also appeal to the 
general observations concerning Bianchi class A developments made in Subsection~\ref{ssection:genobBianchiA} below. 

\begin{proof}[Example~\ref{example:usigmaconvgeneralI}]
We wish to apply Proposition~\ref{prop:osc}. To this end, we first need to verify that the conditions of Theorem~\ref{thm:main} are satisfied. 
Due to (\ref{eq:hasharpsubexpbd}), the monotone volume singularity $t_{-}$ is silent. Moreover, by assumption, $g$ solves (\ref{eq:Einsteinsequations}); 
$\rho\geq\bp$; $\rho\geq 0$; $\Lambda\geq 0$; and $\varphi_{0}\in C^{\infty}(M)$ only depends on $t$. 
Next, we have to verify that $\g_{\bS}\in L^{1}(-\infty,0]$. In the case of all Bianchi class A types but IX, $\g_{\bS}=0$, and there is nothing to 
prove. However, the Bianchi type IX case requires an argument. That $\g_{\bS}\in L^{1}(-\infty,0]$ can be verified directly without appealing to 
(\ref{eq:hasharpsubexpbd}). However, since the corresponding argument is somewhat more involved, we here rely on (\ref{eq:hasharpsubexpbd}).
In order to prove the integrability of $\g_{\bS}$, note that combining the arguments presented in Subsection~\ref{ssection:genobBianchiA} below
with \cite[(9), p.~414]{BianchiIXattr} yields
\begin{align}
\d_{\tau}\ha^{11} = & 2(q+2\Sigma_{+})\ha^{11},\ \ \d_{\tau}(\Nt\Nth)=2(q+2\Sp)\Nt\Nth,\label{eq:hauooevo}\\
\d_{\tau}\ha^{22} = & 2(q-\Sigma_{+}-\sqrt{3}\Sm)\ha^{22},\ \ \d_{\tau}(\No\Nth)=2(q-\Sigma_{+}-\sqrt{3}\Sm)\No\Nth,\label{eq:hauttevo}\\
\d_{\tau}\ha^{33} = & 2(q-\Sigma_{+}+\sqrt{3}\Sm)\ha^{33},\ \ \d_{\tau}(\No\Nt)=2(q-\Sigma_{+}+\sqrt{3}\Sm)\No\Nt.\label{eq:hauththevo}
\end{align}
Moreover, $\ha^{ij}=0$ if $i\neq j$. In particular, since the $N_{i}$ all have the same sign in the case of Bianchi type IX, there 
are constants $C_{a,i}>0$, $i=1,2,3$, such that 
\[
\ha^{11}=C_{a,1}\Nt\Nth,\ \ \
\ha^{22}=C_{a,1}\No\Nth,\ \ \
\ha^{33}=C_{a,1}\No\Nt.
\]
In particular, all the $N_{i}N_{j}$, $i\neq j$, are integrable due to (\ref{eq:hasharpsubexpbd}). On the other hand, 
(\ref{eq:scalarcurvaturerescaled}) below implies that 
\[
\g_{\bS}\leq N_{1}N_{2}+N_{2}N_{3}+N_{3}N_{1}.
\]
Combining these observations yields the conclusion that $\g_{\bS}\in L^{1}(-\infty,0]$.

In order to verify (\ref{eq:hvphzhasharpcond}), note that since $q\geq 0$ in the present setting, it is clear that $e^{\tau}\theta$ is bounded
from below by a positive constant for $\tau\leq 0$. Combining this observation with the fact that $\varphi_{0}$ is bounded, it is clear that 
$e^{-2\tau}|\hvph_{0}|$ is bounded for $\tau\leq 0$. In particular, there is thus a constant $C_{\rotot}>0$ such that if we let 
\begin{equation}\label{eq:mffdef}
\mff(\tau):=C_{\rotot}\exp(-\lambda_{0}\ldr{\tau}^{\a_{0}}), 
\end{equation}
then (\ref{eq:hvphzhasharpcond}) holds for all $\tau\leq 0$. Note also that $\mff$ has the properties stated at the beginning of 
Subsection~\ref{ssection:energiesintro}. To conclude, the conditions of Theorem~\ref{thm:main} are satisfied. 

Finally, in order to apply Proposition~\ref{prop:osc}, we need to verify that (\ref{eq:mffsqmomoneintegrable}) holds. However, this 
is an immediate consequence of (\ref{eq:mffdef}). In order to justify (\ref{eq:usigmauoasymptoticswithrate}), note that the proof of 
Proposition~\ref{prop:osc} yields the conclusion that 
\[
\|u_{\sigma}(\cdot,\sigma_{1})-u_{\sigma}(\cdot,\sigma_{0})\|_{C^{l}(K)}
\leq C_{K,l}\int_{\tau_{0}}^{\tau_{1}}C\ldr{\tau}\mff^{2}(\tau)d\tau 
\]
where $\tau_{0}$, $\tau_{1}$ correspond to $\sigma_{0}$, $\sigma_{1}$ respectively. Combining this estimate with (\ref{eq:explintest}) and
(\ref{eq:mffdef}) yields (\ref{eq:usigmauoasymptoticswithrate}).

Finally, in order to prove the statements in Example~\ref{example:usigmaconvgeneralI} concerning Bianchi type I, II, VI${}_{0}$ and 
VII${}_{0}$ vacuum solutions, it is sufficient to combine (\ref{eq:hauooevo})--(\ref{eq:hauththevo}) with the conclusions of 
\cite{minbok}. In order to justify this statement in greater detail, note that if $(\Sigma_{+},\Sigma_{-})$
converges to one of the points in (\ref{eq:speciallimitpoints}), then the corresponding monotone volume singularity is not silent; this is 
justified in Subsection~\ref{ssection:genobBianchiA} below. In particular, this situation is excluded by the assumptions. Due to 
\cite[Propositions~22.15, 22.16 and 22.18]{minbok} and \cite[Lemma~22.17, p.~240]{minbok}, the only possibility that remains is that the 
solution converges to a point on the Kasner circle different from the points in (\ref{eq:speciallimitpoints}). Due to the observations made
in connection with (\ref{eq:criteriaforsilence}), this implies that $\|\ha^{-1}\|$ converges to zero exponentially. In particular, 
(\ref{eq:hasharpsubexpbd}) holds with $\a_{0}=1$.
\end{proof}

\subsection{Conditional results yielding full asymptotics}\label{ssection:condresoscillatoryII}
The purpose of the present subsection is to prove Proposition~\ref{prop:fullasymptoticsBVIIIandIXgen} as well as the statements made in 
Example~\ref{example:fullasymptotics}. In what follows, we appeal to the general observations concerning Bianchi class A developments 
made in Subsection~\ref{ssection:genobBianchiA} below. 

\begin{proof}[Proposition~\ref{prop:fullasymptoticsBVIIIandIXgen}]
In analogy with the argument justifying (\ref{eq:usigmasigmaCkestimate}), it can be demonstrated that 
\begin{equation}\label{eq:usigmasigmaCkestimateII}
\|u_{\sigma\sigma}[\cdot,\sigma(\tau)]\|_{C^{l}(K)}\leq C_{K,l}\exp\left(2\int_{\tau}^{0}(q-2)d\tau'\right)[\|\ha^{-1}(\tau)\|+|\hvph_{0}(\tau)|]\ldr{\sigma(\tau)}.
\end{equation}
Integrating this estimate from $\sigma_{0}$ to $\sigma_{1}$, where $\sigma_{0}\leq\sigma_{1}\leq 0$ yields 
\begin{equation*}
\begin{split}
 & \|u_{\sigma}(\cdot,\sigma_{1})-u_{\sigma}(\cdot,\sigma_{0})\|_{C^{l}(K)}\\
 \leq & C_{K,l}\int_{\tau_{0}}^{\tau_{1}}\exp\left(\int_{\tau}^{0}(q-2)d\tau'\right)[\|\ha^{-1}(\tau)\|+|\hvph_{0}(\tau)|]\ldr{\sigma(\tau)}d\tau;
\end{split}
\end{equation*}
here $\tau_{i}$ corresponds to $\sigma_{i}$, $i=0,1$. Due to the proof of Proposition~\ref{prop:osc}, we obtain convergence as 
$\tau_{0}\rightarrow-\infty$. Moreover, 
\begin{equation*}
\begin{split}
 & \|u_{\sigma}(\cdot,\sigma_{1})-u_{1}\|_{C^{l}(K)}\\
 \leq & C_{K,l}\int_{-\infty}^{\tau_{1}}\exp\left(\int_{\tau}^{0}(q-2)d\tau'\right)[\|\ha^{-1}(\tau)\|+|\hvph_{0}(\tau)|]\ldr{\sigma(\tau)}d\tau.
\end{split}
\end{equation*}
Integrating this estimate from $\sigma_{a}$ to $\sigma_{b}$, where $\sigma_{a}\leq \sigma_{b}\leq 0$ yields
\begin{equation*}
\begin{split}
 & \|u(\cdot,\sigma_{b})-u(\cdot,\sigma_{a})-u_{1}(\sigma_{b}-\sigma_{a})\|_{C^{l}(K)}\\
 \leq & C_{K,l}\int_{\tau_{a}}^{\tau_{b}}
\int_{-\infty}^{\tau_{1}}\exp\left(\int_{\tau}^{0}(q-2)d\tau'\right)[\|\ha^{-1}(\tau)\|+|\hvph_{0}(\tau)|]\ldr{\sigma(\tau)}d\tau\\
 & \phantom{C_{K,l}\int_{\tau_{a}}^{\tau_{b}}\int_{-\infty}^{\tau_{1}}}\cdot\exp\left(-\int_{\tau_{1}}^{0}(q-2)d\tau'\right)d\tau_{1}\\
 \leq & C_{K,l}\int_{\tau_{a}}^{\tau_{b}}
\int_{-\infty}^{\tau_{1}}\exp\left(\int_{\tau}^{\tau_{1}}(q-2)d\tau'\right)[\|\ha^{-1}(\tau)\|+|\hvph_{0}(\tau)|]\ldr{\sigma(\tau)}d\tau d\tau_{1},
\end{split}
\end{equation*}
where $\tau_{a}$ and $\tau_{b}$ correspond to $\sigma_{a}$ and $\sigma_{b}$ respectively. Combining this estimate with the fact that 
$(q-2)_{+}\in L^{1}(-\infty,0]$ (note that the integrability of $(q-2)_{+}$ follows from Remark~\ref{remark:boundonqminustwo} and the fact that the conditions
of Theorem~\ref{thm:main} are fulfilled) yields
\begin{equation}\label{eq:Usbsaest}
\|U(\cdot,\sigma_{b})-U(\cdot,\sigma_{a})\|_{C^{l}(K)} \leq C_{K,l}\int_{\tau_{a}}^{\tau_{b}}
\int_{-\infty}^{\tau_{1}}[\|\ha^{-1}(\tau)\|+|\hvph_{0}(\tau)|]\ldr{\sigma(\tau)}d\tau d\tau_{1},
\end{equation}
where 
\[
U(\cdot,\sigma):=u(\cdot,\sigma)-u_{1}\cdot\sigma. 
\]
In order to estimate $\ldr{\sigma(\tau)}$, note that (\ref{eq:sigmatauform}) yields
\begin{equation*}
\begin{split}
|\sigma(\tau)| \leq & C_{\sigma}\int_{\tau}^{0}\exp\left(-\int_{\tau'}^{0}(q-2)d\tau''\right)d\tau'
 \leq C|\tau|\exp\left(-\int_{\tau}^{0}(q-2)d\tau'\right),
\end{split}
\end{equation*}
where we used the fact that $(q-2)_{+}\in L^{1}(-\infty,0]$ and the fact that $\sigma(0)=0$; cf. the requirement following (\ref{eq:dsigmadtdefrel}). 
To conclude, 
\[
\ldr{\sigma(\tau)}\leq C\ldr{\tau}\exp\left(-\int_{\tau}^{0}(q-2)d\tau'\right)
\]
for all $\tau\leq 0$. Inserting this information into (\ref{eq:Usbsaest}), it is clear that we need to estimate 
\begin{equation}\label{eq:integralyieldfullas}
\int_{\tau_{a}}^{\tau_{b}}
\int_{-\infty}^{\tau_{1}}[\|\ha^{-1}(\tau)\|+|\hvph_{0}(\tau)|]\ldr{\tau}\exp\left(-\int_{\tau}^{0}(q-2)d\tau'\right)d\tau d\tau_{1}.
\end{equation}
By assumption, this expression converges as $\tau_{a}\rightarrow-\infty$. In particular, $U(\cdot,\sigma)$ converges in $C^{l}(K)$ for every 
compact $K\subset G$ and every $0\leq l\in\zo$. Thus there is a function $u_{0}$ with the properties stated in the proposition.
\end{proof}

Next, we prove the statements made in Example~\ref{example:fullasymptotics}. 

\begin{proof}[Example~\ref{example:fullasymptotics}]
Due to the assumptions in Example~\ref{example:fullasymptotics}, we know that $\varphi_{0}$ is bounded. We begin by using this information 
in order to estimate the contribution of $\varphi_{0}$ to the integral (\ref{eq:integralyieldfullas}). Note, to this end, that
\[
|\hvph_{0}|\exp\left(-\int_{\tau}^{0}(q-2)d\tau'\right)=9\theta^{-2}|\varphi_{0}|\exp\left(-\int_{\tau}^{0}(q-2)d\tau'\right).
\]
Since $\theta$ satisfies (\ref{eq:thetaprimeitoq}), this equality implies
\begin{equation}\label{eq:hvphiprodest}
|\hvph_{0}|\exp\left(-\int_{\tau}^{0}(q-2)d\tau'\right)=9\theta^{-2}(0)|\varphi_{0}|\exp\left(-\int_{\tau}^{0}3qd\tau'\right).
\end{equation}
On the other hand, due to (\ref{eq:hauooevo})--(\ref{eq:hauththevo}), it is clear that 
\[
\d_{\tau}\det\ha^{-1}=6q\det\ha^{-1}.
\]
Combining this equality with (\ref{eq:hvphiprodest}) and the assumption that $\varphi_{0}$ is bounded yields 
\begin{equation}\label{eq:hvphiprodestfinal}
|\hvph_{0}|\exp\left(-\int_{\tau}^{0}(q-2)d\tau'\right)\leq C(\det\ha^{-1})^{1/2}\leq C\|\ha^{-1}\|^{3/2}.
\end{equation}
Due to the assumptions, the proof of 
Proposition~\ref{prop:fullasymptoticsBVIIIandIXgen}, (\ref{eq:hvphiprodestfinal}) and the above observations, it follows that 
\begin{equation}\label{eq:UsbsaestII}
\|U(\cdot,\sigma_{b})-U(\cdot,\sigma_{a})\|_{C^{l}(K)} \leq C_{K,l}\int_{\tau_{a}}^{\tau_{b}}
\int_{-\infty}^{\tau_{1}}\ldr{\tau}\|\ha^{-1}(\tau)\|\exp\left(-\int_{\tau}^{0}(q-2)d\tau'\right)d\tau d\tau_{1}
\end{equation}
for all $\sigma_{a}\leq\sigma_{b}\leq 0$. Due to (\ref{eq:hasharpsubexpbdint}) and two applications of (\ref{eq:explintest}), it follows that 
\[
\|U(\cdot,\sigma_{b})-U(\cdot,\sigma_{a})\|_{C^{l}(K)} \leq C_{K,l}\ldr{\tau_{b}}^{3-2\a_{0}}e^{-2\lambda_{0}\ldr{\tau_{b}}^{\a_{0}}}
\]
for all $\sigma_{a}\leq\sigma_{b}\leq 0$. Thus (\ref{eq:asymptoticsforucompas}) holds. 

Finally, we need to justify the statements in Example~\ref{example:fullasymptotics} concerning Bianchi type I, II, VI${}_{0}$ and 
VII${}_{0}$ vacuum solutions. However, the only additional complication in the present setting compared with the statements of 
Example~\ref{example:usigmaconvgeneralI} is that we need to estimate the second factor on the left hand side of (\ref{eq:hasharpsubexpbdint}).
Note, to this end, that due to the proof of Example~\ref{example:usigmaconvgeneralI}, we know that the Wainwright-Hsu variables converge 
to a point on the Kasner circle different from the points in (\ref{eq:speciallimitpoints}). Considering the equations for the $N_{i}$, note
that $N_{i}'=f_{i}(q,\Sigma_{+},\Sigma_{-})N_{i}$ for some function $f_{i}$; cf. \cite[(9), p.~414]{BianchiIXattr}. Since the solution converges
to the Kasner circle, $f_{i}(q,\Sigma_{+},\Sigma_{-})$ converges to a number. Moreover, the form of the $f_{i}$ immediately implies that, on the 
Kasner circle, the only way for $f_{i}$ to equal zero is if the relevant point on the Kasner circle is one of the points appearing in 
(\ref{eq:speciallimitpoints}). In other words, in the limit point, $f_{i}$ is either strictly positive or strictly negative. If $f_{i}$ is 
strictly negative in the limit and $N_{i}\neq 0$, then $|N_{i}|$ tends to infinity exponentially. This is not consistent with the fact that 
the solution converges to a point on the Kasner circle. If $N_{i}\neq 0$, the limit of $f_{i}$ must thus be positive. In other words, all the 
$N_{i}$ converge to zero exponentially. Combining this observation with the Wainwright-Hsu equations, cf. \cite[pp.~414--415]{BianchiIXattr},
it follows that $q-2$ converges to zero exponentially. The second factor on the left hand side of (\ref{eq:hasharpsubexpbdint}) is thus 
bounded. The statements of the example follow. 
\end{proof}

\section{Proofs II}\label{section:proofsII}

The purpose of the present section is to justify the statements made in Examples~\ref{example:asymptoticsstifffluid}, 
\ref{example:nonoscBclassAdev},  \ref{example:genericBianchiclassA} and \ref{example:nonexcBianchiclassB}. To do so, we need to prove that 
Proposition~\ref{prop:asymptoticsexponconvofqtotwo} applies in the situations considered in these examples. The proofs are based on 
results concerning the asymptotics of Bianchi class A and non-special Bianchi class B developments obtained in 
\cite{BianchiIXattr,RadermacherNonStiff,RadermacherStiff}. We make some general observations concerning 
the relevant developments in Subsections~\ref{ssection:genobBianchiA} and \ref{ssection:genobBianchiB} below. These observations follow from 
\cite{RadermacherNonStiff,BianchiIXattr} and form the basis of the analysis of the present section.
We begin by considering the stiff fluid case, which is defined by the condition that $\g=2$. 

\subsection{The stiff fluid case.}\label{ssection:proofsstifffluidcase}
In the case of orthogonal stiff fluids, there are results concerning the asymptotics of Bianchi class A solutions (cf., e.g., 
\cite{BianchiIXattr}) and the asymptotics of non-special Bianchi class B solutions (cf., e.g., \cite{RadermacherStiff}).
The present analysis is based on these references. However, since the details are somewhat different in the two cases, we 
treat them separately. 

\textit{Bianchi class A.} We need to demonstrate that Proposition~\ref{prop:asymptoticsexponconvofqtotwo} applies. To this
end, we combine \cite{BianchiIXattr} with the general observations collected in Subsection~\ref{ssection:genobBianchiA}.
The asymptotic behaviour is described in terms of the Wainwright-Hsu variables $\Sigma_{\pm}$, $N_{i}$, $i=1,2,3$, and $\Omega$;
cf. Subsection~\ref{ssection:genobBianchiA}. Due to \cite[Theorem~19.1, p. 478]{BianchiIXattr}, these variables converge
to a type I point. Say that the corresponding value of $(\Sigma_{+},\Sigma_{-})$ is $(s_{+},s_{-})$. Then, due to \cite[Theorem~19.1, p. 478]{BianchiIXattr},
$s_{+}^{2}+s_{-}^{2}<1$. Moreover, due to \cite[Theorem~19.1, p. 478]{BianchiIXattr}, there are additional restrictions for Bianchi types II, 
VI${}_{0}$, VII${}_{0}$, VIII and IX. In order to explain the restrictions, note that the $N_{i}$ satisfy $N_{i}'=f_{i}(q,\Sigma_{+},\Sigma_{-})N_{i}$ for some 
function $f_{i}$; cf. \cite[(9), p.~414]{BianchiIXattr}. The restrictions on $(s_{+},s_{-})$ appearing in \cite[Theorem~19.1, p.~478]{BianchiIXattr}
correspond to the following requirement: if the development is such that $N_{i}$ is non-zero, then $f_{i}(q,\Sigma_{+},\Sigma_{-})$ 
converges to a strictly positive number. In particular, the $N_{i}$ thus decay to zero exponentially. Combining this observation with the 
Hamiltonian constraint (cf. \cite[(11), p.~415]{BianchiIXattr}) yields the conclusion that $q-2$ converges to zero exponentially. 
Since $\bS/\theta^{2}$ is a quadratic polynomial in the $N_{i}$, cf. (\ref{eq:scalarcurvaturerescaled}) below, 
it is also clear that $\g_{\bS}$ decays exponentially. Turning to the restrictions on $\varphi_{0}$, note first that since $q-2$ converges to zero 
exponentially, $e^{3\tau}\theta(\tau)$ converges to a positive number as $\tau\rightarrow-\infty$; recall that $\theta'=-(1+q)\theta$. 
In particular there is a positive constant, say $C_{\theta}$, such that $e^{3\tau}\theta(\tau)\geq C_{\theta}^{-1}$ for all $\tau\leq 0$.
On the other hand, $\det a(\tau)=e^{6\tau}$ by the definition of $\tau$. Combining these observations with the assumption 
(\ref{eq:varphizeroestimate}) yields
\begin{equation}\label{eq:hvphqtotwoest}
|\hvph_{0}(\tau)|=9[\theta(\tau)]^{-2}|\varphi_{0}(\tau)| \leq 9C_{\theta}^{2}C_{\varphi}e^{6\tau}(e^{6\tau})^{-1+\eta_{\varphi}}
\leq 9C_{\theta}^{2}C_{\varphi}e^{6\eta_{\varphi}\tau}
\end{equation}
for all $\tau\leq 0$. Thus $\hvph_{0}$ converges to zero exponentially. Next, since $s_{+}^{2}+s_{-}^{2}<1$, the comments made in connection 
with (\ref{eq:criteriaforsilence}) below imply that $\|\ha^{-1}\|$ converges to zero exponentially. Thus Proposition~\ref{prop:asymptoticsexponconvofqtotwo} 
applies and we have verified the statements made in Example~\ref{example:asymptoticsstifffluid} for Bianchi class A developments. 

\textit{Bianchi class B.} Turning to the non-exceptional Bianchi class B developments, we again need to demonstrate that we can apply 
Proposition~\ref{prop:asymptoticsexponconvofqtotwo}. To this end, we combine \cite{RadermacherStiff} with the general observations 
collected in Subsection~\ref{ssection:genobBianchiB}. The 
asymptotic behaviour is described in terms of the variables introduced by Hewitt and Wainwright and denoted by $\Sigma_{+}$, $\tSi$, 
$\Delta$, $N_{+}$, $\tA$ and $\Omega$; cf. Subsection~\ref{ssection:genobBianchiB}. 

Due to \cite[Proposition~4.3, p.~9]{RadermacherStiff}, there are $s\in (-1,1)$ and $\ts\in [0,1)$ such that 
\[
\lim_{\tau\rightarrow-\infty}(\Sigma_{+},\tSi,\Delta,N_{+},\tA)(\tau)=(s,\ts,0,0,0).
\]
Moreover, $s^{2}+\ts<1$. Using the terminology of \cite{RadermacherStiff}, solutions thus converge to an element of the so-called \textit{Jacobs set} 
$\mJ$; cf. \cite[Definition~1.2, p.~4]{RadermacherStiff}. Next, \cite[(4), p.~3]{RadermacherStiff} yields
\[
q=2(1-\tA-\tN),
\]
where $\tN$ is given by (\ref{eq:tNdef}) below. 
In particular, it is thus clear that $q$ converges to $2$. Combining this observation with the fact that $s>-1$, it follows 
that $\tA$ converges to zero exponentially; cf., e.g., \cite[(3), p.~3]{RadermacherStiff}. Turning to $\Delta$ and $N_{+}$, note
that there are three possibilities as far as the limit as $\tau\rightarrow-\infty$ is concerned; a given solution converges to a 
point in one of the following sets:
\begin{align*}
\mJ_{-} := & \mJ\cap \{(1+\Sigma_{+})^{2}<3\tSi\},\\
\mJ_{0} := & \mJ\cap \{(1+\Sigma_{+})^{2}=3\tSi\},\\
\mJ_{+} := & \mJ\cap \{(1+\Sigma_{+})^{2}>3\tSi\}.
\end{align*}
If the limit point is in $\mJ_{-}\cup \mJ_{0}$, then $\tA$, $\Sigma_{+}$, $\tSi$, $\tN$, $q$, $\Delta$ and $N_{+}$ all converge 
exponentially towards their limiting values; cf. \cite[Proposition~5.1, pp.~10--11]{RadermacherStiff}. Assume now that the 
limiting point is in $\mJ_{+}$. Returning to \cite[(3), p.~3]{RadermacherStiff}, it is clear that 
\begin{equation}\label{eq:DeltaNpsys}
\left(\begin{array}{c} \Delta\\ N_{+}\end{array}\right)'=B\left(\begin{array}{c} \Delta\\ N_{+}\end{array}\right),
\end{equation}
where 
\begin{equation}\label{eq:limit}
\lim_{\tau\rightarrow-\infty}B(\tau)=\left(\begin{array}{cc} 2(s+1) & 2\ts \\ 6 & 2(s+1)\end{array}\right)=:B_{\infty}. 
\end{equation}
In particular, the real parts of the eigenvalues of $B_{\infty}$ are strictly positive if the limit point is in $\mJ_{+}$. 
Combining the above observations with, e.g., the results of \cite{IOP} yields the conclusion that $\Delta$ and 
$N_{+}$ converge to zero exponentially. To conclude, $\tA$, $\tN$, $N_{+}$, $\Delta$ and $q-2$ converge to zero exponentially. 
This implies that $\Sigma_{+}$ and $\tSi$ converge to their limits exponentially. Due to the observations made at the end 
of Subsection~\ref{ssection:genobBianchiB}, it is thus clear that $\|\ha^{-1}\|$ converges to zero exponentially. 
In the case of Bianchi class B, the scalar curvature of the spatial hypersurfaces of homogeneity is negative, so that 
$\g_{\bS}=0$; cf., e.g., \cite[Appendix~E]{stab}. Finally, in order to estimate $\hvph_{0}$, we can proceed as in the Bianchi class A case. Thus 
Proposition~\ref{prop:asymptoticsexponconvofqtotwo} applies and we have verified the statements made in 
Example~\ref{example:asymptoticsstifffluid} for non-exceptional Bianchi class B developments. 

\subsection{The non-stiff fluid case}\label{ssection:asymptoticsinthenonstifffluidcase}
Next, we consider Bianchi orthogonal perfect fluid developments with a linear equation of state $p=(\g-1)\rho$, where $\g<2$. In the case of 
Bianchi class A, we focus on vacuum and equations of state with $2/3<\g<2$. In the case of Bianchi class B, we focus on vacuum and equations 
of state with $0\leq \g<2/3$. In both cases, the restrictions arise from the results available in the literature. As before, we need to verify that
Proposition~\ref{prop:asymptoticsexponconvofqtotwo} applies. We begin by considering vacuum Bianchi class A. 

\textbf{Non-oscillatory Bianchi class A vacuum developments.} Recall the analysis concerning vacuum Bianchi type I, II, VI${}_{0}$ and VII${}_{0}$
developments presented in the proofs of Examples~\ref{example:usigmaconvgeneralI} and \ref{example:fullasymptotics}; cf. 
Subsections~\ref{ssection:condresoscillatory} and \ref{ssection:condresoscillatoryII}. Due to this analysis, $q-2$ and $\|\ha^{-1}\|$ converge 
to zero exponentially. Moreover, $\g_{\bS}=0$ for the Bianchi types of interest here. Finally, the assumptions of 
Example~\ref{example:nonoscBclassAdev} combined with an analysis similar to that presented in the stiff fluid setting, cf. 
(\ref{eq:hvphqtotwoest}), imply that $\hvph_{0}$ converges to zero exponentially. Thus Proposition~\ref{prop:asymptoticsexponconvofqtotwo} 
applies and the statements made in Example~\ref{example:nonoscBclassAdev} follow.

\textbf{Generic Bianchi type I, II and VII${}_{0}$ developments.} Next, we prove the statements made in Example~\ref{example:genericBianchiclassA}.
In the case of all the relevant Bianchi types, $\g_{\bS}=0$. In what follows, we therefore do not comment on the corresponding condition in
Proposition~\ref{prop:asymptoticsexponconvofqtotwo}. We consider the different Bianchi types separately and we start with generic Bianchi type I 
perfect fluid developments with $2/3<\g<2$.

\textit{Bianchi type I.} According to \cite[Proposition~8.1, p.~428]{BianchiIXattr}, $(\Sigma_{+},\Sigma_{-},\Omega)$ converges 
to $(s_{+},s_{-},0)$, with $s_{+}^{2}+s_{-}^{2}=1$ (note that the fixed point $F$ has been removed by assumption). If $(s_{+},s_{-})$ equals
one of the points appearing in (\ref{eq:speciallimitpoints}), then the monotone volume singularity is not silent; cf. 
Subsection~\ref{ssection:genobBianchiA}. Since this is incompatible with the assumptions, we can assume the limit $(s_{+},s_{-})$ to be 
different from the points appearing in (\ref{eq:speciallimitpoints}). By the arguments presented in connection 
with (\ref{eq:speciallimitpoints}), it follows that $\|\ha^{-1}\|$ converges to zero exponentially. Turning to $\Omega$, it converges to zero 
exponentially due to \cite[(9), p.~414]{BianchiIXattr} 
and the fact that $q\geq 2$ in the limit. Combining this observation with the constraint, \cite[(11), p.~415]{BianchiIXattr}, implies that 
$q-2$ converges to zero exponentially. Finally, that $\hvph_{0}$ converges to zero exponentially follows from (\ref{eq:hvphqtotwoest}). 
Thus Proposition~\ref{prop:asymptoticsexponconvofqtotwo} applies and the statements made in Example~\ref{example:genericBianchiclassA} 
follow in the case of Bianchi type I.

\textit{Bianchi type II.} In the case of Bianchi type II, we can, without loss of generality, assume that $N_{1}>0$, $N_{2}=0$
and $N_{3}=0$. In this case, $\Sigma_{-}$ is either always zero or never zero (since the conditions $\Sigma_{-}=0$ and $N_{2}=N_{3}$ 
define an invariant set). Let us first consider the case that $\Sigma_{-}=0$. According to 
\cite[Proposition~9.1, p.~428]{BianchiIXattr}, there are then three possibilities: the solution converges to $F$; the solution
equals $P_{1}^{+}(II)$; or $(\Omega,\Sigma_{+},N_{1})$ converges to $(0,-1,0)$. The first two possibilities are incompatible with the 
assumption of genericity. The third possibility is incompatible with the requirement of silence; cf. 
Subsection~\ref{ssection:genobBianchiA}. Assume that $\Sigma_{-}\neq 0$. Then, according to 
\cite[Proposition~9.1, p.~428]{BianchiIXattr}, the solution converges to a point in $\mK_{2}\cup\mK_{3}$. Here, the 
$\mK_{i}$ are given by \cite[Definition~6.1, p.~421]{BianchiIXattr}. In particular, $(\Omega,\Sigma_{+},\Sigma_{-},N_{1})$
converges to $(0,s_{+},s_{-},0)$, where $s_{+}^{2}+s_{-}^{2}=1$ and $(s_{+},s_{-})$ is different from the points 
appearing in (\ref{eq:speciallimitpoints}). By the arguments presented in connection with (\ref{eq:speciallimitpoints}), it
then follows that $\|\ha^{-1}\|$ converges to zero exponentially.  Since $q$ converges to $2$, it is clear that $\Omega$
converges to zero exponentially; cf. \cite[(9), p.~414]{BianchiIXattr}. Since the solution converges to a point in 
$\mK_{2}\cup\mK_{3}$, $s_{+}<1/2$, so that $N_{1}$ converges to zero exponentially. Combining these observations with the 
constraint, \cite[(11), p.~415]{BianchiIXattr}, yields the conclusion that $q-2$ converges to zero exponentially. 
Finally, that $\hvph_{0}$ converges to zero exponentially follows from (\ref{eq:hvphqtotwoest}). 
Thus Proposition~\ref{prop:asymptoticsexponconvofqtotwo} applies and the statements made in Example~\ref{example:genericBianchiclassA} 
follow in the case of Bianchi type II. 

\textit{Bianchi type VI${}_{0}$.} Since we lack an appropriate reference on the asymptotic behaviour of Bianchi type VI${}_{0}$ non-vacuum
orthogonal perfect fluid developments with $\g<2$, we are not in a position to analyse the asymptotic behaviour of solutions to the Klein-Gordon 
equation on such backgrounds. 

\textit{Bianchi type VII${}_{0}$, the locally rotationally symmetric case.} In the case of Bianchi type VII${}_{0}$, we can, without 
loss of generality, assume that $N_{1}=0$, $N_{2}>0$ and $N_{3}>0$. It is of interest to first consider the locally rotationally
symmetric subcase; i.e., to assume that $N_{2}=N_{3}$ and $\Sigma_{-}=0$. Note that in this case, the constraint, \cite[(11), p.~415]{BianchiIXattr}, 
reads
\begin{equation}\label{eq:BianchiVIIzeroconstraintTaubcase}
\Omega+\Sigma_{+}^{2}=1.
\end{equation}
Due to \cite[Proposition~10.1, p.~430]{BianchiIXattr}, there are three possibilities. Either the solution converges to $\Sigma_{+}=1$ on
the Kasner circle; the solution converges to $F$; or $\Sigma_{+}$ converges to $-1$. Convergence to $F$ has been excluded by the condition
of genericity. If $\Sigma_{+}$ converges to $-1$, then the monotone volume singularity is not silent; cf. the comments made at the end of 
Subsection~\ref{ssection:genobBianchiA}. Again, this case is thus excluded by the assumptions. Let us therefore assume that $\Sigma_{+}$ 
converges to $1$. Then $q$ converges to $2$, so that $\Omega$ converges to zero exponentially due to \cite[(9), p.~414]{BianchiIXattr}. 
Combining this observation with (\ref{eq:BianchiVIIzeroconstraintTaubcase}) yields the conclusion that $\Sigma_{+}$ converges to $1$ 
exponentially. In particular, $q$ converges to $2$ exponentially. Since $q$ converges to $2$ and $(\Sigma_{+},\Sigma_{-})$ converges to 
$(1,0)$, the observations made in connection with (\ref{eq:speciallimitpoints}) imply that $\|\ha^{-1}\|$ converges to zero 
exponentially. Finally, that $\hvph_{0}$ converges to zero exponentially follows from (\ref{eq:hvphqtotwoest}). 
Thus Proposition~\ref{prop:asymptoticsexponconvofqtotwo} applies and the statements made in Example~\ref{example:genericBianchiclassA} 
follow in the case of locally rotationally symmetric Bianchi type VII${}_{0}$ developments. 

\textit{Bianchi type VII${}_{0}$, the general case.} Due to \cite[Proposition~10.2, p.~431]{BianchiIXattr}, generic and non-locally rotationally
symmetric Bianchi type VII${}_{0}$ solutions (with $N_{1}=0$, $N_{2}>0$ and $N_{3}>0$) converge to a point in $\mK_{1}$. Recalling that $\mK_{1}$ 
is the subset of the Kasner circle with $\Sigma_{+}>1/2$, it is clear that $N_{2}$, $N_{3}$ and $\Omega$ converge to zero exponentially, so that 
we can argue as in the locally rotationally symmetric setting; cf. \cite[(9), p.~414]{BianchiIXattr}. This completes the proof of the 
statements made in Example~\ref{example:genericBianchiclassA}. 

\textbf{Bianchi class B.}
In the case of Bianchi class B, we are mainly interested in solutions converging to the \textit{Kasner parabola}; cf. 
\cite[Definition~1.15, p.~8]{RadermacherNonStiff}. The Kasner parabola is defined to be the set of points with 
\[
(\Omega,\Delta,\tA,N_{+})=(0,0,0,0),
\]
so that $\Sigma_{+}^{2}+\tSi=1$; cf. (\ref{eq:HamConBianchiBRescaled}) and (\ref{eq:tNdef}) below. Moreover, its elements are fixed points. 
In addition to the Kasner parabola, there are some solutions that converge to a plane wave equilibrium point; cf. 
\cite[Definition~1.17, p.~8]{RadermacherNonStiff}. There are also some solutions that converge to the fixed point, say $F_{B}$, characterised by 
\[
(\Sigma_{+},\tSi,\Delta,\tA,N_{+})=(0,0,0,0,0).
\]
Next, the points Taub 1 and Taub 2 (denoted T1 and T2) are of special importance. Here T1 and T2 are defined by the conditions
\[
(\Sigma_{+},\tSi,\Delta,\tA,N_{+})=(-1,0,0,0,0),\ \ \ 
(\Sigma_{+},\tSi,\Delta,\tA,N_{+})=(1/2,3/4,0,0,0)
\]
respectively. 

Due to the fact that the results of \cite{RadermacherNonStiff} are restricted to either vacuum or $0\leq \g<2/3$, we also
restrict our attention to these two cases. Combining \cite[Proposition~4.2, pp.~19--20]{RadermacherNonStiff},
\cite[Proposition~4.4, p.~21]{RadermacherNonStiff}, \cite[Proposition~5.1, p.~23]{RadermacherNonStiff} and 
\cite[Proposition~6.1, p.~26]{RadermacherNonStiff}, there are the following possibilities as far as the asymptotics are concerned:
\begin{itemize}
\item The solution coincides with the fixed point T1.
\item The solution coincides with the fixed point $F_{B}$.
\item The solution converges to a plane wave equilibrium point. 
\item The solution converges to a point on the Kasner parabola different from T1. 
\end{itemize}
If the solution coincides with the fixed point T1, it is clear that $2+2\Sigma_{+}=0$ for the entire solution. In particular, the 
function $\hf^{1}_{1}$ introduced in Subsection~\ref{ssection:genobBianchiB} is constant; cf. (\ref{eq:hfootaudersv}) and 
\cite[(6), p.~7]{RadermacherNonStiff}. Thus the corresponding monotone volume singularity is not silent. By arguments presented in connection with 
(\ref{eq:hfootaudersv}), a solution that converges to a plane wave equilibrium point is such that the corresponding monotone volume 
singularity is not silent. Next, let us assume that the solution is the fixed point $F_{B}$. This can only happen in the non-vaccum
setting, and then $q<0$, $\Sigma_{+}=0$ and $\tSi=0$; cf. \cite[(6), p.~7]{RadermacherNonStiff}. In particular, $\hf^{1}_{1}$ and the matrix 
with components $\hf^{A}_{B}$ are
unbounded as $\tau\rightarrow-\infty$. Thus, the corresponding monotone volume singularity is not silent. What remains is to consider
a solution converging to a point on the Kasner parabola different from T1. 

\textit{Convergence to a point strictly between T1 and T2.} Assuming the solution to converge to T2 or to a point on the Kasner parabola to 
the left of T2, \cite[Proposition~6.2, p.~26]{RadermacherNonStiff} applies and yields the conclusion that $(\Sigma_{+},\tSi)$ converges 
exponentially to its limiting value, say $(s,\ts)$. Moreover, $\Omega$, $\tA$, $N_{+}$, $\Delta$, $\tN$ and $q-2$ converge to zero exponentially. 
Assume now that $-1<s<1/2$. Then, by observations made in Subsection~\ref{ssection:genobBianchiB} (cf., in particular, the conclusions following 
(\ref{eq:ineqcharBVImone})), it follows that $\|\ha^{-1}\|$ converges to zero exponentially. 
Next, note that $\g_{\bS}=0$ for all Bianchi class B developments. Finally, that $\hvph_{0}$ converges to zero exponentially follows from 
(\ref{eq:hvphqtotwoest}). Thus Proposition~\ref{prop:asymptoticsexponconvofqtotwo} applies and the statements made in 
Example~\ref{example:nonexcBianchiclassB} follow in the case of convergence to a point on the Kasner parabola strictly between T1 and T2.

\textit{Convergence to T2.} Assume the solution to converge to T2. Then $\Omega$, $\tA$, $N_{+}$, $\Delta$, $\tN$ and $q-2$ converge to 
zero exponentially, as before. Thus, due to the observations made at the end of Subsection~\ref{ssection:genobBianchiB} below, the solution 
is a locally rotationally symmetric Bianchi type VI${}_{-1}$ solution, a case which is excluded by the assumptions. 

\textit{Convergence to a point to the right of T2.} Assume that the Wainwright-Hsu variables of the solution converge to a point to the 
right of T2. Then $(\Sigma_{+},\tSi)$ converges to, say, $(s,\ts)$, where $s>1/2$. Moreover, as in the stiff fluid case, $\Delta$ and 
$N_{+}$ satisfy the equation (\ref{eq:DeltaNpsys}), where $B$ satisfies (\ref{eq:limit}). Since, $s>1/2$, the eigenvalues of the right 
hand side of (\ref{eq:limit}) both have positive real part. Thus, due to, e.g., the results of \cite{IOP}, the functions $\Delta$ and 
$N_{+}$ converge to zero exponentially. Next, since $q$ converges to $2$ and $\Sigma_{+}$ converges to $s$, it follows from 
\cite[(5), p.~6]{RadermacherNonStiff} that $\tA$ converges to zero exponentially. In addition, \cite[(11), p.~7]{RadermacherNonStiff};
the fact that $q$ converges to $2$; and the fact that $\g<2$ imply that $\Omega$ converges to zero exponentially. Finally, since 
$\tA$ and $N_{+}$ converge to zero exponentially, it follows that $\tN$ converges to zero exponentially; cf. 
\cite[(7), p.~7]{RadermacherNonStiff}. Combining these observations with \cite[(15), p.~16]{RadermacherNonStiff} yields the conclusion that 
$q-2$ converges to zero exponentially. Next, by observations made in Subsection~\ref{ssection:genobBianchiB} (cf., in particular, the 
conclusions following (\ref{eq:ineqcharBVImone})), it follows that $\|\ha^{-1}\|$ converges to zero exponentially. Moreover, $\g_{\bS}=0$. 
Finally, that $\hvph_{0}$ converges to zero exponentially follows from (\ref{eq:hvphqtotwoest}). Thus 
Proposition~\ref{prop:asymptoticsexponconvofqtotwo} applies and the statements made in Example~\ref{example:nonexcBianchiclassB} follow in 
the case of convergence to a point on the Kasner parabola strictly to the right of T2.

\section{Proofs III}\label{section:proofsIII}

The purpose of the present section is to prove the statements made in Example~\ref{example:nongenBclassAdev}. We need to consider two cases; 
either that the Wainwright-Hsu variables of the solution converge to $P_{i}^{+}(II)$ for some $i=1,2,3$; or that they converge to the fixed point
$F$. 

\textit{Convergence to $P_{i}^{+}(II)$.}
Consider the fixed points $P_{i}^{+}(II)$, $i=1,2,3$. Since the different points are related by a symmetry (cf. \cite[p.~415]{BianchiIXattr}), 
it is sufficient to focus on $P_{1}^{+}(II)$; cf. Definition~\ref{def:fixedpoints}. Consider a solution converging to $P_{1}^{+}(II)$. Then $q$ 
converges to $q_{\infty}=(3\g-2)/2$, $\Sigma_{+}$ converges to $s_{+}=(3\g-2)/8$, $\Sigma_{-}$ converges to zero and $\Omega$ converges to 
$\Omega_{0}:=1-(3\g-2)/16$. Since $N_{2}$ and $N_{3}$ converge to zero and $N_{1}$ converges to a strictly positive number, it is clear from 
(\ref{eq:scalarcurvaturerescaled}) that $\bS$ eventually becomes negative. This means that there is a $T$ such that $\g_{\bS}(\tau)=0$ for 
$\tau\leq T$. Next, note that $q\rightarrow q_{\infty}<2$ (since $\g<2$). Moreover, since $\Omega'=2(q-q_{\infty})\Omega$, 
cf. \cite[(9), p.~414]{BianchiIXattr}, and $\Omega$ converges to a strictly positive number, it is clear that 
\[
\int_{\tau}^{0}[q(\tau')-q_{\infty}]d\tau
\]
converges as $\tau\rightarrow-\infty$. In particular, 
\begin{equation}\label{eq:intqest}
\int_{\tau}^{0}q(\tau')d\tau=\int_{\tau}^{0}[q(\tau')-q_{\infty}]d\tau-q_{\infty}\tau=-q_{\infty}\tau+O(1). 
\end{equation}
Next, note that \cite[(9), p.~414]{BianchiIXattr}, combined with the fact that $N_{1}$ converges to a strictly positive number implies that 
\[
\int_{\tau}^{0}[q(\tau')-4\Sigma_{+}(\tau')]d\tau'
\]
converges as $\tau\rightarrow-\infty$. Thus 
\[
\int_{\tau}^{0}\Sigma_{+}(\tau')d\tau'=\int_{\tau}^{0}\left[\Sigma_{+}(\tau')-\frac{1}{4}q(\tau')\right]d\tau'
+\frac{1}{4}\int_{\tau}^{0}[q(\tau')-q_{\infty}]d\tau'-\frac{1}{4}q_{\infty}\tau.
\]
Thus
\begin{equation}\label{eq:intsigmaplus}
\int_{\tau}^{0}\Sigma_{+}(\tau')d\tau'=-\frac{1}{4}q_{\infty}\tau+O(1).
\end{equation}
Next, note that $N_{2}$ and $N_{3}$ converge to zero exponentially; this follows from \cite[(9), p.~414]{BianchiIXattr} and the assumed 
convergence to $P_{1}(II)$. Thus $S_{-}$ introduced on \cite[p.~415]{BianchiIXattr} converges to 
zero exponentially. Combining this information with \cite[(9), p.~414]{BianchiIXattr}; the fact that $q$ converges to $q_{\infty}<2$; and 
the fact that $\Sigma_{-}$ converges to zero implies that $\Sigma_{-}$ converges to zero exponentially. Due to (\ref{eq:intqest}); 
(\ref{eq:intsigmaplus}); the exponential convergence of $\Sigma_{-}$ to zero; and (\ref{eq:hauooevo})--(\ref{eq:hauththevo}), it is 
clear that 
\[
|\ha^{11}(\tau)|\leq Ce^{3q_{\infty}\tau},\ \ \ |\ha^{22}(\tau)|\leq Ce^{3q_{\infty}\tau/2},\ \ \ |\ha^{33}(\tau)|\leq Ce^{3q_{\infty}\tau/2}
\]
for all $\tau\leq 0$. In particular, $\|\ha^{-1}(\tau)\|\leq Ce^{3q_{\infty}\tau/2}$ for all $\tau\leq 0$. Next, note that 
\[
\theta^{-2}(\tau)=\theta^{-2}(0)\exp\left(-\int_{\tau}^{0}2[1+q(\tau')]d\tau'\right)\leq Ce^{2(1+q_{\infty})\tau}
\]
for all $\tau\leq 0$. Since $\varphi_{0}$ is bounded, we conclude that $|\hvph_{0}(\tau)|\leq Ce^{2(1+q_{\infty})\tau}$ for all $\tau\leq 0$. 
Summarising, it is clear that Proposition~\ref{prop:asymptoticsexponconvofqtoqinfdifffromtwo} applies with $\eta_{0}=3q_{\infty}/2$.
The statements in Example~\ref{example:nongenBclassAdev} concerning solutions converging to one of the $P_{i}(II)$, $i=1,2,3$, follow. 

\textit{Convergence to $F$.}
Recall the critical point $F$ introduced in Definition~\ref{def:fixedpoints}. Consider a development such that the corresponding 
Wainwright-Hsu variables converge to $F$. Then $q$ converges to $q_{\infty}:=(3\g-2)/2$; cf. the formula at the 
bottom of \cite[p.~414]{BianchiIXattr}. Note that $q_{\infty}>0$ since $\g>2/3$. Since $\Sigma_{+}$ and $\Sigma_{-}$ converge to zero, 
\cite[(9), p.~414]{BianchiIXattr} implies that the $N_{i}$ converge to zero exponentially. In particular, the $S_{\pm}$ introduced on 
\cite[p.~415]{BianchiIXattr} converge to zero exponentially. The only way for the above to be consistent with 
\cite[(9), p.~414]{BianchiIXattr} is that the $\Sigma_{\pm}$ converge to zero exponentially; note that $q_{\infty}<2$ since $\g<2$. Combining 
this observation with the constraint, 
\cite[(11), p.~415]{BianchiIXattr}, yields the conclusion that $\Omega$ converges exponentially to $1$ and $q$ converges exponentially to 
$q_{\infty}$. Returning to \cite[(9), p.~414]{BianchiIXattr} with this information in mind yields the conclusion that there is a constant $C$ such that 
\begin{equation}\label{eq:NietcestFconv}
|N_{i}(\tau)|^{2}+|\Sigma_{+}(\tau)|+|\Sigma_{-}(\tau)|+|\Omega(\tau)-1|+|q(\tau)-q_{\infty}|\leq Ce^{2q_{\infty}\tau}
\end{equation}
for all $\tau\leq 0$. Next, since $\bS/\theta^{2}$ is given by (\ref{eq:scalarcurvaturerescaled}), it is clear that $\g_{\bS}\leq Ce^{2q_{\infty}\tau}$ 
for all $\tau\leq 0$. Moreover, (\ref{eq:hauooevo})--(\ref{eq:hauththevo}) imply that $\|\ha^{-1}\|\leq Ce^{2q_{\infty}\tau}$ 
for all $\tau\leq 0$. Finally, 
\[
|\hvph_{0}(\tau)|\leq 9\theta^{-2}(\tau)|\varphi_{0}(\tau)|\leq Ce^{2(q_{\infty}+1)\tau}
\]
for all $\tau\leq 0$. Due to the above, it is clear that Proposition~\ref{prop:asymptoticsexponconvofqtoqinfdifffromtwo} applies with 
$\eta_{0}=2q_{\infty}$. The statements in Example~\ref{example:nongenBclassAdev} concerning solutions converging to $F$ follow.

\section{Appendix}

\subsection{General observations, Bianchi class A developments}\label{ssection:genobBianchiA}
In this subsection, we collect some general observations concerning Bianchi class A orthogonal perfect fluid developments. 
The main source of the statements made here is \cite{BianchiIXattr}.

In the case of Bianchi class A orthogonal perfect fluid developments, the matrix $a$ appearing in (\ref{eq:Bianchimetricdef}) is 
diagonal; cf. \cite[Lemma~21.2, p.~488]{BianchiIXattr}. Moreover, the Wainwright-Hsu variables $\Sigma_{\pm}$, $N_{i}$, $i=1,2,3$, 
and $\Omega$ introduced in \cite{waihsu89} can be used to analyse the asymptotics of solutions. These variables are also introduced 
in \cite[p.~487]{BianchiIXattr}; see also \cite[p.~233]{minbok} for a presentation in the vacuum setting. Note that the $\theta$ 
appearing in \cite{BianchiIXattr} is the mean curvature of the spatial hypersurfaces
of homogeneity, so that it coincides with the $\theta$ appearing in the present paper. Since the time coordinate $\tau$ appearing 
in \cite{BianchiIXattr} satisfies \cite[(137), p.~487]{BianchiIXattr} and the $\tau$ appearing in the present paper satisfies the same
relation (cf. the proof of Lemma~\ref{lemma:hgintro}), the two $\tau$'s can be assumed to coincide. Finally, since $\theta'=-(1+q)\theta$
both in the present paper and in \cite{BianchiIXattr}, it is clear that the $q$ appearing in \cite{BianchiIXattr} coincides with the 
$q$ appearing in the present paper. 

Next, consider $\ha_{ij}$. Due to (\ref{eq:dtauhamjfinal}), it can be calculated that 
\begin{equation}\label{eq:dtauhaijbianchia}
\d_{\tau}\ha^{ij}=2\left(q\de^{i}_{l}-3\Sigma^{i}_{\phantom{i}l}\right)\ha^{lj}.
\end{equation}
In the cases of interest here, $\Sigma^{i}_{\phantom{i}l}=0$ if $i\neq l$ and $\ha^{ij}=0$ if $i\neq j$; cf. 
\cite[Lemma~21.2, p.~488]{BianchiIXattr}. However, when comparing the present paper with \cite{BianchiIXattr} some care is required
when referring to indices. The reason for this is that we use a fixed frame (which is not orthonormal) in the present paper, whereas
the indices appearing in \cite{BianchiIXattr} refer to an orthonormal frame. On the other hand, the orthonormal frame appearing in 
\cite{BianchiIXattr} is of the form $\be_{0}=\d_{t}$ and $\be_{i}=a_{i}e_{i}$ (no summation on $i$), where $\{e_{\a}\}$ is the frame 
appearing in the present paper; cf. \cite[Lemma~21.2, p.~488]{BianchiIXattr}. Moreover, the $a_{i}$ are strictly positive functions 
of time only. In particular, the $\Sigma_{22}$ appearing in \cite{BianchiIXattr} is different from the $\Sigma_{22}$ appearing in the 
present paper. On the other hand, $\Sigma^{2}_{\phantom{2}2}$ and $\Sigma^{3}_{\phantom{3}3}$ of the present paper coincide with $\Sigma_{22}$ 
and $\Sigma_{33}$ of \cite{BianchiIXattr}. Using the terminology of the present paper, the $\Sigma_{\pm}$ variables introduced in 
\cite[(138), p.~487]{BianchiIXattr} thus satisfy 
\[
\Sigma_{+}=\frac{3}{2}(\Sigma^{2}_{\phantom{2}2}+\Sigma^{3}_{\phantom{3}3}),\ \ \
\Sigma_{-}=\frac{\sqrt{3}}{2}(\Sigma^{2}_{\phantom{2}2}-\Sigma^{3}_{\phantom{3}3}).
\]
Assume now that $q$ converges to a limit, say $q_{\infty}$, and that $\Sigma_{\pm}(\tau)\rightarrow s_{\pm}$ as $\tau\rightarrow-\infty$. 
Then 
\begin{equation}\label{eq:criteriaforsilence}
(q-3\Sigma^{1}_{\phantom{1}1},q-3\Sigma^{2}_{\phantom{2}2},q-3\Sigma^{3}_{\phantom{3}3})\rightarrow
(q_{\infty}+2s_{+},q_{\infty}-s_{+}-\sqrt{3}s_{-},q_{\infty}-s_{+}+\sqrt{3}s_{-}).
\end{equation}
If all the components of the vector on the right hand side are strictly positive, it then follows from (\ref{eq:dtauhaijbianchia})
that $\|\ha^{-1}\|$ converges to zero exponentially. We are mainly interested in the case that $q_{\infty}=2$ and 
$s_{+}^{2}+s_{-}^{2}\leq 1$. In that setting, the only way that one of the components of the vector on the right hand side of 
(\ref{eq:criteriaforsilence}) can be $\leq 0$ is that one of the following conditions are satisfied:
\begin{equation}\label{eq:speciallimitpoints}
(s_{+},s_{-})=(-1,0),\ \ \
(s_{+},s_{-})=(1/2,-\sqrt{3}/2),\ \ \
(s_{+},s_{-})=(1/2,\sqrt{3}/2).
\end{equation}
Next, it is of interest to note that 
\begin{equation}\label{eq:scalarcurvaturerescaled}
\frac{\bS}{\theta^{2}}=-\frac{1}{2}(N_{1}^{2}+N_{2}^{2}+N_{3}^{2})+N_{1}N_{2}+N_{2}N_{3}+N_{3}N_{1};
\end{equation}
cf., e.g., \cite[Lemma~19.11, p.~209]{minbok}. This expression can only be strictly positive in the case of Bianchi type 
IX. Finally, note that if $q_{\infty}>-1$ and $\varphi_{0}$ is bounded, then $\hvph_{0}$ converges to zero exponentially; this 
is due to the fact that $\theta'=-(1+q)\theta$. 

\textit{The non-stiff and non-silent Bianchi class A setting.} Considering (\ref{eq:speciallimitpoints}) and the adjacent text, it 
is clear that the points in (\ref{eq:speciallimitpoints}) play a special role. Assume that $(\Sigma_{+},\Sigma_{-})$ converges to one 
of these points. By applying the symmetries of the equations 
(cf. \cite[p.~415]{BianchiIXattr} and \cite{waihsu89}), we can assume $(\Sigma_{+},\Sigma_{-})$ to converge to $(-1,0)$. 
Due to \cite[Proposition~3.1, p.~416]{BianchiIXattr}, it then follows that the corresponding solution is contained in the 
invariant set characterised by $\Sigma_{-}=0$ and $N_{2}=N_{3}$ (assuming $2/3<\g<2$). Since $q$ is bounded from below by $2\Sigma_{+}^{2}$, it is
clear that $q$, in the limit, is bounded from below by $2$. Combining this information with \cite[(9), p.~414]{BianchiIXattr}
implies that $\Omega$ converges to zero exponentially (in the non-stiff fluid setting). Thus $q$ converges to $2$. Combining this information with 
\cite[(9), p.~414]{BianchiIXattr} again yields the conclusion that $N_{1}$, $N_{1}N_{2}$ and $N_{1}N_{3}$ converge to zero 
exponentially. On the other hand, the constraint \cite[(11), p.~415]{BianchiIXattr} implies that 
\[
\Omega+\Sigma_{+}^{2}+\frac{3}{4}(N_{1}^{2}-4N_{1}N_{2})=1.
\]
Combining this equality with previous observations implies that $\Sigma_{+}+1$ converges to zero exponentially and that $q-2$ converges to 
zero exponentially. In particular, $q+2\Sigma_{+}$ is integrable. On the other hand, 
\[
\d_{\tau}\ha^{11}=2(q-3\Sigma^{1}_{\phantom{1}1})\ha^{11}=2(q+2\Sigma_{+})\ha^{11}. 
\] 
In particular, $\ha^{11}$ thus converges to a strictly positive number as $\tau\rightarrow-\infty$. Consequently, the singularity is not silent;
cf. (\ref{eq:intainvhainv}). 

\subsection{General observations,  Bianchi class B solutions}\label{ssection:genobBianchiB}
In the case of non-exceptional Bianchi class B solutions, there are variables introduced by Hewitt and 
Wainwright, cf. \cite{hewandwain}, that can be used to describe the asymptotics; cf. also \cite{RadermacherNonStiff}. The relevant 
variables are denoted $\Sigma_{+}$, $\tSi$, $\Delta$, $N_{+}$, $\tA$ and $\Omega$. For future reference, it is of interest to note
that the Hamiltonian constraint is equivalent to 
\begin{equation}\label{eq:HamConBianchiBRescaled}
\Omega+\Sigma_{+}^{2}+\tSi+\tA+\tN=1,
\end{equation}
where
\begin{equation}\label{eq:tNdef}
\tN:=\frac{1}{3}(N_{+}^{2}-\kappa \tA)
\end{equation}
and $\kappa$ is a parameter associated with the Lie group. In order see why this is the case, note that on 
\cite[p.~65]{RadermacherNonStiff}, it is pointed out that the Hamiltonian constraint is equivalent to 
\cite[(56), p.~60]{RadermacherNonStiff}. By a simple change of variables, the latter equality is equivalent to 
\cite[(65), p.~61]{RadermacherNonStiff}. Normalising according to \cite[(67)--(68), p.~62]{RadermacherNonStiff}
reproduces (\ref{eq:HamConBianchiBRescaled}). Turning to the time coordinates, note that the $\theta$ appearing in 
\cite{RadermacherNonStiff} is the mean curvature of $G_{t}$. Thus the $\theta$ appearing in \cite{RadermacherNonStiff}
coincides with the $\theta$ appearing in the present paper. Since the time coordinate $\tau$ appearing in 
\cite{RadermacherNonStiff} is related to proper time according to \cite[(69), p.~62]{RadermacherNonStiff}, it is 
clear that the $\tau$ appearing in \cite{RadermacherNonStiff} can be chosen to coincide with the time coordinate 
$\tau$ appearing in the present paper. Finally, due to \cite[(70), p.~62]{RadermacherNonStiff}, it is clear that the 
deceleration parameter $q$ appearing in \cite{RadermacherNonStiff} coincides with the $q$ appearing in the present 
paper; $\theta'=-(1+q)\theta$ both in \cite{RadermacherNonStiff} and in the present paper, so that, since $\tau$ and $\theta$
conicide, the two $q$'s coincide. To conclude: the $\tau$, $\theta$ and $q$ appearing in \cite{RadermacherNonStiff} are 
the same as the objects $\tau$, $\theta$ and $q$ appearing in the present paper. Note also that, in the case of Bianchi 
class B, the scalar curvature of the spatial hypersurfaces of homogeneity is negative, so that $\g_{\bS}=0$. 

\textit{The causal structure.} In the applications, we need to estimate $\|\ha^{-1}\|$. Note, to this end, that the frame 
used to describe the metric in \cite[Subsection~11.6, pp.~67--72]{RadermacherNonStiff} is orthonormal and denoted by 
$\{e_{\a}\}$, where $e_{0}=\d_{t}$. However, it is constructed using a basis $e_{0}$, $\tilde{e}_{i}$, $i=1,2,3$, where 
$\tilde{e}_{i}$, $i=1,2,3$, is a basis of $\mfg$. Due to \cite[pp.~69--70]{RadermacherNonStiff}, 
\[
e_{A}=f_{A}^{B}\tilde{e}_{B},\ \ \
e_{1}=f_{1}^{1}\tilde{e}_{1}
\]
for $A,B\in \{2,3\}$ (in what follows, capital Latin indices range from $2$ to $3$), where the $f_{A}^{B}$ and $f_{1}^{1}$ 
satisfy the following initial value problems:
\begin{align*}
\d_{t}f_{1}^{1} = & \left(\bsigma_{A}^{\phantom{A}A}-\frac{1}{3}\theta\right)f_{1}^{1},\ \ \ f_{1}^{1}(0)=1,\\
\d_{t}f_{A}^{C} = & \left(-\bsigma_{A}^{\phantom{A}B}-\frac{1}{3}\theta\de_{A}^{B}+\Omega_{1}\e_{A}^{\phantom{A}B}\right)f_{B}^{C},\ \ \ 
f_{A}^{C}(0)=\de_{A}^{C}.
\end{align*}
Moreover, $\bsigma_{A}^{\phantom{A}B}$ denotes the shear; cf. $\bsigma_{ij}$ introduced in (\ref{eq:thetaandsigmaijdef}).
Here $\Omega_{1}$ is a gauge quantity which can be chosen freely; in what follows, we choose it to equal $0$. Let $h_{1}^{1}$ and 
$h_{A}^{B}$ be such that
\[
h_{1}^{1}f_{1}^{1}=1,\ \ \ h_{A}^{B}f_{B}^{C}=\de_{A}^{C}. 
\]
Next, let $\{\txi^{i}\}$ be the dual basis of $\{\te_{i}\}$ and define
\[
\xi^{1}:=h_{1}^{1}\txi^{1},\ \ \
\xi^{A}:=h^{A}_{B}\txi^{B},
\]
$A=2,3$. Then $\{\xi^{i}\}$ is the dual basis of $\{e_{i}\}$. In particular, the spacetime metric can be written
\[
g=-dt\otimes dt+\textstyle{\sum}_{i}\xi^{i}\otimes \xi^{i}=-dt\otimes dt+(h_{1}^{1})^{2}\txi^{1}\otimes \txi^{1}
+\sum_{A}h^{A}_{B}h^{A}_{C}\txi^{B}\otimes \txi^{C}.
\]
In particular, writing $g$ in the form (\ref{eq:Bianchimetricdef}) (with $\xi^{i}$ replaced by $\txi^{i}$), the matrix with 
components $a_{ij}$ is block diagonal (with $a_{1A}=a_{A1}=0$), and 
\[
a^{11}=(f^{1}_{1})^{2},\ \ \
a^{AB}=\textstyle{\sum}_{C}f^{A}_{C}f^{B}_{C}.
\]
In particular, 
\[
\ha^{11}=9\theta^{-2}(f^{1}_{1})^{2},\ \ \
\ha^{AB}=9\theta^{-2}\textstyle{\sum}_{C}f^{A}_{C}f^{B}_{C}.
\]
It is therefore of interest to introduce
\[
\hf^{1}_{1}:=3\theta^{-1}f^{1}_{1},\ \ \
\hf^{A}_{B}:=3\theta^{-1}f^{A}_{B}
\]
and to calculate
\begin{equation}\label{eq:hfootauder}
\d_{\tau}\hf^{1}_{1}=(1+q)\hf^{1}_{1}+(3\Sigma_{A}^{\phantom{A}A}-1)\hf^{1}_{1}=(q+3\Sigma_{A}^{\phantom{A}A})\hf^{1}_{1},
\end{equation}
where
\[
\Sigma_{A}^{\phantom{A}B}:=\theta^{-1}\bsigma_{A}^{\phantom{A}B}.
\]
At this stage it is important to note that we here use the index conventions of \cite{RadermacherNonStiff} (as opposed to those of 
the present paper). In particular, the indices refer to an orthonormal frame, so that they are raised and lowered with the 
Kronecker delta. Similarly, 
\begin{equation}\label{eq:hfABtauder}
\d_{\tau}\hf^{A}_{B}=(q\de^{C}_{B}-3\Sigma_{B}^{\phantom{B}C})\hf^{A}_{C}.
\end{equation}
Before proceeding, it is of interest to note that $\Sigma_{A}^{\phantom{A}A}=2\Sigma_{+}/3$; this is a consequence of 
\cite[(52), p.~59]{RadermacherNonStiff} and \cite[(67), p.~62]{RadermacherNonStiff}. Thus
\begin{equation}\label{eq:hfootaudersv}
\d_{\tau}\hf^{1}_{1}=(q+2\Sigma_{+})\hf^{1}_{1}.
\end{equation}
Note that $\tA^{1/2}$ satisfies the same equation, cf. \cite[(5), p.~6]{RadermacherNonStiff}, and that $\tA>0$ in Bianchi class B. 
Thus $\hf^{1}_{1}$ is a strictly positive multiple of $\tA^{1/2}$. In particular, $\hf^{1}_{1}$ converges to zero exponentially if and only if 
$\tA$ converges to zero exponentially. 

\textit{The plane wave solutions.} In the case of Bianchi class B, the so-called \textit{plane wave equilibrium points} are 
equilbrium points characterised by \cite[Definition~1.17, p.~8]{RadermacherNonStiff}. In particular, a solution that converges 
to a plane wave equilibrium point is such that the limit of $\tA$ is strictly positive. Combining this fact with the observations
made following (\ref{eq:hfootaudersv}), it is clear that the corresponding monotone volume singularity is not silent.

\textit{Convergence of the normalised shear.} Before proceeding, it is of interest to verify that $\Sigma_{B}^{\phantom{B}C}$
converges. Note that, due to the constructions described in \cite[Subsection~11.6, pp.~67--72]{RadermacherNonStiff}, 
the equation \cite[(54), p.~59]{RadermacherNonStiff} is satisfied. Moreover, the $\bsigma_{AB}$ and $n_{AB}$ can be retrived
from solutions to these equations using \cite[(52), p.~59]{RadermacherNonStiff}. If $\tsi_{AB}$ is defined by 
\cite[(52), p.~59]{RadermacherNonStiff}, let $\tSi_{AB}:=\tsi_{AB}/\theta$. Then \cite[(54), p.~59]{RadermacherNonStiff}
implies that 
\[
\d_{\tau}\tSi_{AB}=\frac{3}{\theta^{2}}\d_{t}\tsi_{AB}+(1+q)\tSi_{AB}
=(q-2)\tSi_{AB}-2N_{+}\tN_{AB}\pm 2\tA^{1/2}{}^{*}\tN_{AB},
\]
where $\tN_{AB}=\tin_{AB}/\theta$, ${}^{*}\tN_{AB}={}^{*}\tin_{AB}/\theta$, $N_{+}=n_{+}/\theta$ and we use the notation introduced
in \cite[(52)--(53), p.~59]{RadermacherNonStiff}. Note that if $N_{+}$ and $\tA$ converge to zero exponentially, then 
$\tN=3\tN^{AB}\tN_{AB}/2$ converges to zero exponentially; cf. (\ref{eq:tNdef}), \cite[(58), p.~60]{RadermacherNonStiff} and 
\cite[(67), p.~62]{RadermacherNonStiff}. Since ${}^{*}\tN_{AB}=\tN_{A}^{\phantom{A}C}\e_{CB}$,
cf. \cite[(53), p.~59]{RadermacherNonStiff}, it follows that ${}^{*}\tN_{AB}$ converges to zero exponentially. Assuming 
additionally that $q-2$ converges to zero exponentially, it is clear that $\tSi_{AB}$ converges exponentially. On the other hand, 
due to \cite[(52), p.~59]{RadermacherNonStiff} and \cite[(67), p.~62]{RadermacherNonStiff}, $\Sigma_{AB}=\tSi_{AB}+\Sigma_{+}\de_{AB}/3$.
Assuming, in addition to the above, that $\Sigma_{+}$ converges exponentially, it follows that $\Sigma_{AB}$ converges exponentially.

\textit{Convergent asymptotics, $q\rightarrow 2$.} Before discussing convergence, note that the state space associated with the 
Hewitt-Wainwright variables is compact; cf. \cite[(7)--(9), p.~7]{RadermacherNonStiff}. Assume now that $N_{+}$, 
$\tA$ and $q-2$ converge to zero exponentially. Then $\tN$ converges to zero exponentially due to \cite[(7), p.~7]{RadermacherNonStiff}. 
Combining these observations with \cite[(5), p.~6]{RadermacherNonStiff} yields the conclusion that $(\Sigma_{+},\tSi)$ converges exponentially
to, say, $(s,\ts)$. Due to the exponential convergence of the variables and the above discussion of the convergence of $\Sigma_{AB}$, it is 
clear that $\Sigma_{AB}$ converges exponentially to a limit, say $S_{AB}$. Next, due to \cite[(11), p.~7]{RadermacherNonStiff}, it is also
clear that $\Omega$ converges exponentially to $0$ if $\g<2$ and that if $\g=2$ and $\Omega$ is initially positive, then $\Omega$ converges 
exponentially to a strictly positive number. In either case, $\Omega$ converges exponentially to a limit, say $\Omega_{0}$. Finally, due to
\cite[(8), p.~7]{RadermacherNonStiff}, $\Delta$ converges to zero exponentially. Considering 
(\ref{eq:hfootaudersv}), it is clear that if $s>-1$, then $\hf^{1}_{1}$ converges to zero exponentially. 
Turning to (\ref{eq:hfABtauder}), it is of interest to calculate the eigenvalues of $3S_{B}^{\phantom{B}C}/2$ in terms of $s$ and $\ts$. Say that 
the eigenvalues of the matrix with components $3S_{B}^{\phantom{B}C}/2$ are $\lambda_{\pm}$. Then
\begin{equation}\label{eq:spluslambdatworel}
s=\lim_{\tau\rightarrow-\infty}\Sigma_{+}(\tau)=\frac{3}{2}\lim_{\tau\rightarrow-\infty}\Sigma_{A}^{\phantom{A}A}(\tau)
=\lambda_{+}+\lambda_{-}.
\end{equation}
It can also be computed that 
\[
\Sigma^{AB}\Sigma_{AB}=\tSi^{AB}\tSi_{AB}+\frac{2}{9}\Sigma_{+}^{2}=\frac{2}{3}\tSi+\frac{2}{9}\Sigma_{+}^{2},
\]
where we use the notation introduced in \cite[Subsection~11.2, pp.~57--62]{RadermacherNonStiff}. In particular, 
\begin{equation}\label{eq:lambdatwosqrel}
\lambda_{+}^{2}+\lambda_{-}^{2}=\frac{9}{4}\lim_{\tau\rightarrow-\infty}\Sigma^{AB}\Sigma_{AB}=
\frac{3}{2}\ts+\frac{1}{2}s^{2}.
\end{equation}
Combining (\ref{eq:spluslambdatworel}) and (\ref{eq:lambdatwosqrel}) yields 
\begin{equation}\label{eq:lambdapmform}
\lambda_{\pm}=\frac{s}{2}\pm\frac{\sqrt{3}\ts^{1/2}}{2}.
\end{equation}
Moreover, taking the limit of the Hamiltonian constraint yields
\begin{equation}\label{eq:limitofhamiltoniancon}
\Omega_{0}+\ts+s^{2}=1.
\end{equation}
Note that $\lambda_{+}\geq \lambda_{-}$. Assuming $\lambda_{+}\geq 1$ yields the conclusion that 
\begin{equation}\label{eq:ineqcharBVImone}
\left(s-\frac{1}{2}\right)^{2}+\left(\ts^{1/2}-\frac{\sqrt{3}}{2}\right)^{2}=s^{2}+\ts+1-2\lambda_{+}\leq-\Omega_{0},
\end{equation}
where we appealed to (\ref{eq:lambdapmform}) and (\ref{eq:limitofhamiltoniancon}). To conclude, the only way for an eigenvalue of the matrix with 
components $3S_{B}^{\phantom{B}C}/2$ to be $\geq 1$ is if $s=1/2$, $\ts=3/4$ and $\Omega_{0}=0$. To summarise: if $N_{+}$, $\tA$ and $q-2$ converge to 
zero exponentially; $s>-1$; and $(s,\ts,\Omega_{0})\neq(1/2,3/4,0)$, then $\|\ha^{-1}\|$ converges to zero 
exponentially as $\tau\rightarrow-\infty$.

\textit{Bianchi VI${}_{-1}$.} Due to the above, what remains to be considered is the case that $\g\in [0,2)$, $s=1/2$ and $\ts=3/4$. However, due to 
\cite[Section~7, pp.~40-41]{RadermacherNonStiff}, convergence to the corresponding point on the Kasner parabola implies that the solution is 
a locally rotationally symmetric Bianchi type VI${}_{-1}$ solution.

\subsection{Blow up criteria}\label{ssection:blowupcriteria}
The purpose of the present subsection is to justify the statements made in Section~\ref{section:blowupintro}. In order to do so, 
we apply the results of \cite[Chapter~8]{finallinsys} to (\ref{eq:nonflatKasnerKleinGordon}). We therefore need to verify that the 
conditions of \cite[Proposition~8.1, pp.~91--92]{finallinsys} and \cite[Proposition~8.9, p.~93]{finallinsys} are satisfied. 

\textit{Applicability of the results of \cite{finallinsys}.} To begin with, we need to verify that \cite[Definition~7.8, p.~82]{finallinsys} 
is satisfied. This verification partially overlaps with \cite[Example~4.20, pp.~42--43]{finallinsys}. However, for the benefit of the reader, we 
provide a complete justification here. The metric associated with (\ref{eq:nonflatKasnerKleinGordon}) is given by 
\begin{equation}\label{eq:gconnonflatKasner}
g_{\rocon}=-d\tau\otimes d\tau+\textstyle{\sum}_{r=1}^{R}e^{-2\bbe_{r}\tau}\bge_{r}
\end{equation}
on $M_{\rocon}$, where $M_{\rocon}$ is defined in connection with (\ref{eq:nonflatKasnerKleinGordon}). Here the $\bbe_{r}$ are distinct and defined so that 
the set of $\bbe_{r}$'s equals the set of $\b_{i}$'s. We also order
the $\bbe_{r}$'s so that $\bbe_{1}<\cdots<\bbe_{R}<0$. Let $d_{r}$ equal the number of $\b_{i}$'s such that $\b_{i}=\bbe_{r}$. Then $\bge_{r}$ is 
the standard metric on $\tn{d_{r}}$. We denote by $\bge$ and $\bk$ the metric and second fundamental form induced on 
$\bS_{\tau}:=\tn{d}\times\{\tau\}$ by $g_{\rocon}$. Comparing (\ref{eq:nonflatKasnerKleinGordon}) with \cite[(1.2), p.~4]{finallinsys}, it is clear 
that $g^{00}=-1$; $d=0$; $a_{r}=e^{-\bbe_{r}\tau}$; $\a=0$; $\zeta(\tau)=m^{2}e^{-2\tau}$; and $f=0$. Moreover, 
\begin{equation}\label{eq:bkKasnerform}
\bk=-\textstyle{\sum}_{r=1}^{R}\bbe_{r}e^{-2\bbe_{r}\tau}\bge_{r}.
\end{equation}
Letting $U:=\d_{\tau}$, it is clear that $U$ is the future directed unit normal to the hypersurfaces $\bS_{\tau}$ (with respect
to the metric $g_{\rocon}$). Moreover,
\begin{equation}\label{eq:mlUbkKasnerform}
\ml_{U}\bk=\textstyle{\sum}_{r=1}^{R}2\bbe_{r}^{2}e^{-2\bbe_{r}\tau}\bge_{r}.
\end{equation}
For future reference, it is of interest to note that 
\[
\bk\geq -\bbe_{R}\bge.
\]
The first step in verifying that \cite[Definition~7.8, p.~82]{finallinsys} is satisfied is to verify that $(M_{\rocon},g_{\rocon})$ is a 
canonical separable
cosmological manifold in the sense of \cite[Definition~1.18, p.~11]{finallinsys}. However, this follows immediately from 
(\ref{eq:gconnonflatKasner}). 
Next, we need to verify that (\ref{eq:nonflatKasnerKleinGordon}) is $C^{2}$-balanced. Note, to this end, that $\d_{\tau}$ is future 
uniformly timelike in the sense of \cite[Definition~3.1, p.~32]{finallinsys}. Moreover, it is clear that $\a$ and $\zeta$ are $C^{1}$-future
bounded in the sense of \cite[Definition~3.8, p.~34]{finallinsys}. Since $d=0$, the conditions on $\mcX$ and the shift vector field are void. 
Finally, in order to verify that the second fundamental form is $C^{1}$-bounded, it is sufficient to recall that $\bk$ and $\ml_{U}\bk$ 
satisfy (\ref{eq:bkKasnerform}) and (\ref{eq:mlUbkKasnerform}); cf. \cite[Definition~3.4, p.~32]{finallinsys}.  Due to 
\cite[Definition~3.8, p.~34]{finallinsys}, it is thus clear that (\ref{eq:nonflatKasnerKleinGordon}) is $C^{2}$-balanced. 

The condition that the shift vector field be negligible (cf. \cite[Definition~7.3, p.~81]{finallinsys}) is void since $d=0$. Next, note that 
$\bk$ is diagonally convergent with $\b_{\roRi,r}=\bbe_{r}$, where $\kappa_{\rod}$ can be 
choosen to be as large as desired; cf. \cite[Definition~7.3, p.~81]{finallinsys}. The condition that $\bge$ be $C^{2}$-asymptotically 
diagonal (cf. \cite[Definition~7.3, p.~81]{finallinsys}) is void since $d=0$ in our case. 

Considering \cite[Definition~7.5, p.~82]{finallinsys}, it is clear that the main coefficients are convergent with $\a_{\infty}=0$; $\zeta_{\infty}=0$;
$C_{\romn}=|m|^{2}$; and $\kappa_{\romn}:=2$ (note that since $d=0$, the conditions on $X^{j}$ are void). Finally, considering 
\cite[Definition~7.6, p.~82]{finallinsys}, it is clear that the equation is asymptotically non-degenerate with $Q=R$. To conclude, 
(\ref{eq:nonflatKasnerKleinGordon}) is geometrically non-degenerate, diagonally dominated, balanced and convergent; i.e., definition
\cite[Definition~7.8, p.~82]{finallinsys} is satisfied. It is also clear that $f=0$ and $\bbe_{Q}<0$. In particular, it follows that 
\cite[Propositions~8.1 and 8.9, pp.~91--93]{finallinsys} are applicable to (\ref{eq:nonflatKasnerKleinGordon}). 

\textit{Continuity properties of the asymptotic maps.} As a preparation to stating the conclusions, note that 
\begin{equation}\label{eq:AinfinityKasner}
A_{\infty}:=\left(\begin{array}{cc} 0 & 1 \\ 0 & 0\end{array}\right);
\end{equation}
cf. the statements of \cite[Propositions~8.1 and 8.9, pp.~91--93]{finallinsys}. As a consequence, the $\kappa_{\rosil,+}$ appearing in the 
statements of the 
propositions vanishes; note that $\kappa_{\rosil,+}$ is defined to be the largest real part of an eigenvalue of $A_{\infty}$; cf. the statements of 
the propositions and \cite[Definition~4.3, p.~38]{finallinsys}. Moreover, in our case, $\b_{\rem}=-\bbe_{R}$. Due to the form of $A_{\infty}$, 
it is clear 
that if $0<\b\leq\b_{\rem}$, then the first generalised eigenspace in the $\b,A_{\infty}$-decomposition of $\cn{2}$, say $E_{a}$, equals $\cn{2}$, 
irrespective of the choice of $\b$; cf. \cite[Definition~4.7, pp.~39--40]{finallinsys}. Next, note that $\kappa_{q,\pm}$ is given by 
\cite[(7.24), p.~83]{finallinsys},
so that $\kappa_{q,\pm}=\pm\bbe_{q}/2$ in our case. In particular, $s_{\roh,\b,+}=0$ if $0<\b<-\bbe_{q}/2$ for all $q\in\{1,\dots,R\}$; cf. 
\cite[(8.3), p.~91]{finallinsys}. Fixing such a $\b$, as well as an $\e>0$, there is a constant $C_{\e,\b}$, depending only on $\e$, $\b$, 
the spectra of the $\bge_{r}$ and the coefficients of the equation; and a numerical constant $N$ such that for every smooth solution $u$ to 
(\ref{eq:nonflatKasnerKleinGordon}), there are smooth functions $v_{\infty}$ and $u_{\infty}$ such that 
\begin{equation}\label{eq:convergencetoasymptotics}
\begin{split}
 & \left\|\left(\begin{array}{c} u(\cdot,\tau) \\ u_{\tau}(\cdot,\tau)\end{array}\right)
-\left(\begin{array}{c} v_{\infty}\tau+u_{\infty} \\ v_{\infty}\end{array}\right)\right\|_{(s)}\\
 \leq &  C_{\e,\b}\ldr{\tau}^{N}e^{-\b\tau}\left[\|u_{\tau}(\cdot,0)\|_{(s+\e)}+\|u(\cdot,0)\|_{(s+1+\e)}\right]
\end{split}
\end{equation}
for all $\tau\geq 0$ and all $s\in\ro$. 

Next, fixing $0<\e<1$ and choosing $\b=-\e\bbe_{R}/2$, it is clear that $s_{\roh,\b}=-1/2+\e/2$; cf. \cite[(8.6), p.~92]{finallinsys}. Applying
\cite[Proposition~8.1, pp.~91--92]{finallinsys}  with this $\b$ and with $\e$ replaced by $\e/2$ then yields a constant $C_{\e}$ such that 
for every smooth solution $u$ to (\ref{eq:nonflatKasnerKleinGordon}), the corresponding $v_{\infty}$ and $u_{\infty}$ satisfy
\[
\|u_{\infty}\|_{(s)}+\|v_{\infty}\|_{(s)}\leq C_{\e}\left[\|u_{\tau}(\cdot,0)\|_{(s-1/2+\e)}+\|u(\cdot,0)\|_{(s+1/2+\e)}\right]
\]
for all $s\in\ro$. Letting $\Psi_{\infty}$ be the map introduced in Section~\ref{section:blowupintro}, it is thus clear that 
$\Psi_{\infty}$ extends to a bounded linear map as in (\ref{eq:Psiinfext}). Next, note that $s_{\roh,-}=1/2$; cf. 
\cite[(8.8), p.~93]{finallinsys} and the above. Appealing to \cite[Proposition~8.9, p.~93]{finallinsys}, in particular 
\cite[(8.7), p.~93]{finallinsys}
then yields the conclusion that $\Phi_{\infty}$ extends to a bounded linear map as in (\ref{eq:Phiinfext}).

\textit{Criteria guaranteeing $L^{2}$-blow up.} Using the continuity of $\Psi_{\infty,0,\e}$, it is clear that 
\[
\ma_{\e}:=\Psi_{\infty,0,\e}^{-1}[(L^{2}(\tn{n})-\{0\})\times L^{2}(\tn{n})]\cap C^{\infty}(\tn{n})\times C^{\infty}(\tn{n})
\]
is a subset of $C^{\infty}(\tn{n})\times C^{\infty}(\tn{n})$ which is open with respect to the $H^{(1/2+\e)}(\tn{n})\times H^{(-1/2+\e)}(\tn{n})$-topology. 
Due to the homeomorphism between initial data and asymptotic data (in the $C^{\infty}$-topology) and the linearity of the map $\Psi_{\infty}$, it is 
straightforward to verify that $\ma_{\e}$ is dense with respect to the $C^{\infty}$-topology. To conclude, the subset of smooth initial data such that 
$\|v_{\infty}\|_{2}>0$ is open with respect to the $H^{(1/2+\e)}(\tn{n})\times H^{(-1/2+\e)}(\tn{n})$-topology and dense with respect to the 
$C^{\infty}$-topology. Returning to (\ref{eq:convergencetoasymptotics}), it is also clear that if $u$ is a solution corresponding to initial data 
in $\ma_{\e}$, then $\|u(\cdot,\tau)\|_{L^{2}}\rightarrow\infty$ as $\tau\rightarrow\infty$. 

\textit{Criteria guaranteeing $C^{1}$-blow up.}
Next, note that choosing $s_{0}:=n/2+1+\e/2$ yields a continuous map
\[
\Psi_{\infty,s_{0},\e/2}:H^{((n+3)/2+\e)}(\tn{n})\times H^{((n+1)/2+\e)}(\tn{n})\rightarrow H^{(s_{0})}(\tn{n})\times H^{(s_{0})}(\tn{n}).
\]
On the other hand, due to Sobolev embedding, $H^{(s_{0})}(\tn{n})$ embeds continuously into $C^{1}(\tn{n})$. Denote the corresponding map $\Upsilon$. 
Next, define $\mc_{(s_{0})}$ and $\mc_{1}$ as follows. First, $\varphi\in\mc_{1}$ if and only if $\varphi\in C^{1}(\tn{n})$ and $d\varphi(\bx)\neq 0$ for 
every $\bx\in\tn{n}$ such that $\varphi(\bx)=0$. Second, let $\mc_{(s_{0})}:=\mc_{1}\cap H^{(s_{0})}(\tn{n})$. Note that $\Upsilon^{-1}(\mc_{1})=\mc_{(s_{0})}$. 
Moreover, $\mc_{1}$ is an open subset of $C^{1}(\tn{n})$; 
we leave the verification of this statement to the reader. Thus $\mc_{(s_{0})}$ is an open subset of $H^{(s_{0})}(\tn{n})$. Define
\[
\mb_{\e}:=\Psi_{\infty,s_{0},\e/2}^{-1}[\mc_{(s_{0})}\times H^{(s_{0})}(\tn{n})]\cap C^{\infty}(\tn{n})\times C^{\infty}(\tn{n}).
\]
Then $\mb_{\e}$ is open with respect to the $H^{((n+3)/2+\e)}(\tn{n})\times H^{((n+1)/2+\e)}(\tn{n})$-topology. In order to verify that $\mb_{\e}$ is dense, 
note that there is a smooth pair of functions $\psi=(\psi_{1},\psi_{0})$ on $\tn{n}$ such that $\Psi_{\infty,s_{0},\e/2}(\psi)=(1,0)$. Assume that 
$\phi\in C^{\infty}(\tn{n})\times C^{\infty}(\tn{n})$. Letting $\varphi=(\varphi_{1},\varphi_{2})$ be the image of $\phi$ under $\Psi_{\infty,s_{0},\e/2}$ 
and $\de\in\ro$, it is clear that 
\[
\Psi_{\infty,s_{0},\e/2}(\phi+\de\psi)=\varphi+\de\cdot (1,0).
\] 
Since the measure of the set of irregular values of a smooth function is zero, it is clear that the set of $\de$ such that $\varphi_{1}+\de\in\mc_{(s_{0})}$
has full measure; this is due to Sard's theorem; cf., e.g., \cite[Theorem~6.10, p.~129]{lee}. In particular, it is thus clear that $\phi$ is in the 
closure of $\mb_{\e}$. To conclude, $\mb_{\e}$ is open with respect to the $H^{((n+3)/2+\e)}(\tn{n})\times H^{((n+1)/2+\e)}(\tn{n})$-topology and dense with 
respect to the $C^{\infty}$-topology. The demonstration of the remaining statements in Section~\ref{section:blowupintro} is left to the reader.

\section*{Acknowledgments}

The author would like to acknowledge the support of the G\"{o}ran Gustafsson Foundation for Research in Natural Sciences and 
Medicine. This research was funded by the Swedish research council, dnr. 2017-03863.

\end{document}